\documentclass[11pt]{amsart}
\usepackage{amsmath}
\usepackage{amssymb}
\usepackage{amsthm}
\usepackage{mathrsfs}
\usepackage{comment}
\usepackage{hyperref}

\usepackage[all,cmtip]{xy}\usepackage{xcolor}
\usepackage{enumerate}
\usepackage{bm}
\usepackage{dsfont}
\usepackage{mathtools}
\usepackage[cal=euler]{mathalfa}
\usepackage[top=1in, bottom=1.25in, left=1.25in, right=1.25in]{geometry}
\usepackage{parskip}

\hypersetup{colorlinks=true,linkcolor=magenta,citecolor=blue}

\DeclareMathAlphabet{\mathpzc}{OT1}{pzc}{m}{it}

\theoremstyle{plain}
\newtheorem{theorem}{Theorem}[section]
\newtheorem*{theorem*}{Theorem}

\newtheorem{lemma}[theorem]{Lemma}
\newtheorem*{claim*}{Claim}
\newtheorem{proposition}[theorem]{Proposition}

\newtheorem{corollary}[theorem]{Corollary}

\theoremstyle{definition}
\newtheorem{definition}[theorem]{Definition}

\newtheorem{example}[theorem]{Example}

\newtheorem{remark}[theorem]{Remark}

\numberwithin{equation}{section}
\numberwithin{figure}{section}

\newcommand{\F}{\mathbb{F}}

\newcommand{\Gal}{\mathrm{Gal}}

\newcommand{\Aut}{\mathrm{Aut}}

\newcommand{\ignore}[1]{}

\newcommand{\can}{\overline{\phantom{x}}}



\begin{document}
\setlength{\parindent}{0pt}
\title[Skew polycyclic codes up to isometry and equivalence]{Using nonassociative algebras to classify skew polycyclic codes up to isometry and equivalence}

\author{Susanne Pumpl\"un}
\address{School of Mathematical Sciences, University of Nottingham,
Nottingham NG7 2RD}
\email{Susanne.Pumpluen@nottingham.ac.uk}

\date{\today}

\begin{abstract}
 Employing isomorphisms between their ambient rings, we  propose new definitions of equivalence and isometry for skew polycyclic codes  that will lead to tighter classifications than existing ones. This reduces the number of previously known isometry and equivalence classes.
  In the process, we classify classes of skew $(f,\sigma,\delta)$-polycyclic codes with the same performance parameters, to avoid duplicating already existing codes, and state precisely when different notions of equivalence coincide.

 The generator of a skew polycyclic code is in one-one correspondence with the generator of a principal left ideal in its nonassociative unital ambient  ring.
   By allowing the ambient rings to be nonassociative, we eliminate the need on restrictions on the length of the codes. Ring isomorphisms that preserve the Hamming distance (called isometries) map generators of principal left ideals to generators of principal left ideals and preserve length, dimension, and Hamming distance of the corresponding isometric skew polycyclic codes.
\end{abstract}

\maketitle

\section*{Introduction}

 Any results on  equivalences for skew polycyclic codes
  can be used to optimize searches for good codes, avoid duplicating existing code parameters, and find equivalence classes in given sets of codes.
Therefore there has been extensive work towards classifying constacyclic, polycyclic, skew constacyclic and more generally skew polycyclic codes (also called skew $(\sigma,\delta)$-polycyclic codes or $(\sigma,\delta)$-codes  in  \cite{BU09, BouLe2013, BouU2014}) via different equivalence relations, exploiting their algebraic properties, see for example  \cite{Chen2012, Chen2015,  CON2025, Oua2025}.
  For instance, in \cite{BLMS2024}  left submodules of $R/Rf$  represent the skew skew $(f,\sigma,\delta)$-polycyclic codes and $(\sigma,\delta)$-pseudo-linear transformations when $f$ is not two-sided, and the Hamming isometrical equivalent $(\sigma,\delta)$-polycyclic codes are classified.

Let $S$ be a unital 
associative ring, $S^\times$ its set of invertible elements, and $f\in R=S[t;\sigma,\delta]$ a monic skew polynomial. When the context is clear, we will simply write skew polycyclic code
 instead of  skew $(f,\sigma,\delta)$-polycyclic code, and skew constacyclic code instead of skew $(\sigma,c)$-constacyclic code (when $f(t)=t^{\mathbf{n}}-a$). The most important case here will be the one where $S$ is commutative and a finite field or ring, however,  parts of the theory hold for any  unital ring $S$ and so we will develop those as such.

  In the present paper, we  extend and unify the existing classifications employing isometries between nonassociative rings, as initiated in \cite{NevPum2025}.

 1. We  view  skew $(f,\sigma,\delta)$-polycyclic codes as left principal ideals in their nonassociative unital \emph{ambient ring} $R/Rf=S[t;\sigma,\delta]/S[t;\sigma,\delta]f $  \cite{Pum2017}.  This approach is well known when $f$ generates a two-sided ideal in $R$, e.g. when $f$ lies in the center of $R$, in which case $R/Rf$ is a unital associative quotient ring.  For instance, some results for skew constacyclic codes over fields developed in \cite{LM25} were given using the associative quotient rings
 $ K[t;\sigma]/K[t;\sigma] (t^{\mathbf{n}}-a)$.   When $Rf$ is not a two-sided ideal, it is well known that $R/Rf$ is an $R$-module. However, $R/Rf$  also still carries a ring structure, only a nonassociative one, and is a nonassociative Petit algebra over  its center. We will call the (nonassociative) ambient ring $R/Rf$ a \emph{Petit ring}.

 A large percentage of the existing results on skew polycyclic codes, when viewed as ideals in associative rings, rely on multiplications of not more than two elements in the ideal, or can be effortlessly adjusted to the nonassociative setting. This basically guarantees their validity also when the ambient ring is not associative.

2.
We call two skew polycyclic codes    \emph{isometric}, if there is a Hamming-weight preserving map between them that is canonically induced by a \emph{monomial} isomorphism \emph{of degree} $k$ between their nonassociative ambient Petit rings  mapping $t$ to $\alpha t^k$ for some invertible $\alpha \in S^\times$ and restricting to some $\tau\in {\rm Aut}(S)$,
where the \emph{Hamming weight} of a polynomial $p(t)$ is defined as the number of nonzero coefficients of $p(t)$.  Any such an isomorphism between Petit rings, called an \emph{isometry},
 gives a correspondence between their ideals, hence between the associated codes, which preserves  Hamming distance, dimension and  length (Lemma \ref{le:pres}).
 This approach generalizes the previously used notions of isometry and equivalence between skew polycyclic codes that were defined using algebraic structures so far (which exclusively worked with isometries restricting to $id\in {\rm Aut}(S)$). It is lead by our  knowledge of automorphisms between  nonassociative Petit algebras \cite{BrownPumpluen2019, CB, BPS, BP, Pum2024}. These initial results were too narrow, however, and will be generalized to isomorphisms between  nonassociative Petit rings. In other words,  we will view Petit algebras not just as algebras over their centers, but over any subring of their center, and study their isomorphisms instead.

 3.  We then restrict ourselves to commutative rings $S$ and compute all the classes of codes that are equivalent to a given one (Theorem \ref{t:equivclasses}).
  Our refined notion of equivalence yields tighter classifications.
  For some  examples for codes over finite fields  see Section \ref{s:finite}. (This was observed for the first time in \cite[Example 5.8]{NevPum2025}.)  We  give a list of new and easy necessary and  sufficient criteria for  isometry and equivalence of skew polycyclic codes that will help distinguish non-equivalent codes with potentially  different performance parameters.

4. When $S$ is a finite commutative Chain ring, we generalise existing results for skew polycyclic codes over fields and study some examples.

We will focus mostly on  $(f,\sigma)$-polycyclic codes and commutative $S$, especially skew constacyclic codes, as for these the isomorphisms of the corresponding ambient Petit rings are quick to establish. By allowing the case that $\sigma=id$, we include the case that $R=S[t]$ where possible, thus recovering results from  \cite{Chen2012, Chen2015, OaHA25II}  as special cases.  Overall, our equivalence relations will result in tighter partitions of non-equivalent classes of codes than the ones defined in  \cite{Aydin2022, CON2025, Oua2025}.

  (Skew) polycyclic codes over finite rings, e.g. over mixed alphabets, are among the most prominent linear codes currently studied. They can be used to construct
  quantum error-correcting codes, usually via Gray maps and  the CSS construction, see for instance  \cite{BagPa2025, Quantum1, Quantum2, Quantum3, Quantum4}.
We expect that the results of this paper will make the search for new SPC codes and correspondingly quantum error-correcting codes more efficient and more organized. Good parameters for quantum codes are still under investigation, and any results on isometries and equivalences for SPC codes
  can be used to optimize search algorithms, in turn facilitating the discovery of new record-breaking quantum
codes.
In particular, when considering the CSS construction for quantum error-correcting codes, one needs to  concentrate on codes $C$ such that the dual code $C^\perp$ satisfies $C^\perp\subset C$ as explained in \cite{DinhBagUpadhyayBandiTansuchat2021}.
Thus an important next step in the classification will be to consider when  isometries preserve some suitable notion of duality of codes.

 Our main results can be found in Section \ref{sec:equiv} onwards and employ  Hamming weight preserving isomorphisms between Petit rings treated in Section \ref{sec:iso}. We look at the case that $\delta=0$ and assume $\sigma$ is an automorphism.

Given  two monic polynomials
 $$f(t) = t^{\mathbf{n}}-\sum_{i=0}^{\mathbf{n}-1} a_i t^i ,\quad h(t) = t^{\mathbf{n}}-\sum_{i=0}^{\mathbf{n}-1} b_i t^i \in R=S[t;\sigma],$$
we call two classes  of  skew $(f,\sigma)$-polycyclic codes  and skew $(h,\sigma)$-polycyclic codes  of length $\mathbf{n}$  \emph{isometric}, if there is an algebra isomorphism $G_{\tau,\alpha,k}$ between their ambient rings $R/Rf$ and $R/Rh$ which is defined by restricting to some $\tau\in {\rm Aut}(S)$ and maps $t$ to $\alpha t^k$ for some integer $k$ with $0<k<\mathbf{n}-1$ and some $\alpha\in S^\times$, and  \emph{Chen isometric} if $\tau=id$. When $k=1$, we call the classes equivalent, and \emph{Chen equivalent}, when additionally $\tau=id$, and write $G_{\tau,\alpha}=G_{\tau,\alpha,k}$. The current literature only uses Chen isometry and Chen equivalence.

When $\sigma\not=id$ has order $\mathbf{m}$, we show that we need $k\equiv 1 \mod \mathbf{m}$ and $\gcd(k,\mathbf{n})=1$ as necessary conditions for isometries $G_{\tau,\alpha,k}$ to exist (Proposition \ref{prop:monica}), which shows that only limited $k>1$ will actually appear.

\begin{theorem*}  (Theorem \ref{general_isomorphism_theorem1code})
 Let $\sigma$ have order $\mathbf{m}$. Two classes  ${\bf C}_f$ and  ${\bf C}_h$ of skew polycyclic codes of length $\mathbf{n}$ over $S$ are equivalent
 if and only if there exists
$\tau\in {\rm Aut}(S)$ that commutes with $\sigma$ and $\alpha \in S^{\times}$  such that
$$
\tau(a_i) = N_{\mathbf{n}-i}^{\sigma}(\sigma^{i} (\alpha))b_i
$$
 for all $i \in\{ 0, \ldots, \mathbf{n}-1\}$ (resp., Chen-equivalent if and only if above equation is true for $\tau=id$).
 \end{theorem*}

For codes over fields, isometry and equivalence coincide in a large number of cases. Equivalent codes have the same performance parameters.

\begin{theorem*}  (Theorem \ref{t:2})
 Suppose $K$ is a field and $\Aut(K)$ is abelian. Let $\sigma$ have order $\mathbf{m}$ such that $\mathbf{m} \geq \mathbf{n}-1$ and assume that $f$ does not generate two-sided ideals in $K[t;\sigma]$.
   Then
two classes  ${\bf C}_f$ and  ${\bf C}_h$ of skew polycyclic codes  of length $\mathbf{n}$ over $S$ are equivalent
 if and only if they are isometric, and Chen equivalent if and only if they are Chen isometric.
\end{theorem*}

It is known that if $\sigma$ has  order $\mathbf{m}$, $S_0$ denotes the ring fixed by $\sigma$, and either $a,b\in S \setminus S_0$, or  $a,b\in S_0^\times$ and  $\mathbf{m} \nmid \mathbf{n}$ then equivalence and isometry coincide for the classes of skew $(\sigma,a)$-constacyclic
and skew $(\sigma,b)$-constacyclic codes of length $\mathbf{n}$ \cite[Corollary 4.5]{NevPum2025}.

Moreover, when $\sigma$ has  order $\mathbf{m}$, the equivalence classes $[{\bf C}_{h}] $ of classes of skew polycyclic codes of length $\mathbf{n}$ over $S$ that are equivalent to ${\bf C}_{h}$ (respectively, Chen equivalent) are computed in Theorem \ref{t:equivclasses}.

 There is an intricate connection between classes of equivalent skew $(f,\sigma)$-polycyclic and equivalent skew $(\sigma,a)$-constacyclic codes over commutative rings $S$:
Two classes ${\bf C}_f$ and ${\bf C}_h$ of skew polycyclic codes over $S$ of length  $\mathbf{n}$ are equivalent via $G_{\tau,\alpha}$ if and only if
 for all $i \in\{ 0, \dots, \mathbf{n}-1\}$ where $a_i\not=0$, the class of $(\sigma,a_i)$-constacyclic codes is $(\mathbf{n}-i)$-equivalent to the class of $(\sigma,b_i)$-constacyclic codes via $G_{\tau,\sigma^i(\alpha)}$ (Proposition \ref{p:Chen equivpolycycliccode}).

For $a,b\in S^\times$,
the classes of skew $(\sigma, a)$-constacyclic and skew $\sigma$-cyclic codes  of length $\mathbf{n}$ are  equivalent if and only if
$a \in N_{\mathbf{n}}^\sigma (S^\times).$
  The classes of skew $(\sigma, a)$-constacyclic codes and  skew $\sigma$-negacyclic codes of length $\mathbf{n}$ are equivalent if and only if
 $ -a \in N_{\mathbf{n}}^\sigma (S^\times).$ In both cases,  isometry and equivalence coincide when $\mathbf{n}$  does not divide $ \mathbf{m}$ (Corollary \ref{c:varia}).

For finite rings we count the number of Chen isometry classes of skew constacyclic codes for classes of codes with a nonassociative ambient algebra (Theorem \ref{c:OuazzoufiniteS}).

 Conditions for  when two skew polycyclic codes cannot be equivalent are given in Section \ref{s:end}.

  We then look at isometries and equivalences of skew polycyclic codes over finite commutative chain rings in Section \ref{sec:chain}. We relate the isometries between the ambient Petit rings of skew polycyclic codes over $S$ and the ambient Petit rings of skew polycyclic codes over the residue field $K$ of $S$.

We extend Theorem \ref{general_isomorphism_theorem2} and \cite[Theorem 4.4]{NevPum2025} to Galois extensions $S/S_0$ of finite commutative chain rings when $\sigma\not=id$ and $\delta=0$.
 For skew constacyclic codes over finite chain rings, $\mathbf{n}$-isometries between their ambient rings where $k>1$ are rare, and can only occur under rigid conditions on $k$ and the ambient algebras.

 \begin{theorem*}(Theorem \ref{thm:chainmain})
 Let $S/S_0$ be a Galois extensions of chain rings with ${\rm Gal}(S/S_0)=\langle\sigma\rangle$.
Let $f,h\in R=S[t;\sigma] $ such that $\bar f(t)=t^{\mathbf{n}}-a_0$ and $\bar h(t)=t^{\mathbf{n}}-b_0$.
 Suppose  that $2\leq k < \mathbf{n}$ and that there exists an isometry $ G_{\tau,\alpha,k}: R/Rf\to R/Rh $, then
 \\ (i) $k\equiv 1 \mod \mathbf{m}$;
 \\ (ii) $ \mathbf{m}\mid  \mathbf{n}$;
 \\ (iii)  $\gcd(k, \mathbf{n})=1$;
 \\ (iv)  $\overline{b_0}, \overline{b_0}\in K_0$;
        \\ (v)  $(N_{K/K_0}(\alpha))^{\mathbf{n}/ \mathbf{m}}\overline{b_0}^k=\tau(\overline{a_0})$.
    \end{theorem*}

\begin{theorem*}(Theorem \ref{thm:chain2}) 
Let $S/S_0$ be a Galois extensions of chain rings with ${\rm Gal}(S/S_0)=\langle\sigma\rangle$. Suppose that one of the following holds.
\\ (i)  $\mathbf{n} \leq \mathbf{m}-1$ and  $\bar f$  does not generate a two-sided ideal in $K[t;\sigma]$.
 \\ (ii)  $\bar f(t)=t^{\mathbf{n}}-\overline{ a_0}$ and either $ \mathbf{m} \nmid \mathbf{n}$, or  $\overline{ a_0}$  is not in $K_0$.
\\ Then the notions of equivalence and isometry (resp., of Chen equivalence and Chen isometry) coincide for all skew $(f,\sigma)$-polycyclic codes over $S$ of length $\mathbf{n}$.
\end{theorem*}

Some examples for Galois rings are presented.

 Our results for isometries between the ambient rings also have consequences for classifying skew polycyclic codes with respect to the rank metric, especially over finite fields $K$  (where we need $\alpha\in K'$, the suitable subfield of $K$ we consider the rank metric over).
  In particular,  they will  help understand the different  MRD codes obtained from these \cite[Proposition 3.13]{OaHA25}. The most canonical approach here will be to additionally employ the isometries $G_{\tau,\alpha}$, for all $\tau\in {\rm Aut}(K)$ (Remark \ref{r:rank}).

The readers interested in codes are encouraged to skip the technical sections dealing with isomorphisms between nonassociative Petit rings, and focus on the results achieved for skew polycyclic codes.

 \section{Preliminaries}\label{sec:prelim}

\subsection{Nonassociative rings} \label{subsec:2}

Let $R$ be a unital commutative ring and let $A$ be an $R$-module.
We call $A$ a \emph{(nonassociative) algebra} over $R$ if there exists an $R$-bilinear
map $A\times A\mapsto A$, $(x,y) \mapsto x \cdot y$, usually denoted simply
by juxtaposition $xy$, the  \emph{multiplication} of $A$. An algebra
$A$ is called \emph{unital} if there is an element in $A$, denoted by
1, such that $1x=x1=x$ for all $x\in A$. We will only consider unital
algebras.

 Define $[x, y, z] = (xy) z - x (yz)$. The associativity in $A$ is measured by the {\it left nucleus} defined as ${\rm Nuc}_l(A) = \{ x \in A \, \vert \, [x, A, A]
= 0 \}$, the {\it middle nucleus} ${\rm Nuc}_m(A) = \{ x \in A \,
\vert \, [A, x, A]  = 0 \}$ and  the {\it right nucleus}   ${\rm
Nuc}_r(A) = \{ x \in A \, \vert \, [A,A, x]  = 0 \}$.  The sets ${\rm Nuc}_l(A)$, ${\rm Nuc}_m(A)$ and ${\rm Nuc}_r(A)$ are associative
subalgebras of $A$. Their intersection
 ${\rm Nuc}(A) = \{ x \in A \, \vert \, [x, A, A] = [A, x, A] = [A,A, x] = 0 \}$ is the {\it nucleus} of $A$ and is an associative subalgebra of $A$ containing $R1$
and $x(yz) = (xy) z$ whenever one of the elements $x, y, z$ is in
${\rm Nuc}(A)$. The  {\it commuter} of $A$ is defined as ${\rm
Comm}(A)=\{x\in A\,|\,xy=yx \text{ for all }y\in A\}$ and the {\it
center} of $A$ is ${\rm C}(A)=\text{Nuc}(A)\cap  {\rm Comm}(A)$
 \cite{Sch}.

  We note that every unital nonassociative algebra is a unital nonassociative ring $(A,+,\cdot)$, with $(A,+)$ being an abelian  group, $(A,\cdot)$ a groupoid, and with distributivity laws satisfied by definition of the multiplication. Conversely, every nonassociative ring can be viewed as a nonassociative algebra over its center.

 A nonassociative ring  $A$  is called a \emph{proper} nonassociative ring  if it is not associative.

  Two $R$-algebras $A$ and $A'$ are called \emph{$R$-isotopic} if and only if there exist $R$-linear isomorphisms $F,G,H:A\rightarrow A'$ such that $H(uv)=F(u)G(v)$ for all $u,v\in A$.

A nonassociative algebra $A\not=0$ (resp., a ring $A\not=0$) over a
field $K_0$ is called a \emph{division algebra}, resp., a \emph{division ring}, if
for all $a\in A$, $a\not=0$, the left and  right multiplication  with $a$,
$L_a(x)=ax$ and $R_a(x)=xa$, are  bijective maps. If $A$ is a
finite-dimensional algebra over $K_0$, then $A$ is a division algebra if and only if $A$ has no zero divisors \cite{Sch}.

\subsection{Skew polynomial rings}

Let $S$ be a unital associative ring (it need not be commutative), and $\sigma$ a ring automorphism of
$S$. A \emph{left $\sigma$-derivation} is an additive map $\delta:S\rightarrow S$, such that
$\delta(ab)=\sigma(a)\delta(b)+\delta(a)b$
for all $a,b\in S$.  Define $S_0={\rm Fix}(\sigma)=\{a\in S\,|\,
\sigma(a)=a\}$ and ${\rm Const}(\delta)=\{a\in S\,|\, \delta(a)=0\}$.

The \emph{skew polynomial ring} $R=S[t;\sigma,\delta]=\left\{ a_0+a_1t+\dots +a_nt^n \,|\, a_{i} \in S, n \in \mathbb{N}\right\}$ is the
set of skew polynomials $a_0+a_1t+\dots +a_nt^n$, $a_i\in
S$, with the usual addition of two polynomials,  where the multiplication is defined by the rule
 $ta=\sigma(a)t+\delta(a)$ for all $a\in S$. That means
$$at^nbt^{m}=\sum_{j=0}^n a(\Delta_{n,j}\,b)t^{m+j}$$ for all $a,b\in
S$, where the map $\Delta_{n,j}$ is defined recursively via
$$\Delta_{n,j}=\delta(\Delta_{n-1,j})+\sigma (\Delta_{n-1,j-1}),$$ with
$\Delta_{0,0}=id_S$, $\Delta_{1,0}=\delta$, $\Delta_{1,1}=\sigma $  \cite{O1}.
Therefore $\Delta_{n,j}$ is the sum of all monomials in $\sigma$
and $\delta$ of degree $j$ in $\sigma$ and degree $n-j$ in $\delta$
\cite[p.~2]{J96}. If $\delta=0$, then $\Delta_{n,n}=\sigma^n$.

When $\sigma=id$ and $\delta=0$, $S[t]=S[t;id,0]$ is the usual ring of left
polynomials. When $\delta=0$,  we use the notation $S[t;\sigma]=S[t;\sigma,0]$.

For a positive integer $\mathbf{n}$, let $R_\mathbf{n}=\{g\in R \,|\, \deg(g)<\mathbf{n}-1 \}$.

 For $f(t)=a_0+a_1t+\dots +a_\mathbf{n}t^\mathbf{n}\in R$ with $a_\mathbf{n}\not=0$ define ${\rm deg}(f)=\mathbf{n}$ and ${\rm deg}(0)=-\infty$.
Then ${\rm deg}(gh)\leq{\rm deg} (g)+{\rm deg}(h)$ (with equality if and only if the product of the leading coefficients of
$h$ and $g$ is not zero).
 An element $f\in R$ is \emph{irreducible} in $R$ if it is not a unit and  it has no proper factors, i.e if there do not exist $g,h\in R$ with
 $1\leq {\rm deg}(g),{\rm deg} (h)<{\rm deg}(f)$ such
 that $f=gh$.

When $S$ is a division algebra, for $u,v\in R$, the \emph{greatest common right divisor} ${\rm gcrd}(u,v)$ is the polynomial $w\in R$ of smallest degree such that $w$ right divides both $u$ and $v$ (i.e., $u=sw$ and $v=s'w$ for some $s,s'\in R$) and for every other $w'\in R$ that right divides $u$ and $v$ we have that $w'$ right divides $w$. The \emph{greatest common left divisor} ${\rm gcld}(u,v)$ is defined analogously.

 Two nonzero skew polynomials $f$ and $g$ in $R$ are  \emph{similar}, written $f\sim g$, if $R/Rf\cong R/Rg$ as left $R$-modules.
  When $S$ is a division algebra, equivalently we can say that $f\sim g$, if there exist $u,v\in R$, such that the greatest common right divisor ${\rm gcrd}(f,u)=1$, the greatest common left divisor ${\rm gcld}(g,v)=1$ and $gu=vf$.  
 The element $u\in R$ can be chosen such that $u\in R_\mathbf{n}$ \cite[p.~11]{J96}.

 We call  $f\in R$ \emph{two-sided} if $Rf$  is a two-sided ideal in $R$.

 The \emph{companion matrix} of a monic polynomial $f=t^{\mathbf{n}}-\sum_{i=0}^{\mathbf{n}-1}a_it^i\in R=S[t;\sigma,\delta]$ of degree $\mathbf{n}$  is given by
{\begin{equation*}
 A_f= \left( \begin{array}{cccccccc}
                                  0&     1                                    &     \cdots  &   &   &     &     0\\
                                  &         0  &    1   \cdots    &     &     & \cdots   &   \\
  \vdots                     &            &    \ddots         &       &    &         &   \vdots\\
                                            &                                                  &                                      &                \ddots                          &                                          &             & \\
  &                                               &             &                                     &
 \ddots                       &          & \\
  \vdots                     &                                            &           &                                       &                                       & \ddots  &  1\\
 a_0 &     a_1       &      \cdots   &  &       & \cdots   &      a_{\mathbf{n}-1}                                           \\
 \end{array} \right)
 \end{equation*}}

 Let $h\in R=S[t;\sigma,\delta]$ be a monic polynomial then $f\sim h$ if and only if there exists an invertible matrix $B\in M_\mathbf{n}(K)$ such that
 $A_f=\sigma(B)A_h B^{-1}+\delta (B)$.
  Here, $\sigma$ respectively $\delta$ applied to $B$ means applying it to each entry of the matrix.

 Note that the fact that $R/Rf\cong R/Rh$ as left $R$-modules does not imply that $R/Rf\cong R/Rh$ as nonassociative rings. In certain cases, however, it may imply that $R/Rf$ and $ R/Rh$ are isotopic rings  \cite[Theorem 3.4]{Pu2025}.

\subsection{Isomorphisms between skew polynomial rings}\label{sec:isoR}

Let $S$ be a unital associative ring (it need not be commutative), and let $S^\times$ be the set of invertible elements in $S$.
 Let $f,h\in R=S[t;\sigma,\delta]$ be monic.
For all $\tau\in {\rm Aut}(S)$ and  $\beta \in S$ and $\alpha\in S^\times$, such that
$$\beta\tau(b)+\alpha\delta(\tau(b))=\tau(\sigma(b))\beta +\tau(\delta(b)) \text{ and } \alpha\sigma(\tau(b))=\tau(\sigma(b)) \alpha$$
for all $b\in S$, the map
$H_{\tau,\beta,\alpha}: S[t;\sigma,\delta]\rightarrow S[t;\sigma,\delta]$,
$$H_{\tau,\beta,\alpha}: S[t;\sigma,\delta]\rightarrow S[t;\sigma,\delta],\quad \sum_{i=0}^{n} b_it^i \mapsto \sum_{i=0}^{n} \tau(b_i) (\beta+\alpha t)^i $$
is a ring automorphism of $R$  \cite[p. 4]{LaLe92}.
 In particular,  if
$$\alpha\delta(\tau(b))=\tau(\delta(b)) \text{ and } \alpha\sigma(\tau(b))=\tau(\sigma(b))\alpha$$
for all $b\in S$ then
$$G_{\tau,\alpha}: S[t;\sigma,\delta]\rightarrow S[t;\sigma,\delta],$$
$$\sum_{i=0}^{n} d_it^i \mapsto \sum_{i=0}^{n} \tau(d_i) ( \alpha t)^i $$
 is a ring automorphism that preserves the Hamming weight.

For the special case that $\delta=0$ we know more.

\begin{theorem} \cite[Theorem 3]{Ri}\label{thm:rimm}
Let $R=S[t;\sigma]$ (we allow $\sigma=id$).
Then the map $G:S[t;\sigma]\rightarrow S[t;\sigma]$,
$G(t)=\sum_{i=0}^{\ell} \alpha_i t^i$, $\alpha_i \in S$, extends  $\tau\in {\rm Aut}(S)$ to an isomorphism $G: S[t;\sigma] \rightarrow S[t;\sigma]$ if and only if
  \begin{enumerate}
  \item $\sigma(\tau(a)) \alpha_i=\alpha_i\tau(\sigma^i (a))$  for all $ a\in S $ and $i\in\{0,\dots, \ell\},$
\item $  \alpha_1\in S^\times,$
\item $ \alpha_i \text{ is nilpotent for all }  i\in\{2,\dots \ell\}.$
\end{enumerate}
\end{theorem}

 Note that if $S$ is a  division algebra, for example a field,
 $\alpha_i=0 $ for all $ i\in\{2,\dots, \ell\}$ holds automatically in the setting of Theorem \ref{thm:rimm}.

Since we are only interested in maps that preserve the Hamming distance,
 we will only consider the case that $\alpha_0=0$ and $\alpha_i=0 $ for all $ i\in\{2,\dots, \ell\}$. We conclude that $G: S[t;\sigma] \rightarrow S[t;\sigma]$
 extends  $\tau\in {\rm Aut}(S)$ to an isomorphism that respects the Hamming weight with  $G(t)= \alpha t$ for some $\alpha \in S^\times$, if and only if
  $\sigma(\tau(a)) \alpha= \alpha \tau(\sigma (a))$  for all $ a\in S $.

\section{Skew polycyclic codes viewed as ideals in nonassociative Petit rings} \label{subsec:structure}

\subsection{Petit rings}
Let $S$ be a unital associative ring (it need not be commutative).
Let $f \in R=S[t;\sigma,\delta]$ have an invertible leading coefficient. Then
for all $g\in R$,  there exist  uniquely
determined $r,q\in R$ with
 ${\rm deg}(r)<{\rm deg}(f)$, such that
$g(t)=q(t)f(t)+r(t),$ and since we assume $\sigma\in{\rm Aut}(S)$,
also uniquely
determined $r',q'\in R$ with ${\rm deg}(r')<{\rm deg}(f)$, such
that $g(t)=f(t)q'(t)+r'(t)$ (e.g., \cite[Proposition 1]{Pum2017}).

Let ${\rm mod}_r f$ denote the remainder of right division by $f$ and
 ${\rm mod}_l f$ the remainder of left division by $f$.
Since the remainders are uniquely determined, the set $\{g\in R\,|\, {\rm deg}(g)<\mathbf{n}\}$ canonically represent the
elements of the (left, resp. right)
$R$-modules $R/Rf$ and $R/ fR$.

 \begin{theorem} \label{def:Petit}
$(i)$ The additive group $\{u\in R\,|\, {\rm deg}(g)< \mathbf{n}\}$  together with the multiplication
$$u \circ_f v=uv \,\,{\rm mod}_r f,$$
 where the right hand side is the remainder of $uv$ after division by $f$ on the right, is a unital nonassociative ring  denoted by denoted by $R/Rf$ or  $\mathbb{S}_f$.
\\ $(ii)$  The additive group $\{u\in R\,|\, {\rm deg}(g)< \mathbf{n}\}$  together with the multiplication
$$u \circ_{f,l} v=uv \,\,{\rm mod}_l f, $$
 where the right hand side is the remainder of $uv$ after division by $f$ on the left, is a unital nonassociative ring denoted by $R/fR$ or $\,_f\mathbb{S}$.
\end{theorem}

We will usually simply use juxtaposition to write the multiplication in $\mathbb{S}_f$ and $\,_f\mathbb{S}$.

Note that $\mathbb{S}_f$ is an associative ring exactly when $S[t;\sigma,\delta] f$ is a two-sided ideal which is the case if and only if $f$ is two-sided. In that case, $\mathbb{S}_f$ coincides with the classical quotient ring.
 For example, this happens when $f\in C(R)$.
When  $S$ is a division algebra and $\delta=0$, the two-sided elements $f(t)$ are all of the form $ac(t)t^\ell$, where $a\in S^\times$, $c(t)\in C(R)$ and $\ell\geq 0$ is an integer
 \cite[Theorem 1.1.22]{J96}.

Let $f\in S[t;\sigma,\delta]$ have degree $\mathbf{n}\geq 2$ and an
invertible leading coefficient. Then $\mathbb{S}_f$ is a free left $S$-module  of rank $\mathbf{n}$ with basis $t^0=1,t,\dots,t^{\mathbf{n}-1}$.

  We will only consider monic polynomials $f\in R$, and when $S$ is a division algebra, this will not even limit our choices since $\mathbb{S}_f = \mathbb{S}_{af}$ for all invertible $a\in S^\times$.  If $f$ is reducible then $\mathbb{S}_f$ contains zero divisors.
In the following,  we will thus always assume that
  $$f\in R \text{\emph{ has degree} } \mathbf{n}>1 \text{\emph{ and an invertible leading coefficient. } }$$

Both $\mathbb{S}_f$ and $\,_f\mathbb{S}$ are unital nonassociative algebras (called \emph{Petit algebras})  over the
commutative subring $$\{a\in S\,|\, au=ua \text{ for all } u\in
\mathbb{S}_f\}={\rm Comm}(\mathbb{S}_f)\cap S$$ of $S$, and have center
$$ C(S)\cap{\rm Fix}(\sigma)\cap {\rm Const}(\delta).$$
If $S$ is a division ring and  $\mathbb{S}_f$ is not associative then $C(\mathbb{S}_f)=\{a\in S\,|\, au=ua \text{ for all } u\in
\mathbb{S}_f\}$.

 If $\mathbb{S}_f$ is not associative then
$S\subset{\rm Nuc}_l(\mathbb{S}_f),\,\,S\subset{\rm Nuc}_m(\mathbb{S}_f)$ and $$\{u\in
R\,|\, {\rm deg}(u)<\mathbf{n} \text{ and }fu\in Rf\}= {\rm Nuc}_r(\mathbb{S}_f).$$
When $S$ is a division ring, these inclusions become equalities.

 In the following, we will exclusively need the nonassociative ring structure on $R/Rf$ and thus consider the Petit  algebras $\mathbb{S}_f$ and $\,_f \mathbb{S}$ (Petit's original paper \cite{P66} indeed only
considered division rings $S$) as nonassociative rings and call them \emph{Petit rings}. We will  focus on the rings
$\mathbb{S}_f$, since  $\,_f\mathbb{S}$ is anti-isomorphic to $\mathbb{S}_f$. We note that the principal right ideals of the rings $\,_fS$ play an equally relevant role in linear code constructions along similar arguments as presented in  \cite{BLMS2024}.

 \begin{proposition}
 Let $f\in R=S[t;\sigma,\delta]$ have an invertible leading coefficient.
 The canonical anti-isomorphism
 $$\psi: S[t;\sigma,\delta]\rightarrow S^{op}[t;\sigma^{-1},-\delta\circ\sigma^{-1}],$$
 $$\psi(\sum_{k=0}^{\ell}d_kt^k)=\sum_{k=0}^{\ell}(\sum_{i=0}^{k}\Delta_{\ell,i}(d_k))t^k$$
 between the noncommutative associative rings $S[t;\sigma,\delta]$ and $ S^{op}[t;\sigma^{-1},-\delta\circ\sigma^{-1}]$ induces an anti-automorphism
 between the nonassociative unital rings $$\mathbb{S}_f=S[t;\sigma, \delta]/ S[t;\sigma,\delta]f$$
  and
 $$\,_{\psi(f)}\mathbb{S}=S^{op}[t;\sigma^{-1},-\delta\circ\sigma^{-1}]/\psi(f) S^{op}[t;\sigma^{-1},-\delta\circ\sigma^{-1}].$$
\end{proposition}

For a proof see e.g. \cite[Proposition 3]{Pum2017}.

If $\delta=0$
we have
$$\psi(\sum_{k=0}^{\ell}d_kt^k)=\sum_{k=0}^{\ell}\sigma^{-k}(d_k)t^k.$$

 For $0\not=a\in \mathbb{S}_f$, left multiplication $L_a$ is an
$S_0$-module endomorphism. Moreover, right multiplication $R_a$ is a left $S$-module homomorphism for all $0 \neq a \in \mathbb{S}_f$.

\begin{proposition}\label{prop:ideals} \cite{Pum2017} Let $S$ be a unital associative division ring and $f\in R=S[t;\sigma,\delta]$.
\\ $(i)$ All left ideals in $\mathbb{S}_f$ are generated by some monic right divisor $g$ of $f$ in $R$.
\\ $(ii)$ If $f$ is irreducible, then $\mathbb{S}_f$ has no non-trivial left ideals.
\end{proposition}

\begin{example} \cite{Pum2017}
Let $\sigma\in {\rm Aut}(S)$ have  order $\mathbf{m}$ and ${\rm Fix}(\sigma)=S_0$. For $f(t)=t^{\mathbf{n}}-d \in R=S[t;\sigma]$, $d\in S$ invertible, the multiplication in the $S_0$-algebra $R/R(t^{\mathbf{n}}-d)$ is defined via
\[
 (at^i)(bt^j) =
  \begin{cases}
   a\sigma^i(b)   t^{i+j} & \text{if } i+j < \mathbf{n}, \\
     a \sigma^i(b)  t^{(i+j)-m}d & \text{if } i+j \geq \mathbf{n},
  \end{cases}
\]
for all $a,b\in S$ and then linearly extended. In the literature,  $S$ is often a  field in which case we know that $R/R(t^{\mathbf{n}}-d)$ is an associative algebra over the fixed field of $\sigma$, if and only if $d\in S_0$ and  the  order $\mathbf{m}$ of $\sigma$ divides $\mathbf{n}$ \cite[(7), (9)]{P66}. For $d=1$ in particular, $f$ is two-sided if and only
$\mathbf{m}\mid \mathbf{n}$.
\end{example}

Homomorphisms between skew polynomial rings where $S$ a division algebra were investigated in \cite{LaLe92}, the main results we need can be found in Section \ref{sec:isoR}.
Indeed, when $S$ is a division algebra then every ring isomorphism $G: S[t;\sigma,\delta]\rightarrow S[t;\sigma,\delta]$ is monomial of degree one (i.e. maps $t$ onto a polynomial of degree one) \cite{LaLe92}. When $S$ is not a division algebra, this need not be the case, see  Theorem \ref{thm:rimm} in Section \ref{sec:isoR}.

 It is straightforward to see that every ring isomorphism $G: S[t;\sigma,\delta]\rightarrow S[t;\sigma,\delta]$ which restricts to some automorphism $\alpha$ on $S$ and  maps $t$ to $\alpha t$ for some invertible $\alpha\in S^\times$,
 canonically induces an isomorphism between the nonassociative rings $R/Rf$  and  $R/RG(f)$, since $G|_{R_\mathbf{n}}:R_\mathbf{n}\rightarrow R_\mathbf{n}$ and   $G(u\circ_f v)=G(uv -q f)=G(u)G(v)-G(q)G(f)=G(u)\circ_{G(f)} G(v)$ for some uniquely determined $q\in R_\mathbf{n}$.
We will use this observation repeatedly, for instance in the proof of the next lemma.  Note that the skew polynomial $G(f(t)) = \sum_{i=0}^{\mathbf{n}} \tau(a_i) (\alpha t)^i \in S[t;\sigma,\delta]$ is usually not monic, but will have an invertible lead coefficient since $\alpha\in S^\times$.

\begin{lemma}\label{le:inducediso}
Let $\alpha\in S^\times$, and $f(t)=t^{\mathbf{n}}-\sum_{i=0}^{\mathbf{n}-1} a_i t^i  \in S[t;\sigma,\delta]$ be a monic polynomial. Define
$$s(t) =  (\alpha t)^\mathbf{n}-\sum_{i=0}^{\mathbf{n}} \tau(a_i) (\alpha t)^i \in S[t;\sigma,\delta].$$
If
$$\alpha\delta(\tau(b))=\tau(\delta(b)) \text{ and } \alpha\sigma(\tau(b))=\tau(\sigma(b)) \alpha$$
for all $b\in S$, then there exists a ring isomorphism 
$G: R/Rf\to R/Rs$.
\end{lemma}

\begin{proof}
 Suppose that $\alpha\in S^\times$ and that we have that
$\alpha\delta(\tau(b))=\tau(\delta(b))$  and  $ \alpha\sigma(\tau(b))=\tau(\sigma(b))\alpha$
for all $b\in S$.
 It is well known that under these assumptions, the map $G: S[t;\sigma,\delta]\rightarrow S[t;\sigma,\delta]$, $$G(\sum_{i=0}^{\ell}d_it^i)=\sum_{i=0}^{\ell}\tau(d_i) (\alpha t)^i$$
  is a ring isomorphism (e.g., see Section \ref{sec:isoR}, or \cite[p.4]{LaLe92} or \cite{LaLeLe89} which to our knowledge are the earliest papers where this is explained; the argument presented there also works when the ring $S$ is not a division algebra). Since  $G(t)=\alpha t$ and $\alpha\in S^\times$,
 $G|_{R_\mathbf{n}}:R_\mathbf{n}\rightarrow R_\mathbf{n}$ is a well-defined bijective additive map.
 Moreover, since $f$ is monic we know that for all $u,v\in R_\mathbf{n}$ there exist  unique $q,r\in R_\mathbf{n}$, $\deg (r)< \mathbf{n}$, such that $uv= qf+r$. Hence we  can write
  $u\circ_f v=uv -q f$ for some unique $q\in R_\mathbf{n}$ and thus obtain that
 $$G(u\circ_f v)=G(uv -q f)=G(u)G(v)-G(q)G(f)=G(u)\circ_{G(f)} G(v)$$
  with $G(q)\in R_\mathbf{n}$. Thus $G$ is also multiplicative. Therefore the map $G$
 canonically induces an isomorphism between the nonassociative rings $R/Rf$  and  $R/RG(f)$. When restricted to $S$, this isomorphism is $\tau$, and it maps $t$ to $\alpha t$. The assertion follows by putting $s(t)=G(f(t))$.
\end{proof}

 \subsection{Skew $(f,\sigma,\delta)$-polycyclic codes}

 A \emph{left (resp., right) linear code  of length $\mathbf{n}$ over $S$} is a left (resp., right)  free submodule of the $S$-module $S^\mathbf{n}$.
From now on, let $$f(t)=t^{\mathbf{n}}-\sum_{i=0}^{\mathbf{n}-1}a_it^i\in S[t;\sigma,\delta]$$ be a reducible monic polynomial of degree $\mathbf{n}>1$.

We will use the bijective map $\Phi:S^{\mathbf{n}}\rightarrow R/Rf$,
$$(c_0,c_1,\dots,c_{\mathbf{n}-1})\mapsto \sum_{i=0}^{\mathbf{n}-1}c_it^i,$$
to ``move'' between vectors  of length $\mathbf{n}$  and skew polynomials of rank  ${\mathbf{n}}-1$:
 For a left linear code $C$  of length $\mathbf{n}$ over $S$ we denote by $C(t)$ the set of skew-polynomials
 $c(t)=\sum_{i=0}^{\mathbf{n}-1}c_it^i$ associated to the codewords $(c_0,\dots,c_{\mathbf{n}-1})\in C$.

The currently prevailing way to define  (skew) polycyclic codes in the literature is via some shift operation generalising the shift used to define classical cyclic codes, e.g. cf. \cite[Definition 2.1]{BagPa2025},  \cite[Definition 1]{BLMS2024} \cite[Definition 2.18.]{OaHA25}.

 \begin{definition}\label{def:main}
  $(i)$
  A left linear code $C\subset S^{\mathbf{n}}$ is called a   \emph{right skew  $(f,\sigma,\delta)$-polycyclic code}, if for each codeword  $(c_0,c_1,\ldots, c_{\mathbf{n}-1})$ of $C$,  also
$$
\left( 0, \sigma(c_0),\sigma(c_1),\ldots, \sigma(c_{\mathbf{n}-2})\right) + \sigma(c_{\mathbf{n}-1}) (a_0,a_1,\ldots, a_{\mathbf{n}-1})+(\delta(c_0),\delta(c_1),\ldots, \delta(c_{\mathbf{n}-1}))  \in C.
$$
$(ii)$ A left linear code $C\subset S^{\mathbf{n}}$ is called a \emph{left skew  $(f,\sigma,\delta)$-polycyclic code},  if for each codeword  $(c_0,c_1,\ldots, c_{\mathbf{n}-1})$ of $C$,  also
$$
\left(\sigma(c_1),\sigma(c_2),\ldots, \sigma(c_{\mathbf{n}-1}),0\right) + \sigma(c_{0}) (a_0,a_1,\ldots, a_{\mathbf{n}-1})+ (\delta(c_0),\delta(c_1),\ldots, \delta(c_{\mathbf{n}-1}))   \in C.
$$
\end{definition}

We will only look at right skew polycyclic codes unless explicitly stated otherwise, and we will drop the ``right'' in their name from now on.
We can rewrite $(i)$ as
$$
\left( 0, \sigma(c_0),\sigma(c_1),\ldots, \sigma(c_{\mathbf{n}-2})\right) +  (\sigma(c_{\mathbf{n}-1}) a_0,\sigma(c_{\mathbf{n}-1})a_1,\ldots, \sigma(c_{\mathbf{n}-1})a_{\mathbf{n}-1})+\delta(c)$$
$$
  = ( \sigma(c_{\mathbf{n}-1}) a_0, \sigma(c_0)+\sigma(c_{\mathbf{n}-1})a_1,\sigma(c_1)+\sigma(c_{\mathbf{n}-2})a_2,\ldots, \sigma(c_{\mathbf{n}-2})+\sigma(c_{\mathbf{n}-1})a_{\mathbf{n}-1})+\delta(c) \in C.
$$

 \begin{definition} A  skew  $(f,\sigma,\delta)$-polycyclic code $C\subset S^{\mathbf{n}}$ over $S$ is called a
 \emph{skew  $(\sigma, d)$-constacyclic code}, if $\delta=0$ and $f(t)=t^{\mathbf{n}}-d$ for a non-zero $d\in S$, that is
   $$(c_0,\dots,c_{\mathbf{n}-1})\in  C\Rightarrow (\sigma(c_{\mathbf{n}-1})d,\sigma(c_0),\dots,\sigma(c_{\mathbf{n}-2}))\in  C.$$
 \end{definition}

 However, we will see that there is also an algebraic way to define an $(f,\sigma,\delta)$-skew polycylic code $C$  of length $\mathbf{n}$ over $S$
  as a linear code $C$  of length $\mathbf{n}$ over $S$, whose the corresponding set of skew polynomials $C(t)$ is a principal left ideal  in the nonassociative Petit ring $R/Rf$.

 We will then use the map $\Phi:S^{\mathbf{n}}\rightarrow R/Rf$
to ``move'' between skew $(f,\sigma,\delta)$-skew polycylic code $C$  of length $\mathbf{n}$ and principal left ideals in the nonassiciative Petit ring $R/Rf$ we call the \emph{ambient ring} of the codes.

More precisely, we will show that our skew $(f,\sigma,\delta)$-skew polycylic codes
$C$  of length $\mathbf{n}$ consist of all vectors $(c_0,\dots,c_{\mathbf{n}-1})\in S^{\mathbf{n}}$ obtained from the elements
 $c(t)=\sum_{i=0}^{\mathbf{n}-1}c_it^i$ in a left principal ideal  of the nonassociative unital Petit ring $\mathbb{S}_f$:
A skew $(f,\sigma,\delta)$-code $C\subset S^{\mathbf{n}}$ is a subset of $S^{\mathbf{n}}$ consisting of the vectors
$(c_0,\dots,c_{\mathbf{n}-1})$ obtained from elements $c(t)=\sum_{i=0}^{\mathbf{n}-1}c_it^i$
in a left principal ideal $g\mathbb{S}_f=S[t;\sigma,\delta]g/S[t;\sigma,\delta]f$ of $\mathbb{S}_f$, with $g$ a right divisor of $f$ which is assumed to be monic 
and is the polynomial of smallest degree that generates the left ideal.  This means that the set of vectors corresponding to the elements $ \{ g,t g,\ldots,t^{k-1}g\} $ in $\mathbb{S}_f$ forms a basis of $C $ and the \emph{dimension} of $  C $ is $ k=n-\deg(g)$. This also means that  there is a one-to-one correspondence between principal left ideals of $\mathbb{S}_f$ that are generated by monic polynomials $g\in R_\mathbf{n}$ and  skew $(f,\sigma,\delta)$-polycyclic codes  of length $\mathbf{n}$ over $S$   \cite{Pum2017}.
Note that in the setup we study here, $g$ is uniquely determined only up to similarity of skew polynomials.

There is an obvious correlation between our point of view using nonassociative unital rings $R/Rf$  as the algebraic objects that correspond to
skew $(f,\sigma,\delta)$-skew polycylic codes, and the standard approach which uses cyclic submodules of the $R$-module $R/Rf$. The left multiplication in the nonassociative ring $R/Rf$ canonically induces the right module structure of $R/Rf$. Thus the principal left ideals $Rg/Rf$ are in one-one correspondence with the cyclic $R$-submodules $Rg/Rf$ of the $R$-module $R/Rf$.

\begin{lemma}
Let $L_a:\mathbb{S}_f\to \mathbb{S}_f$, $L_a(u)=a\circ u$, be the left multiplication by $a\in \mathbb{S}_f$ in the nonassociative unital ring $\mathbb{S}_f=R/Rf$.
\\ (i) For all $a\in R$, $u\in R/Rf$,  the left $R$-module structure of $R/Rf$ is given by
$$a\cdot u= L_a(u)=a\circ u.$$ 
Therefore, the left principal ideal $g\mathbb{S}_f$ generated by a monic right divisor $g$ of $f$ in the nonassociative ring $\mathbb{S}_f$ is the same as the cyclic $R$-submodule $Rg/Rf$.
\\ (ii) Every ring isomorphism $H:R/Rf \to R/Rh$ between two nonassociative unital rings $R/Rf$ and $R/Rh$ canonically maps a principal left ideal $Rg/Rf$ in the nonassociative ring $\mathbb{S}_f=R/Rf$ to a principal left ideal $R \,H(g)/Rh$ in the nonassociative Petit ring $\mathbb{S}_h=R/Rh$.
In particular, this means that $H(g)$ must be a right divisor of a scalar multiple of $h$.
 \\ (iii) Every ring isomorphism $H:R/Rf \to R/Rh$ between  two nonassociative unital rings $R/Rf$ and $R/Rh$ such that $H|_S=\tau\in {\rm Aut}(S)$, canonically induces an $\tau$-semilinear $S$-module isomorphism $H:R/Rf \to R/Rh$ between the two free $S$-modules $R/Rf$ and $R/Rh$.
\end{lemma}

\begin{proof}
Claims $(ii)$ and $(iii)$ directly follow from $(i)$ which is straightforward to see, if we view every $a\in R$ as an element representing some element in $R/Rf$.
\\ $(ii)$ Let $H:R/Rf \to R/Rh$ be a ring isomorphism and $g$ a monic principal right divisor of $f$. Take the left ideal $g\mathbb{S}_f$. Then $H(u\circ_f g)=H(u)\circ_h H(g)$ for all $u\in \mathbb{S}_f=R/Rf$, therefore $v\circ_h H(g)$ for all $v\in \mathbb{S}_h=R/Rh$, so that $H(g)\mathbb{S}_h$ is a principal left ideal in $\mathbb{S}_h$.
\\ $(iii)$ We have $H(u\circ_f v)=H(u)\circ_h H(v)$ for all $u,v\in \mathbb{S}_f=R/Rf$, therefore also $H(a\circ_f v)=\tau(a)\circ_h H(v)$
 for all $a\in S$, i.e. $H(a\cdot v)=\tau(a)\cdot H(v)$ for the scalar multiplications of the modules.
 \end{proof}

This means we can directly translate  results obtained in \cite[Theorem 3.2]{BouLe2013} to the setting of nonassociative rings.

\begin{theorem}\label{thm:newTheorem3.2} \cite{Pum2017}
Let $g=\sum_{i=0}^{r}g_it^i$ be a monic polynomial which is a right divisor of $f$.
\\ $(i)$ The skew $(f,\sigma,\delta)$-polycyclic code $C\subset S^{\mathbf{n}}$ corresponding to the principal ideal
$g \mathbb{S}_f$ is a free left $S$-module of dimension $\mathbf{n}-{\rm deg}(g)$.
\\ $(ii)$ If $(c_0,\dots,c_{\mathbf{n}-1})\in C$ then $L_t(c_0,\dots,c_{\mathbf{n}-1})\in C$, where $L_t$ denotes the left multiplication with $t$ in the nonassociative ring $R/Rf$.
\\ (iii) The matrix generating $C$ represents the right multiplication $R_g$ with $g$ in  the nonassociative ring  $R/Rf$,
calculated with respect to the basis $1,t,\dots,t^{\mathbf{n}-1}$, identifying  elements $c(t)=\sum_{i=0}^{\mathbf{n}-1}c_it^i$
with  vectors $(c_0,\dots,c_{\mathbf{n}-1})$ via $\Phi$.
\end{theorem}

Note that $(i)$ is  \cite[Theorem 3.2. (a)]{BouLe2013}, $(ii)$ is  \cite[Theorem 3.2. (b)]{BouLe2013} and $(iii)$ is  \cite[Theorem 3.2. (c)]{BouLe2013}, rephrased to our setup. Also note that because we have a left Euclidean division algorithm for monic skew polynomials, every left principal ideal $I$ in $R/Rf$ will have a skew polynomial of smallest degree in $I$ that generates $I$.

\begin{theorem}\label{thm:newTheorem3.6}
Let $g=\sum_{i=0}^{r}g_it^i$ be a monic polynomial which is a right divisor of $f$, such that $f=gu=u'g$ for two monic
polynomials $u,u'\in S[t;\sigma,\delta]$. Let $C$ be the skew $(f,\sigma,\delta)$-polycyclic code corresponding to the principal left ideal in $\mathbb{S}_f$ generated by $g$ and let
$c(t)=\sum_{i=0}^{\mathbf{n}-1}c_it^i\in S[t;\sigma,\delta]$.
Then the following are equivalent:
\\ $(i)$ $(c_0,\dots,c_{\mathbf{n}-1})\in C$.
\\ $(ii)$ $c(t)\circ_f u(t)=0$ in $\mathbb{S}_f$.
\\ (iii) 
$R_u(c)=u(t) \circ_f c(t)=0$ in $\mathbb{S}_f$.
\end{theorem}

This  is  \cite[Theorem 3.6.]{BouLe2013}, rephrased to our setup.

This shows that sometimes $u(t)$ is a parity check polynomial for the skew $(f,\sigma,\delta)$-polycyclic code $C$
 also when $f$ is not two-sided. Note that parity check and generator polynomials will only be unique up to similarity in our setting. Also, when we only have
  $u(t)\circ_f g(t)=f(t)$, $u(t)$ monic, and $C$ is the code generated by $g(t)$ then if $c(t) \circ_f u(t)=0$ in $\mathbb{S}_f$, $\Phi^{-1}(c(t))$ is a codeword of $C$.

A matrix $M$ is called a \emph{control matrix} of the skew $(f,\sigma,\delta)$-polycyclic code $C$, if $M=lann(C)$ \cite[Definition 3.7.]{BouLe2013}.

\begin{corollary} \label{cor:newCor3.6}
Let $g=\sum_{i=0}^{r}g_it^i$ be a monic polynomial which is a right divisor of $f$, such that $f=gu=u'g$ for two monic
polynomials $u,u'\in \mathbb{S}_f$. Let $C$ be the skew $(f,\sigma,\delta)$-polycyclic code corresponding to $g$.
Then the matrix representing right multiplication $R_u$ with $u$ in $\mathbb{S}_f$ with respect to
 the basis $1,t,\dots,t^{\mathbf{n}-1}$
is a control matrix of the cyclic $(f,\sigma,\delta)$-code corresponding to $g$.
\end{corollary}

This  is  \cite[Corollary 3.8.]{BouLe2013}.

In particular, we re-obtain a classical result.

\begin{theorem}\label{thm:newTheorem3.2}
Let $f(t)=t^{\mathbf{n}}-\sum_{i=0}^{\mathbf{n}-1}a_it^i\in S[t;\sigma]$ and $g=\sum_{i=0}^{r}g_it^i\in R/Rf$  be a polynomial which is a monic right divisor of $f$.
Let $C\subset S^{\mathbf{n}}$ be the skew $(f,\sigma)$-polycyclic code consisting of all $(c_0,\dots,c_{\mathbf{n}-1})$ such that  $c(t)=\sum_{i=0}^{\mathbf{n}-1}c_it^i \in Rg/Rf$.
\\ $(i)$  If $(c_0,\dots,c_{\mathbf{n}-1})\in C$ then $L_{t^i}(c_0,\dots,c_{\mathbf{n}-1})\in C$ for all $1\leq i \leq n-1$.
\\ $(ii)$ Every left ideal generated by $g$ yields a skew $(f,\sigma)$-polycyclic code $C\subset S^{\mathbf{n}}$
\\ (iii)
The skew $(f, \sigma)$-polycyclic code $C\subset S^{\mathbf{n}}$ corresponding to the principal left ideal in $R/Rf$ generated by $g$
 is a free left $S$-module of dimension $\mathbf{n}-{\rm deg}(g)$.
\\ (iv) The matrix
\begin{equation}\label{eq10}
 G= \left(
 \begin{array}{cccccccc}
 g_{_0} &g_{_1} &\cdots&g_{_{\mathbf{n}-k}}& 0 &\cdots &\cdots & 0\\
 0 & \sigma(g_{_0}) & \sigma(g_{_1}) &\cdots & \sigma(g_{_{\mathbf{n}-k}}) & 0 &\cdots & 0\\
 \vdots &\ddots &\ddots &\ddots & &\ddots & &\vdots\\
 \vdots & &\ddots &\ddots &\ddots & &\ddots &\vdots\\
 0 &\ldots & &0 & \sigma^{k-1}(g_{_0}) & \sigma^{k-1}(g_{_1}) &\ldots & \sigma^{k-1}(g_{_{\mathbf{n}-k}})\\
 \end{array}
 \right)
 \end{equation}
  generating $C$ represents the right multiplication with $g$ in $R/Rf$,
calculated with respect to the basis $1,t,\dots,t^{\mathbf{n}-1}$, identifying $g=\sum_{i=0}^{r}g_it^i$
with the vector $(g_0,\dots,g_{r},0,\dots,0)$.
\end{theorem}

 Part $(i)$ was observed for semifields for the first time in \cite{LaSh}; while $(ii)-(iv)$ generalize \cite[Theorem 2.20, 2.21]{OaHA25} and \cite[Section 5]{BagPa2025}, where $f$ is assumed to be two-sided, to the nonassociative setting, by simply reformulating \cite[Theorem 3.2.]{BouLe2013} to the setting of nonassociative rings.
We give some direct proofs for the convenience of the reader, to show how the nonassociative multiplication comes into play.

\begin{proof}
Let $Rg/Rf  $ be the left ideal in $R/Rf$ generated by $g=\sum_{i=0}^r g_it^i$.
\\ $(i)$ If $c=(c_0,c_1,\dots,c_{\mathbf{n}-1}) \in C$, then by definition $\sum_{i=0}^{\mathbf{n}-1}c_it^i\in Rg/Rf $, and so also $t \circ_f \big(\sum_{i=0}^{\mathbf{n}-1}c_it^i\big)\in Rg/Rf$. We compute
$$t \circ_f \big(\sum_{i=0}^{\mathbf{n}-1}c_it^i\big)= \sum_{i=0}^{\mathbf{n}-1}t c_it^i  \mod_r f=\sigma(c_0)t+\sigma(c_1)t^2+\dots+ \sigma(c_{\mathbf{n}-2})t^{\mathbf{n}-1}+ \sigma(c_{\mathbf{n}-1})t^{\mathbf{n}} \mod_r f$$
$$ =  \sigma(c_{\mathbf{n}-1}) a_0 + (\sigma(c_0)+\sigma(c_{\mathbf{n}-1})a_1^t +(\sigma(c_1)+\sigma(c_{\mathbf{n}-2})a_2)t^2 +\ldots + (\sigma(c_{\mathbf{n}-2})+\sigma(c_{\mathbf{n}-1})a_{\mathbf{n}-1})t^{\mathbf{n}-1}.$$
 Hence
also
$$c=( \sigma(c_{\mathbf{n}-1}) a_0, \sigma(c_0)+\sigma(c_{\mathbf{n}-1})a_1,\sigma(c_1)+\sigma(c_{\mathbf{n}-2})a_2,\ldots, \sigma(c_{\mathbf{n}-2})+\sigma(c_{\mathbf{n}-1})a_{\mathbf{n}-1})) \in C.$$
 $(ii)$ is straightforward.
\\ $(iii, iv)$
Since $\mathbb{S}_f g$ is a left ideal, $t^i   \circ_f g\in \mathbb{S}_f g$ for all $i\in \{1,\dots,\mathbf{n}-1\}.$
It is easy to see that $g,t   g,t^2   g,\dots, t^{\mathbf{n}-r-1}   g$ are linearly independent over $S$, thus form a basis of $S^{\mathbf{n}-r}$. The rows of the matrix generating $C$ are thus given by the vectors corresponding to $g,tg,t^2g,\dots, t^{\mathbf{n}-r-1}g$. We obtain
the generator matrix $ G $ of $ C $ as
 \begin{equation}\label{eq10}
 G= \left(
 \begin{array}{cccccccc}
 g_{_0} &g_{_1} &\cdots&g_{_{m-k}}& 0 &\cdots &\cdots & 0\\
 0 & \sigma_1(g_{_0}) & \sigma_1(g_{_1}) &\cdots & \sigma_{1}(g_{_{\mathbf{n}-k}}) & 0 &\cdots & 0\\
 \vdots &\ddots &\ddots &\ddots & &\ddots & &\vdots\\
 \vdots & &\ddots &\ddots &\ddots & &\ddots &\vdots\\
 0 &\ldots & &0 & \sigma_{k-1}(g_{_0}) & \sigma_{k-1}(g_{_1}) &\ldots & \sigma_{k-1}(g_{_{\mathbf{n}-k}})\\
 \end{array}
 \right)
 \end{equation}
 where $k=\mathbf{n}-\deg(g) $.
 Moreover, we just showed that the skew $(\sigma,d)$-constacyclic code $C\subset S^n$ corresponding to the principal left ideal
$\mathbb{S}_f g $ is a free left $S$-module of dimension $\mathbf{n}-{\rm deg}(g)$.
\end{proof}

Since we assume that  $\sigma\in {\rm Aut}(S)$, note that a linear code $C$ is a left skew $(f,\sigma,\delta)$-polycyclic code if and only if the corresponding polynomials form a left ideal in $\mathbb{S}_f$. Then $C$ is a left skew $(f,\sigma,\delta)$-polycyclic code if and only if it corresponds to a right ideal in $\quad_f \mathbb{S}= S^{op}[t;\sigma^{-1},-\delta\circ \sigma^{-1}]/S^{op}[t;\sigma^{-1},-\delta\circ \sigma^{-1}]f^{op}$, with $f^{op}$ being $f$ but its coefficients written on the right. (This is \cite[Corollary]{BLMS2024} formulated using ideals in our terminology.)

 \section{Equivalence relations for skew polycyclic codes}

 \subsection{Background} \label{subsec:background}

The \emph{Hamming weight} of $c=(c_0,\dots,c_{\mathbf{n}-1})\in S^{\mathbf{n}}$ counts the number of nonzero components in $c$
and is defined as $w(c)=\sum_{r\not=0} n_r(c)$ where
$n_r(c)=\{ i\,|\, c_i=r \}$ counts the number of components that equals the element $r\in S$.  The \emph{Hamming weight} of a polynomial $p(t)\in R$ is defined as the number of nonzero coefficients of $p(t)$.

We want equivalent linear codes to  have the same performance parameters, i.e. the same length $\mathbf{n}$, dimension $k$, and minimum Hamming distance $d$.
When looking at the literature, there are several slightly different ways to define two linear codes with the Hamming metric to be ``equivalent''.
 The most general definition for commutative rings $S$ seems to be the following.

\begin{definition}\label{def:semil}
Two linear codes $C_1$ and $C_2$ over a ring $S$ of the same length $\mathbf{n}$ are called \emph{semilinear-equivalent}, if each element in $C_2$ can be obtained from some element $(c_0,\dots,c_{\mathbf{n}-1})$ in $C_1$ by a finite sequence of any of the following transformations:
\begin{enumerate}
        \item \label{perm} A permutation of the components of $(c_0,\dots,c_{\mathbf{n}-1})$.
        \item \label{mult} The multiplication of elements in a fixed position by a nonzero scalar in $S$.
        \item \label{iso} Applying an isomorphism  $\tau\in {\rm Aut}(S)$ to each component of the code vector $(c_0,\dots,c_{\mathbf{n}-1})$.
         \end{enumerate}
         Two linear codes $C_1$ and $C_2$ over a ring $S$ of the same length $\mathbf{n}$ are called \emph{linear-equivalent}, if each element in $C_2$ can be obtained from some element $(c_0,\dots,c_{\mathbf{n}-1})$ in $C_1$ by a finite sequence of any of (\ref{perm}) or (\ref{mult}).
\end{definition}

These definitions are  motivated by the fact that the above operations all preserve the Hamming distance and the linearity of a code.

More formally, two linear codes $C_1$ and $C_2$ over a ring $S$ of the same length $\mathbf{n}$ are semilinear-equivalent, if and only if
there exists an automorphism $\tau\in {\rm Aut}(S)$, and some $\gamma=(v,\pi)\in (S^{\times})^\mathbf{n} \rtimes S_\mathbf{n}$ such that $C_2=\{y\in S^{\mathbf{n}}\,|\, y=\gamma(\tau(c)) \text{ for }x\in C_1\}$.
Alternatively, if  $C_1$ and $C_2$ have generator matrices $G_1$ and $G_2$, then they are semilinear-equivalent, if there is an invertible $N\in Gl_n(S)$, some $\tau\in {\rm Aut}(S)$ and a
monomial matrix $M\in M_n(S)$ such that $G_2=N\tau(G_1 M)$. If $\tau=id$ then $C_1$ and $C_2$ are linear-equivalent.

Note that  the set of all $\mathbf{n} \times \mathbf{n}$ monomial matrices forms a group, and that every monomial matrix is the product of a permutation matrix $P$ and a diagonal matrix $diag( d_0,\dots,d_{\mathbf{n}-1}) $ with $d_i\in S^\times$ invertible.

Let  ${\bf C}_f$ be the class of skew $(f,\sigma,\delta)$-polycyclic codes and ${\bf C}_h$ the class of skew $(h,\sigma,\delta)$-polycyclic codes.

Chen et al. defined an equivalence relation on the invertible elements of $\F_q$  we will call Chen isometry in the following,  to classify the constacyclic codes over $\F_q$  of length $\mathbf{n}$ \cite{Chen2012}. Elements $a,b\in \F_q^\times$ in the same equivalence class define equivalent classes of   constacyclic codes  ${\bf C}_{t^{\mathbf{n}}-a}$ and ${\bf C}_{t^{\mathbf{n}}-b}$, and in turn guarantee the existence of an isometry between the two associative rings
$\F_q[t]/( t^{\mathbf{n}}-a)$, respectively, $\F_q[t]/( t^{\mathbf{n}}-b)$.

However, Chen isometry does not easily yield the generator polynomial of the code under the isometry.
 This lead to the study of Chen equivalence $\sim_\mathbf{n}$ to classify the constacyclic codes over $\F_q$ \cite{Chen2015} which is a special case of Chen isometry. It trivially allows to relate the generating polynomials of Chen equivalent constacyclic codes. Both definitions were recently generalised to skew  constacyclic and (skew) polycylic codes over finite fields and rings \cite{BBB21, CON2025, Oua2025, LM25, NevPum2025}.
For instance, \cite[Section 6]{LM25} uses the monomial isometries $G_{id,\alpha} $ of degree one, when the rings $R_a$ are associative.

  Another approach was taken in  \cite{BLMS2024}: The classification of $(f,\sigma,\delta)$-polycyclic codes that are Hamming isometrical equivalent
  defines the class of $(f,\sigma,\delta)$-polycyclic codes ${\bf C}_f$ and the class of $(h,\sigma,\delta)$-polycyclic codes ${\bf C}_h$ to be \emph{Hamming isometrical equivalent}, if there exists a monomial matrix $B$ such that $A_{f_1}B=\sigma(B) A_{f_2} +\delta(B)$.
 This definition implies that $f_1\sim f_2$ are similar skew polynomials.
 It is straightforward to see that if two classes of codes are Chen isometric, then they are Hamming isometrical equivalent.

\subsection{Our approach}

Let $S$ be a unital associative ring, $\sigma\in {\rm Aut}(S)$,  $S_0$ the ring fixed under $\sigma$. We allow that $\sigma=id$ and $\delta=0$ unless stated otherwise. Let $f,h\in R=S[t;\sigma,\delta]$ be monic polynomials of degree $\mathbf{n}$.

Let us define those isomorphisms of Petit rings that corresponds to Hamming-weight preserving homomorphisms of skew polycyclic codes.

\begin{definition}\label{D:Gtaualpha}
      Suppose $\tau\in \Aut(S)$, $\alpha \in S^\times$, and $k\in \mathbb{N}$.
    If there exists a  ring isomorphism
    $$
    G: R/Rf \to R/Rh
    $$
    between nonassociative rings,
    defined via $G|_{S}=\tau$ and $G(t)=\alpha t^k$,
    that is,
\begin{equation}\label{E:formulaG(t)2}
G \left( \sum_{i=0}^{\mathbf{n}-1}c_it^i \right) = \sum_{i=0}^{\mathbf{n}-1}\tau(c_i) (\alpha t^k)^i.
\end{equation}
     then we call $G$ an \emph{isometry} or a \emph{monomial isomorphism  of degree $k$}, if we want to clarify the degree $k$, and write $G=G_{\tau,\alpha,k}$. \\
      When $k=1$, we write $G_{\tau,\alpha}$ in place of $G_{\tau,\alpha,1}$ and call $G_{\tau,\alpha}$ an \emph{equivalence}  or a \emph{monomial isomorphism  of degree one}.

      A ring isomorphism of the type
    $
    G_{id,\alpha,k}: R/Rf \to R/Rh
    $
    is called a  \emph{Chen isometry}, and $G_{id,\alpha}$ is called a \emph{Chen equivalence}.
    \\ If $\mathbf{n}$ is not clear from our context ($\mathbf{n}$ corresponds to the length of the code), we will sometimes use the terms \emph{$\mathbf{n}$-isometry}, \emph{$\mathbf{n}$-equivalence}, \emph{$\mathbf{n}$-Chen isometry},
     and \emph{$\mathbf{n}$-Chen equivalence} instead.
      \end{definition}

     We note that the ring isomorphisms $G_{\tau,\alpha,k}:R/Rf\rightarrow R/Rh$
 for  $\alpha\in S^\times$, are $S'$-algebra isomorphisms, for every subring $S'\subset {\rm Fix}(\tau)\cap {\rm Const}(\delta)\cap {\rm Fix}(\sigma)\cap C(S)$.

Up to now, only Chen isometries $G_{id,\alpha,k}$ have been used in the literature for (skew) polycyclic codes, cf. Section \ref{subsec:background}.
Their disadvantage is  they do not include property  (\ref{iso})  when using them to  define the equivalence of (skew) polycyclic codes.
They have also mistakenly sometimes be defined as being $S$-algebra isomorphisms, which is never the case unless $\sigma=id$ and $\delta=0$, that is when $R=S[t]$
(e.g. \cite{Oua2025}.

We know that the above defined Hamming weight preserving ring isomorphisms $G_{\tau,\alpha,k}$ are the only ones there are in many cases, see Section \ref{sec:iso}.
This makes this approach the most general one possible.

\begin{definition}
The classes ${\bf C}_f$  and ${\bf C}_h$  are called
\\ $(i)$  \emph{isometric}, if there exists an isometry
$G_{\tau,\alpha,k}:R/Rf\rightarrow R/Rh,$
 \\ $(ii)$ \emph{Chen isometric},  if there exists an isometry
 $G_{id,\alpha,k}:R/Rf\rightarrow R/Rh,$
\\ (iii)  \emph{equivalent},  if there exists an equivalence
$G_{\tau,\alpha}:R/Rf\rightarrow R/Rh$
\\ (iv) \emph{Chen equivalent},  if there exists a Chen equivalence
$G_{id,\alpha}:R/Rf\rightarrow R/Rh.$
\\ If the length of the codes is not clear from the context, we will use the terminology $\mathbf{n}$-isometric, $\mathbf{n}$-Chen isometric etc. instead.
\end{definition}

If two classes ${\bf C}_f$ and ${\bf C}_h$ are  isometric (in particular, equivalent), then $f$ and $h$ have the same degree $\mathbf{n}$, so the codes in ${\bf C}_f$ and ${\bf C}_h$ are $S$-submodules of $S^{\mathbf{n}}$. For arbitrary rings, not every submodule of $S^{\mathbf{n}}$ will be free. However,
 since every skew $(f,\sigma,\delta)$-polycyclic code $C$ in ${\bf C}_f$ is generated by a monic $g\in R$ that is a right divisor of $f$  (Theorem \ref{thm:newTheorem3.2} $(i)$), every skew $(f,\sigma,\delta)$-polycyclic code $C$ in ${\bf C}_f$ corresponds to a free (and cyclic) $S$-submodule $Rg/Rf$. We also know that $C$ has dimension $k=\mathbf{n}-\deg(g) $ (Theorem \ref{thm:newTheorem3.2} $(i)$).

   \begin{lemma}\label{le:pres}
   If two classes ${\bf C}_f$ and ${\bf C}_h$ are   $\mathbf{n}$-isometric via some monomial isomorphism $G_{\tau,\alpha,\ell}$ of degree $l$, then the skew $(f,\sigma,\delta)$-polycyclic codes in ${\bf C}_f$ have the same Hamming distance, dimension and length as
  the  skew $(f,\sigma,\delta)$-polycyclic code in ${\bf C}_h$.
   \end{lemma}

\begin{proof}
By assumption, there exists an isomorphism of nonassociative ambient Petit rings $G_{\tau,\alpha,\ell}:R/f\to R/Rh$. This immediately implies that the length of the codes in both classes is $\mathbf{n}$, and so the ring isomorphism $G_{\tau,\alpha,\ell}$ preserves the length of codes. For $f=ug$ we have $G_{\tau,\alpha,\ell}(f)=G_{\tau,\alpha,l}(u)G_{\tau,\alpha,\ell}(g)$. By standard ring theory
  every skew $(f,\sigma,\delta)$-polycyclic code $C$ in ${\bf C}_f$ generated by some monic right divisor $g$ of $f$ is in one-to-one correspondence with
  a skew $(f,\sigma,\delta)$-polycyclic code in ${\bf C}_h$  generated by the right divisor  $G_{\tau,\alpha,\ell}(g)$ of $h$.  The skew polynomial $G_{\tau,\alpha,\ell}(g)$ need not be monic, but its highest coefficient lies in $S^\times$ since $\alpha\in S^\times$, so we can assume w.o.l.o.g. that $G_{\tau,\alpha,\ell}(g)$ is monic after scaling, as this will not change the left ideal it is generating, nor the fact that is is a right divisor of $h$.

   If $\deg{g}=r$ then the skew $(f,\sigma,\delta)$-polycyclic code $C$ in ${\bf C}_f$ associated with the left ideal generated by $g$ has dimension $k=\mathbf{n}-\deg(g)$, and the isometric skew $(h,\sigma,\delta)$-polycyclic code in ${\bf C}_f$ associated with the left ideal generated by $G_{\tau,\alpha,\ell}(g)$ has dimension
   $$k'=\mathbf{n}-\deg(G_{\tau,\alpha,\ell}(g) \mod_r h).$$

   When $\ell=1$ it is clear that the ring isomorphism $G_{\tau,\alpha}$ thus preserves the dimension $k$ of  codes that corresponds to each other, as in this case
   $\deg(g)=\deg (G_{\tau,\alpha}(g))$.

   When $\ell>1$ we observe that the ring isomorphism $G_{\tau,\alpha,\ell}$ is a $\tau$-semilinear bijective map between the free $S$-modules $R/Rf$ and $R/Rh$, and it restricts to a  $\tau$-semilinear bijective map between the free $S$-modules $Rg/Rf$ and $G_{\tau,\alpha,\ell}(g)/Rh$. Therefore both must have the same dimension and we obtain that $\mathbf{n}-\deg (g)=\mathbf{n}-\deg(G_{\tau,\alpha,\ell}(g) \mod_r h)$, that is again $k'=k$.
\end{proof}

\begin{remark}
When $S$ is a finite ring we can also argue as follows to get some intuition on why the dimension will be by preserved also under isometries and not just equivalences (indeed, we could not find any proof of this in the current literature, which mostly deals with finite rings). Every ring isomorphism $G_{\tau,\alpha,\ell}$ canonically yields a one-one correspondence between the left ideal $Rg/Rf$ and the left ideal $RG_{\tau,\alpha,\ell}(g)/Rh$, which means that $|Rg/Rf|=|R G_{\tau,\alpha,\ell}(g)/Rh|$ have the same number of elements. We know that both are free submodules of $S^{\mathbf{n}}$, that means $|Rg/Rf|=|S|^k$ and $|G_{\tau,\alpha,\ell}(g)/Rh|=|S|^{k'}$ for some positive integers $k$ and $k'$, so that $k=k'$ and both must have the same dimension. In particular, since $k=\mathbf{n}-\deg (g)$ and $k'=\mathbf{n}-\deg(G_{\tau,\alpha,\ell}(g) \mod_r h)$ this implies that we must have $\deg(G_{\tau,\alpha,\ell}(g) \mod_r h)=\deg(g)$.
\end{remark}

Isometric and Chen isometric classes of codes ${\bf C}_f$  and ${\bf C}_h$ contain codes with the same length and
  Hamming weight, as well as dimensions determined by the degree of $G_{\tau,\alpha}(g)\mod_r f$. Equivalent and  Chen equivalent  classes of codes ${\bf C}_f$  and ${\bf C}_h$ contain codes with the same
parameters, thus have the same performance:
 for every skew polycyclic code $C_1$ in ${\bf C}_f$ we can find a skew polycyclic code $ C_2 $ in ${\bf C}_h$ that has the same length,
  Hamming weight and dimension.  One can be obtained from the other by a finite sequence of any of operations (\ref{perm}),  (\ref{mult}) or (\ref{iso}) specified in Definition \ref{def:semil}.

We will show that in the majority of cases, the generator  of a code in ${\bf C}_f$ can be explicitly matched with  one of the same degree of a code in an equivalent class ${\bf C}_h$.
This is because of the scarcity of monomial isomorphisms of degree $k>1$.

\begin{remark}
If ${\bf C}_f$  and ${\bf C}_h$ are Chen equivalent, then they are also equivalent.
\\
If  ${\bf C}_f$  and ${\bf C}_h$ are Chen isometric, then they are also isometric.
\\
If  ${\bf C}_f$  and ${\bf C}_h$ are equivalent, then they are also isometric.
\\
If ${\bf C}_f$  and ${\bf C}_h$ are Chen equivalent, then they are also equivalent and Chen isometric, hence isometric.
\end{remark}

To motivate these notions we stress that the classifications we will achieve are tighter when employing the notions of equivalence and isometry (Section \ref{s:finite}). In the current literature, only Chen equivalence and Chen isometry have been employed.

\subsection{Equivalence relations on $S^{\mathbf{n}}$}

Let $S$ be a unital commutative associative ring, $\sigma\in {\rm Aut}(S)$,  $S_0$ the ring fixed under $\sigma$. We allow that $\sigma=id$ and $\delta=0$ unless stated otherwise.
Inspired by \cite{OaHA25, OaHA25II} we define equivalence relations on  $(S^\times)^m$ via the following equivalence relations:

\begin{definition} \label{def:equ}
Let $$f(t) = t^{\mathbf{n}}-\sum_{i=0}^{\mathbf{n}-1} a_i t^i ,\quad h(t) = t^{\mathbf{n}}-
\sum_{i=0}^{\mathbf{n}-1} b_i t^i \in S[t;\sigma]\in R.$$
We call $(a_0,\dots,a_{\mathbf{n}-1})$ \emph{isometric} to $(b_0,\dots,b_{\mathbf{n}-1})$, written
$$(a_0,\dots,a_{\mathbf{n}-1})\cong_\mathbf{n} (b_0,\dots,b_{\mathbf{n}-1}),$$
if there exists an isometry $G_{\tau,\alpha,k}:R/Rf\rightarrow R/Rh,$
and \emph{Chen isometric}, if $\tau=id$.
\\
We call $(a_0,\dots,a_{\mathbf{n}-1})$ \emph{equivalent} to $(b_0,\dots,b_{\mathbf{n}-1})$, written
$$(a_0,\dots, a_{\mathbf{n}-1})\sim_\mathbf{n} (b_0,\dots,b_{\mathbf{n}-1}),$$
if there exists an equivalence $G_{\tau,\alpha}:R/Rf\rightarrow R/Rh$,
and \emph{Chen equivalent}, written
$$(a_0,\dots,a_{\mathbf{n}-1})\sim_{Chen, \mathbf{n}} (b_0,\dots,b_{\mathbf{n}-1})$$
if $\tau=id$.
\end{definition}

The above equivalence relations  partition $(S^\times)^\mathbf{n}$ into different equivalence classes.
It is also obvious that
$$(a_0,\dots, a_{\mathbf{n}-1})\cong_\mathbf{n} (b_0,\dots,b_{\mathbf{n}-1})$$
if and only if ${\bf C}_f$  and ${\bf C}_h$ are isometric,
$$(a_0,\dots, a_{\mathbf{n}-1})\cong_{Chen, \mathbf{n}} (b_0,\dots,b_{\mathbf{n}-1})$$
if and only if ${\bf C}_f$  and ${\bf C}_h$ are Chen isometric,
$$(a_0,\dots, a_{\mathbf{n}-1})\sim_\mathbf{n} (b_0,\dots,b_{\mathbf{n}-1})$$
if and only if ${\bf C}_f$  and ${\bf C}_h$ are equivalent, and
$$(a_0,\dots, a_{\mathbf{n}-1})\sim_{Chen, \mathbf{n}} (b_0,\dots,b_{\mathbf{n}-1})$$
if and only if ${\bf C}_f$  and ${\bf C}_h$ are Chen equivalent

\begin{remark}
 The equivalence relation $\sim_\mathbf{n}$ that partitions $(S^\times)^\mathbf{n}$ into $\mathbf{n}$-isometry classes, respectively ${\bf C}_f$  and ${\bf C}_h$ into equivalent classes of codes, is finer than the Chen equivalence and Chen isometry  in
\cite{Chen2012, OaHA25, OaHA25II}, for instance see \cite[Example 5.8]{NevPum2025}.
\end{remark}

\begin{remark} \label{r:rank}
In \cite[Definition 4.1]{OaHA25}, for $K=\mathbb{F}_q$, two skew polynomials $f_1,f_2\in K[t;\sigma]$ are called \emph{Hamming equivalent} (resp. \emph{rank equivalent}), if there is a $K$-vector space isomorphism $\varphi: K[t;\sigma]/K[t;\sigma]f_1\to K[t;\sigma]/K[t;\sigma]f_2$ such that
\begin{enumerate}
    \item $\varphi(g h)=\varphi(g)\varphi(h)$ for all $g,h\in K[t;\sigma]/K[t;\sigma]f_1$

    \item $\varphi$ preserves the Hamming distance $d_H(g,h)=d_H(\varphi(g),\varphi(h)) $
    (respectively,  the rank distance $d_R(g,h)=d_R(\varphi(g),\varphi(h)) $).
\end{enumerate}
Here, the rank distance is defined over arbitrary subfields of $K$.

We can now read this definition as two skew polynomials $f_1,f_2\in K[t;\sigma]$ being defined as Hamming or rank equivalent, if there exist a Chen isometry $G_{id,\alpha,k}: K[t;\sigma]/K[t;\sigma]f_1\to K[t;\sigma]/K[t;\sigma]f_2$.
Thus, the above mentioned - and for $f_1$ not two-sided still undefined -multiplication must be the one in the nonassociative Petit ring $R/Rf_1$, that is viewed as algebra over some subfield of $K$ in the rank metric case.

With the results outlined in this paper, we can say that ``in most cases'' two skew polynomials $f_1,f_2\in K[t;\sigma]$ are Hamming or rank equivalent, if there exist a Chen equivalence $G_{id,\alpha }: K[t;\sigma]/K[t;\sigma]f_1\to K[t;\sigma]/K[t;\sigma]f_2$, where we require that $\alpha\in K'^\times$ lies in a suitable subfield of $K$ when we use the rank metric. In most cases this is the case when $f_1$ is not two-sided and the degree of $f_1$ is greater than ${\mathbf{n}}-1$, or when both $f_1(t)=t^{\mathbf{n}}-a$  and $f_2(t)=t^{\mathbf{n}}-b$ are not two-sided. The decision in \cite{OaHA25} to focus on the case that $\varphi(t)=\alpha t$ is thus well founded. However, to instead also consider the $K'$-algebra isomorphisms obtained via $G_{\tau,\alpha}$ for all $\tau\in {\rm Gal}(K/K')$ should yield tighter classifications here as well.
\end{remark}

\section{Isometries between Petit rings $S[t;\sigma]/S[t;\sigma]f$}\label{sec:iso}

From now on, let $S$ be a unital commutative associative  ring and suppose $\sigma\in \Aut(S)$ has finite  order $\mathbf{m}$.
 Write $S_0=\{s\in S: \sigma(s)=s\}$ for the subring fixed by $\sigma$.
If $\tau\in \Aut(S)$ is an arbitrary automorphism
then we write $S_0^\tau=\{s\in S_0: \tau(s)=s\}$ for the common fixed subring under $\tau$ and $\sigma$.

  For any $i\in \mathbb{N}$, $\tau\in \Aut(S)$ and $\beta\in S$, we define
\begin{equation}\label{E:Normrelation0}
N_i^{\tau}(\beta) := \prod_{j=0}^{i-1}\tau^j(\beta).
\end{equation}
Note that $\beta \tau(N_i^\tau(\beta))=N_{i+1}^\tau(\beta)$ and thus
\begin{equation}\label{E:Normrelation}
N_{i+j}^\tau(\beta)= N_{i}^\tau(\beta) \cdot \tau^{i}(N_j^\tau(\beta))
\end{equation}
for all $i,j\geq 0$.  If $S=K$ is a field, $K/K_0$ a cyclic Galois extension, and $\Gal(K/K_0)=\langle \sigma \rangle$, where $\sigma$ has  order $\mathbf{m}$, then  $N_{\mathbf{m}}^\sigma$ is simply the norm  $N_{K/K_0}$ of the field extension $K/K_0$.  More generally, we have $N_{\mathbf{m}}^\sigma(\beta)\in S_0$ for all $\beta\in S$ so we distinguish this case with the notation
$N_{S/S_0}:=N_{\mathbf{m}}^\sigma$.

\subsection{Equivalences}
  Let
 $$f(t) = t^{\mathbf{n}}-\sum_{i=0}^{\mathbf{n}-1} a_i t^i ,\quad h(t) = t^{\mathbf{n}}-
\sum_{i=0}^{\mathbf{n}-1} b_i t^i \in S[t;\sigma]$$
be two monic polynomials. Employing our  notation (\ref{E:Normrelation0}), we  write $  G_{\tau,\alpha}: R/Rf\to R/Rh  $ as
$$
G_{\tau,\alpha}\left(\sum_{i=0}^{\mathbf{n}-1}d_it^i\right) = \sum_{i=0}^{\mathbf{n}-1}\tau(d_i)N_i^\sigma(\alpha)t^i.
$$
Special cases of the following results were originally proved in \cite[Theorems 28, 29]{BP}, for the case that $S$ is a field and only for $S_0$-isomorphisms.
We require more general results.

\begin{theorem}  \label{general_isomorphism_theorem1}
Let $\alpha\in S^\times$ and assume that $\tau \in {\rm Aut}(S)$ commutes with $\sigma$. Then
   $$  G_{\tau,\alpha}: R/Rf\to R/Rh  $$
    is an isomorphism between nonassociative unital rings that restricts to the identity on  $S_0^\tau$,
     if and only if
\begin{equation}  \label{equ:necII}
\tau(a_i) = N_{\mathbf{n}-i}^{\sigma}(\sigma^{i} (\alpha))b_i
\end{equation}
 for all $i \in\{ 0, \ldots, \mathbf{n}-1\}$.
\end{theorem}

\begin{proof}
Let $G= G_{\tau,\alpha}: R/Rf\to R/Rh$,
$$G_{\tau , \alpha}(\sum_{i=0}^{\mathbf{n}-1} d_i t^i )=  \sum_{i=0}^{\mathbf{n}-1} \tau(d_i)  N_i^\sigma(\alpha) t^i$$
 be such a ring isomorphism. i.e. an isomorphism  of $S_0^\tau$-algebras, with $\alpha \in S^{\times}$. Then we have
$G(z t^i) = G(z)G(t)^i = \tau(z) (\alpha t)^i = \tau(z)  N_i^\sigma(\alpha) t^i$
 for all $i \in \{ 1, \ldots, \mathbf{n}-1 \}$ and $z \in S$.

 Moreover, with $t^{\mathbf{n}}=t t^{\mathbf{n}-1}$ and using that $t^{\mathbf{n}}=\sum_{i=0}^{\mathbf{n}-1} a_i t^i $ in $R/Rf$, also
\begin{equation} \label{eqn:automorphism_necessity_theorem 8}
G(t^{\mathbf{n}}) = G \Big( \sum_{i=0}^{\mathbf{n}-1} a_i t^i \Big) = \sum_{i=0}^{\mathbf{n}-1}
G(a_i) G(t)^i
= \sum_{i=0}^{\mathbf{n}-1} \tau(a_i) \big( \prod_{l=0}^{i-1} \sigma^l(\alpha) \big) t^i=\sum_{i=0}^{\mathbf{n}-1} \tau(a_i)N_i^\sigma(\alpha)  t^i
\end{equation}
and  $G(t t^{\mathbf{n}-1})=G(t)G(t^{\mathbf{n}-1})=G(t)G(t)^{\mathbf{n}-1}$, i.e.
\begin{equation} \label{eqn:automorphism_necessity_theorem 9}
G(t)^\mathbf{n} =G(t)G(t)^{\mathbf{n}-1}= \alpha \sigma(\alpha) \cdots \sigma^{\mathbf{n}-1}(\alpha) t^{\mathbf{n}} =
N^\sigma_{\mathbf{n}}(\alpha)\sum_{i=0}^{\mathbf{n}-1} b_i t^i.
\end{equation}
Comparing \eqref{eqn:automorphism_necessity_theorem 8} and
\eqref{eqn:automorphism_necessity_theorem 9} gives $  \tau(a_i) =
N_{\mathbf{n}-i}^{\sigma}(\sigma^{i} (\alpha))b_i $ for all $i \in\{ 0,
\ldots, \mathbf{n}-1\}$. Thus  Equation (\ref{equ:necII}) holds
for all $i \in\{ 0, \ldots, \mathbf{n}-1 \}$.

Conversely,  if Equation (\ref{equ:necII}) holds
for all $i \in\{ 0, \ldots, \mathbf{n}-1 \}$ and some $\alpha\in S^\times$, then consider the map 
$G_{\tau,\alpha}:R\rightarrow R$, $G(\sum_{i=0}^{\ell}d_it^i)=\sum_{i=0}^{\ell}\tau(d_i) (\alpha t)^i$.
 We look at the condition given in Lemma \ref{le:inducediso}, where now $\delta=0$, and $S$ is commutative, so that the two conditions for the map $G_{\tau,\alpha}$ to be a ring isomorphism collapse to the one condition that
$\sigma(\tau(b))=\tau(\sigma(b))$ is satisfied
for all $b\in S$, that means $\sigma$ and $\tau$ must commute. This is also shown in  \cite[Theorem 3]{Ri}. Due to our assumptions therefore we know that $G_{\tau,\alpha}:R\to R$ is a ring isomorphism.

By the proof of Lemma \ref{le:inducediso}, the monomial isomorphism $G_{\tau,\alpha}:R \to R$ of degree one canonically induces a ring isomorphism between the nonassociative rings $R/Rf$ and $R/RG_{\tau,\alpha}(f)$, and hence
an algebra isomorphism between the Petit rings $R/Rf$ and $R/RG_{\tau,\alpha}(f)$ which are algebras over $S_0^\tau$.
Since
$$
G_{\tau,\alpha}(f)=N_{\mathbf{n}}^\sigma(\alpha)t^{\mathbf{n}}-\sum_{i=0}^{\mathbf{n}-1} \tau(a_i)N_i^\sigma(\alpha)  t^i
$$
$$=N_{\mathbf{n}}^\sigma(\alpha)( t^{\mathbf{n}}-\sum_{i=0}^{\mathbf{n}-1} \tau(a_i) N_{\mathbf{n}}^\sigma(\alpha^{-1}) N_i^\sigma(\alpha)  t^i  )
=N_{\mathbf{n}}^\sigma(\alpha)( t^{\mathbf{n}}-\sum_{i=0}^{\mathbf{n}-1} \tau(a_i) N_{\mathbf{n}-i}^{\sigma}(\sigma^{i} (\alpha^{-1}))  t^i )$$
$$=N_{\mathbf{n}}^\sigma(\alpha)( t^{\mathbf{n}}-\sum_{i=0}^{\mathbf{n}-1}b_it^i )$$
in $R$, we obtain again $\tau(a_i) =
N_{\mathbf{n}-i}^{\sigma}(\sigma^{i} (\alpha))b_i$
 for all $i \in\{ 0, \ldots, \mathbf{n}-1\}$ and thus the assertion.
\end{proof}

When $S=K$ is a field, $K[t;\sigma]/K[t;\sigma]f$ and $K[t;\sigma]/K[t;\sigma]h$ are proper nonassociative rings, and $\Aut(K)$ is abelian, we can say even more.

\begin{theorem}  \label{general_isomorphism_theorem2}
  Let $\mathbf{m} \geq \mathbf{n}-1$ and assume that $f$ and $h$  do not generate two-sided ideals in $K[t;\sigma]$.  Suppose $\Aut(K)$ is abelian.  For any subfield $F\subset K_0={\rm Fix}(\sigma)$, the only nonassociative ring isomorphisms 
  $$
  G:K[t;\sigma]/K[t;\sigma]f\to K[t;\sigma]/K[t;\sigma]h
  $$
  that restrict to the identity on $F$
  are of the form $G=G_{\tau,\alpha}$ such that $\tau\in \Aut(K)$, $K_0^\tau\subset F$, and $\alpha \in K^\times$ satisfying Equation \ref{equ:necII}
\end{theorem}

\begin{proof}
 Let $G:K[t;\sigma]/K[t;\sigma]f\to K[t;\sigma]/K[t;\sigma]h$ be an isomorphism as in our assumption, that is, an  $F$-algebra isomorphism. Since $K[t;\sigma]/K[t;\sigma]f$ is not
associative, its left nucleus is $K$. Since any isomorphism  preserves the left nucleus,  $G(K) = K$ and so $G \vert_K = \tau$ for some
 $\tau \in \text{Aut}_F(K)$. Suppose $G(t) = \sum_{i=0}^{\mathbf{n}-1} k_i t^i$ for some $k_i \in K$.
Then we have
\begin{equation} \label{eqn:automorphism_necessity_theorem 1}
G(tz) = G(t)G(z) = (\sum_{i=0}^{\mathbf{n}-1} k_i t^i ) \tau(z) =
\sum_{i=0}^{\mathbf{n}-1} k_i \sigma^{i}(\tau(z)) t^i \end{equation} and
\begin{equation} \label{eqn:automorphism_necessity_theorem 2}
G(tz) =
G(\sigma(z)t) = \sum_{i=0}^{\mathbf{n}-1} \tau(\sigma(z)) k_i t^i
\end{equation} for all $z \in K$. Comparing the coefficients of $t^i$
in \eqref{eqn:automorphism_necessity_theorem 1} and
\eqref{eqn:automorphism_necessity_theorem 2} we obtain
\begin{equation} \label{eqn:automorphism_necessity_theorem 3} k_i \sigma^{i}(\tau(z)) =
k_i \tau(\sigma^i(z))= \tau(\sigma(z)) k_i=k_i\tau(\sigma(z))
\end{equation} for all $i \in\{ 0, \ldots, \mathbf{n}-1\}$ and all $z \in K$. This implies
$k_i (\tau\big( \sigma^i(z) - \sigma(z) \big) )=0$
for all $i \in\{ 0, \ldots, \mathbf{n}-1 \}$ and all $z \in K$ since $\sigma$
and $\tau$ commute, i.e.
\begin{equation} \label{eqn:automorphism_necessity_theorem 4}
k_i=0 \text{ or } \sigma^{i}(z)=\sigma(z)
\end{equation}
for all $i
\in\{ 0, \ldots, \mathbf{n}-1\}$ and all $z \in K$.

Since $\sigma$ has order $\mathbf{m}\geq \mathbf{n}-1$, which means $\sigma^i\not=\sigma$ for
all $i\in\{ 0, \ldots, \mathbf{n}-1\}$, $i\not=1$,
\eqref{eqn:automorphism_necessity_theorem 4} implies $k_i = 0$ for
all $ i \in\{ 0, \ldots, \mathbf{n}-1\}$, $i\not=1$. Therefore $G(t) = \alpha t$ for some
$\alpha\in K^{\times}$. As in the proof of Theorem \ref{general_isomorphism_theorem1}, these maps are indeed well defined algebra isomorphisms and must be the only one which implies the assertion.
\end{proof}

 \begin{proposition}  \label{prop:Chen equivpolycyclic}
 There exists a nonassociative ring isomorphism $G_{\tau,\alpha}:\mathbb{S}_f\rightarrow \mathbb{S}_h$ if and only if there exist algebra isomorphisms
 $G_{\tau,\sigma^i(\alpha)}:\mathbb{S}_{t^{\mathbf{n}-i}-a_i}\rightarrow \mathbb{S}_{t^{\mathbf{n}-i}-b_i}$  for all $i \in\{ 0, \ldots, \mathbf{n}-1\}$.
 \end{proposition}

\begin{proof}
We know that $G_{\tau,\alpha}$ is a ring isomorphism if and only if
$$\tau(a_i) = N_{\mathbf{n}-i}^\sigma ( \sigma^i(\alpha) ) b_i$$
 for all $i \in\{ 0, \ldots, \mathbf{n}-1\}$.  Put $\alpha_i=\sigma^i(\alpha)$,
 then the above equations are the same as saying that there exist $\alpha_i \in S^{\times}$ such that
$$\tau(a_i) = N_{\mathbf{n}-i}^\sigma ( \alpha_i ) b_i$$
which is the same as saying that  the algebras $\mathbb{S}_{t^{\mathbf{n}-i}-a_i}$ and $\mathbb{S}_{t^{\mathbf{n}-i}-b_i}$ are isomorphic with the isomorphisms given by $G_{\tau,\alpha_i}$.
\end{proof}

\begin{theorem}  \label{t:cycle}
Suppose  $\tau\in {\rm Aut}(S)$ commutes with $\sigma$ and define
$$f^\tau(t)= t^{\mathbf{n}}-\sum_{i=0}^{\mathbf{n}-1}\tau(a_i) t^i.$$
 Then $$G_{\tau,\alpha}:R/Rf\to R/Rh$$
  is a nonassociative ring isomorphism if and only if
 $$G_{id,\alpha}: R/Rf^\tau\to R/Rh$$
  is a nonassociative  ring isomorphism.
  \end{theorem}

  \begin{proof}
  By Theorem \ref{general_isomorphism_theorem1},
$  G_{\tau,\alpha}: R/Rf\to R/Rh  $
    is a nonassociative ring isomorphism (and an isomorphism of $S_0^\tau$-algebras) if and only if
$\tau(a_i) =
N_{\mathbf{n}-i}^{\sigma}(\sigma^{i} (\alpha))b_i$
 for all $i \in\{ 0, \ldots, \mathbf{n}-1\}$.
 But this is equivalent to
 $  G_{\tau,\alpha}: R/Rf^\tau\to R/Rh  $ being a ring isomorphism.
  \end{proof}

\subsection{Isometries}

While we have an excellent understanding about the monomial isomorphisms $ G_{\tau,\alpha,1}= G_{\tau,\alpha} $ which map $t$ to $\alpha t$, it may happen that $ G(t) $ is a polynomial of degree $k>1$ for some isomorphism $G$ between two nonassociative rings $R/Rf$ and $R/Rh$. We are only interested in Hamming weight preserving isomorphism, so we are concerned with the cases when $G(t)=\alpha t^k$ for some $k>1$ and $\alpha\in S^\times$ invertible, i.e. when $G=G_{\tau,\alpha,k}$ is a monomial of degree $k$.

Note that  when $S=K$ is a field, $\Aut(K)$ is abelian, and $f$ and $h$  do not generate two-sided ideals in $K[t;\sigma]$, by Theorem \ref{general_isomorphism_theorem2} we already know when $\sigma$ has  order $\mathbf{m}$, this case can only ever occur when $$\mathbf{m}< \mathbf{n}-1$$
  (when $S$ is a ring, this is not necessarily true).

  Isometries with $k>1$ are the hardest to understand. We collect some partial results here to make our point.
  Currently, we only know when these isometries exist for $f(t)=t^{\mathbf{n}}-a$ and $h(t)=t^{\mathbf{n}}-b$. In this special case, moniomial isometries of degree $k>1$ only exist when the ambient rings are associative. We continue to assume that $\sigma$ has  order $\mathbf{m}$ throughout.

  \begin{theorem}\label{thm:G}\cite[Theorem 4.4]{NevPum2025}
    Suppose $\tau\in \Aut(S)$ commutes with $\sigma$ and $\alpha \in S$ is cancellable.  Then for any $2\leq k < \mathbf{n}$, the map
    $$
    G_{\tau,\alpha,k}: S[t;\sigma]/S[t;\sigma](t^{\mathbf{n}}-a)\to S[t;\sigma]/S[t;\sigma](t^{\mathbf{n}}-b)
    $$
    is a homomorphism of rings if and only if all of the following conditions hold:
    \begin{enumerate} \setlength{\itemsep}{0pt}
        \item \label{C:k=1} $k\equiv 1 \mod \mathbf{m}$;
        \item \label{C:n|m} $\mathbf{m} \mid \mathbf{n}$;
        \item \label{C:inF} $a, b\in S_0$;
        \item \label{C:norm} $(N_{S/S_0}(\alpha))^{\mathbf{n}/\mathbf{m}}b^k=\tau(a)$.
    \end{enumerate}
    It is an isomorphism if and only if, in addition, we have
    \begin{enumerate} \setcounter{enumi}{4}
        \item \label{C:relprime} $\gcd(k,\mathbf{n})=1$ and $\alpha\in S^\times$.
 \end{enumerate}
\end{theorem}

   Recall that we are working in the general case where
 $$f(t) = t^{\mathbf{n}}-\sum_{i=0}^{\mathbf{n}-1} a_i t^i ,\quad h(t) = t^{\mathbf{n}}-
\sum_{i=0}^{\mathbf{n}-1} b_i t^i \in S[t;\sigma].$$

\begin{lemma} \label{l:5.5}
  Let $\sigma\not=id$. Let $\tau\in \Aut(S)$ commute with $\sigma$, let $k\in \mathbb{N}$, $1<k<\mathbf{n}-1$,  and $\alpha\in S^\times$.  If the map  $G_{\tau,\alpha,k}$ defined via $G(t)=\alpha t^k$ and $G(c)=\tau(c)$ for all $c\in S$, defines a ring isomorphism between the nonassociative rings $R/Rf$ and $R/Rh$, such that $G_{S_0^\tau}=id$, then $k\equiv 1 \mod \mathbf{m}$.
\end{lemma}

\begin{proof}
Take $G=G_{\tau,\alpha,k}:S[t;\sigma]\to S[t;\sigma]$, so that
$$
G(\sum_{i=0}^{\ell} d_i t^i) = \sum_{i=0}^\ell \tau(d_i) N_i^\sigma(\alpha) t^{ik},
$$
for any $\ell\in \mathbb{N}$ and $d_i\in S$.
The map $G$ is $S_0^\tau$-linear and multiplicative by our assumption, therefore we have   $G(t)G(c)=G(\sigma(c))G(t)$ for all $c\in S$.
 We compute
$$
G(t)G(c)=\alpha t^k \tau(c) = \alpha \sigma^k(\tau(c)) t^k
$$
and
$$
G(\sigma(c))G(t)=\tau(\sigma(c))\alpha t^k.
$$
 Since $\alpha\in S^\times$ and $\sigma$ and $\tau$ commute, we must have $\sigma=\sigma^k$, or  $k\equiv 1 \mod \mathbf{m}$.
\end{proof}

\begin{lemma}\label{l:5.6}
  Let $\tau\in \Aut(S)$ commute with $\sigma$, $k\in \mathbb{N}$, $1<k<\mathbf{n}-1$,  and $\alpha\in S^\times$.
   If the map $G_{\tau,\alpha,k}$ defined via $G(t)=\alpha t^k$ and $G(c)=\tau(c)$ for all $c\in S$, defines an isomorphism between $R/Rf$ and $R/Rh$ then
  $$N^{\sigma^k}_{\mathbf{n}}(\alpha)(\sum_{i=0}^{\mathbf{n}-1} b_i t^i)^k= \sum_{i=0}^{\mathbf{n}-1}\tau(a_i)N_i^{\sigma^k}(\alpha) t^{ki}\, {\rm mod}_r  \, h.$$
\end{lemma}

\begin{proof}
Since for all $i<\mathbf{n}-1$ the powers $t^i$ of $t$ are well defined (i.e., $t^i$ is power associative as long as $i<\mathbf{n}-1$),
so must be their images, that is the powers $G(t)^i$ of $G(t)$ are well defined, for all $i\leq \mathbf{n}-1$.
 Since $t^{\mathbf{n}} =\sum_{i=0}^{\mathbf{n}-1} a_i t^i $ in $R/Rf$, we have
 $$G(t^{\mathbf{n}})=G(\sum_{i=0}^{\mathbf{n}-1} a_i t^i )=\sum_{i=0}^{\mathbf{n}-1} \tau(a_i) G(t)^i=\sum_{i=0}^{\mathbf{n}-1}\tau(a_i)N_i^{\sigma^k}(\alpha) t^{ki}\, {\rm mod}_r  \, h.$$
 Since $G$ is a homomorphism  and  $t^{\mathbf{n}}=t t^{\mathbf{n}-1}$
 this implies that
 $G(t t^{\mathbf{n}-1})=G(t)G(t^{\mathbf{n}-1})=G(t)G(t)^{\mathbf{n}-1}$, i.e.
$$G(t)^\mathbf{n} =(\alpha t^k)(\alpha t^k)^{\mathbf{n}-1}= N_{\mathbf{n}}^{\sigma^k}(\alpha) (t^k)^{\mathbf{n}} = N_{\mathbf{n}}^{\sigma^k}(\alpha) (t^{\mathbf{n}})^{k} =
N^{\sigma^k}_{\mathbf{n}}(\alpha)(\sum_{i=0}^{\mathbf{n}-1} b_i t^i)^k$$
employing that $ t^{\mathbf{n}}= \sum_{i=0}^{\mathbf{n}-1} b_i t^i$ in $R/Rh$.
(Note we are computing in the associative skew polynomial ring here.)
Comparing these two equations yields the assertion.
\end{proof}

\begin{proposition}\label{prop:monica}
Let $\sigma\not=id$. Let $\alpha\in S^\times$.
 Let  $G: R/Rf \rightarrow R/Rh$ be an  isomorphism such that $G \vert_S = \tau$ for some
$\tau\in {\rm Aut}(S)$ which commutes with $\sigma$, and $G(t) = \alpha t^k$ for some integer $k$ such that $1<k<\mathbf{n}-1$.
    If any of the following conditions fail, then $G$ is not well defined:
    \begin{enumerate}
        \item $k\equiv 1 \mod \mathbf{m}$;
        \item $\gcd(k,\mathbf{n})=1$;
        \item if $\ell$ is the inverse of $k \mod \mathbf{m}$, then 
        $\ell \equiv 1 \mod \mathbf{m}$;
        \item  $N^{\sigma^k}_{\mathbf{n}}(\alpha)(\sum_{i=0}^{\mathbf{n}-1} b_i t^i)^k= \sum_{i=0}^{\mathbf{n}-1}\tau(a_i)N_i^{\sigma^k}(\alpha) t^{ki}\, {\rm mod}_r  \, h.$
    \end{enumerate}
\end{proposition}

\begin{proof}
    The first condition is a consequence of the requirement that $G(ta)=G(\sigma(a)t)$, cf. Lemma \ref{l:5.5}.

    If $G$ is an isomorphism, then $t=G(r)$ for some $r\in R$.  Since $G$ preserves Hamming weight, $r=st^\ell$ for some $\ell$ and $s\in S^\times$.  If $G$ is a homomorphism, then
    $$
    G(st^\ell) = \tau(s)G(t)^\ell \in S t^{k\ell}\subset R'.
    $$
    That is, we must have $k\ell \equiv 1 \mod \mathbf{n}$.  In particular, $k$ is relatively prime to $\mathbf{n}$.  This settles (2).

    If $G$ is an isomorphism, then its inverse must also be an isomorphism.  Therefore by the first condition, we must have $\ell \equiv 1 \mod \mathbf{m}$. This settles  (3).

    Condition  (4) is proved in Lemma \ref{l:5.6}.
\end{proof}

In the special case when $S[t;\sigma]/S[t;\sigma](t^{\mathbf{n}}-a)$ is associative we see an interplay between the two notions of isometry and equivalence.

\begin{theorem}\label{t:Ouazzou}\cite[Theorem 5.9]{NevPum2025}
Suppose that $a,b\in S_0^\times$ and $n\mid m$.
 Assume that $\tau$ commutes with $\sigma$.
Then the following statements are equivalent for any integer $1\leq k<\mathbf{n}$:
\begin{enumerate}[$(i)$]
\item $S[t;\sigma]/S[t;\sigma](t^{\mathbf{n}}-a)$ and $S[t;\sigma]/S[t;\sigma](t^{\mathbf{n}}-b)$ are isometric via $G_{\tau,\alpha,k}$;
\item $S[t;\sigma]/S[t;\sigma](t^{\mathbf{n}}-b^k)$ and $S[t;\sigma]/S[t;\sigma](t^{\mathbf{n}}-a)$ are equivalent via $G_{\tau,\alpha}$, where $k$ satisfies $k\equiv 1 \mod \mathbf{m}$ and $\gcd(k,\mathbf{n})=1$.
\item  $\tau(a)=N_{\mathbf{n}}^\sigma(\alpha)  b^k$ where $k\equiv 1 \mod n$ and  $\gcd(k,\mathbf{n})=1$.
\end{enumerate}
\end{theorem}

\begin{corollary} \cite[Corollary 5.10]{NevPum2025}
 Suppose $K=\mathbb{F}_{p^r}$ and $\sigma(x)=x^{p^s}$ so that $n=r/s$.
    Suppose that $\mathbb{F}_{p^r}[t;\sigma]/\mathbb{F}_{p^r}[t;\sigma](t^{\mathbf{n}}-a)$ is associative, that is,  $a\in \mathbb{F}_{p^s}^\times$ and $\mathbf{m}\mid \mathbf{n}$.
    Then the number of different
    Chen equivalence classes of skew $(\sigma,a)$-constacyclic codes over $\mathbb{F}_{p^r}$ of length $\mathbf{n}$ equals the number of different  Chen isometry classes  of skew $(\sigma,a)$-constacyclic codes.
\end{corollary}

For finite rings we can count the number of Chen isometry classes of skew constacyclic codes for classes of codes with a nonassociative ambient algebra.

 \begin{theorem}\label{c:OuazzoufiniteS}
 Let $a\in S^\times$ and assume that $S$ is a finite ring. Then
 the number of distinct Chen isometry classes of families of skew $(\sigma,a)$-constacyclic codes arising from nonassociative rings $S[t,\sigma]/S[t,\sigma](t^{\mathbf{n}}-a)$ is
    $$
     N= \begin{cases}
    |S^\times|/|N_{\mathbf{n}}^\sigma(S^\times)| &\text{if $\mathbf{m} \nmid \mathbf{n}$; and}\\
    |S^\times|/|N_{\mathbf{n}}^\sigma(S^\times)|-|S_0^\times|/|N_{\mathbf{n}}^\sigma(S^\times)|
     & \text{if  $\mathbf{m}| \mathbf{n}$.}
    \end{cases}
    $$
    There are additionally $|S_0^\times|/|N_{\mathbf{n}}^\sigma(S^\times)|$ different Chen equivalence classes (and thus at most this many Chen isometry classes) of families of skew $(\sigma,a)$-constacyclic codes arising from associative rings $S[t,\sigma]/S[t,\sigma](t^{\mathbf{n}}-a)$, that is, for which $\mathbf{m}\mid \mathbf{n}$ and $a\in S_0^\times$.
    \end{theorem}

\begin{proof}
For any $a,b\in S^\times$, the  class of skew $(\sigma,a)$-constacyclic codes is equivalent to the class of skew $(\sigma,b)$-constacyclic codes if and only if $ b \in aN_{\mathbf{n}}^\sigma(\alpha) $
 for some $\alpha \in S^\times$, that is, if and only if $a$ and $b$ lie in the same coset of the subgroup $N_{\mathbf{n}}^\sigma(S^\times)$.

First suppose that $\mathbf{m}\nmid \mathbf{n}$, so that all of the rings $S[t;\sigma]/S[t;\sigma](t^{\mathbf{n}}-a)$
are not associative.  In this case, every coset $aN_{\mathbf{n}}^\sigma(S^\times)$ corresponds to a class of skew $(\sigma,a)$-constacyclic codes such that  $S[t;\sigma]/S[t;\sigma](t^{\mathbf{n}}-a)$ is not associative, and there are  $ |S^\times|/|N_{\mathbf{n}}^\sigma(S^\times)|$  such classes.

Now suppose $\mathbf{m}\mid \mathbf{n}$, so that we are in the nonassociative case if and only if $a,b\in S\setminus S_0$.  Then $N_{\mathbf{n}}^\sigma(S^\times)\subset S_0$, so that every coset of $N_{\mathbf{n}}^\sigma(S^\times)$ is either entirely contained in $S_0$  or entirely disjoint from it (so corresponds to a nonassociative Petit ambient ring).

    The total number of cosets contained in $S_0$, and thus corresponding to associative rings $S[t;\sigma]/S[t;\sigma](t^{\mathbf{n}}-a)$ with
    $a\in S_0^\times$, is equal to
    $$
    |S_0^\times|/|N_{\mathbf{n}}^\sigma(S^\times)|.
    $$
In that case, we have shown that the notions of Chen-isometry and Chen-equivalence need not coincide.  Thus these Chen-isometry classes may be bigger (thus fewer in number), yielding the last sentence of the theorem.

Finally, subtracting these $
    |S_0^\times|/|N_{\mathbf{n}}^\sigma(S^\times)|$ cosets from the total number of cosets gives the number of Chen-isometry classes for nonassociative rings in the case that $\mathbf{m}\mid \mathbf{n}$,
    as required.
\end{proof}

\subsection{A  closer look at Equation (\ref{equ:necII})}

We now derive some necessary conditions that help us check when two skew polycyclic codes cannot be equivalent employing Equation (\ref{equ:necII}).

\begin{proposition} \label{prop:condition_f_h}
Let  $R=S[t;\sigma]$.  If
$R/Rf \cong R/Rh$ via some $G_{\tau,\alpha}$ then\\
(i)  $a_i=0$ if and only if $b_i=0$, for all $i \in\{ 0, \ldots, \mathbf{n}-1\}$;
\\ (ii) $a_i$ is invertible if and only if $b_i$ is invertible, for all $i \in\{ 0, \ldots, \mathbf{n}-1\}$.
\end{proposition}

\begin{proof}
 If $R/Rf \cong R/Rh$ via some $G_{\tau,\alpha}$, then by Theorem
\ref{general_isomorphism_theorem1}, there exists $\alpha \in S^{\times}$ such that $\tau(a_i) =
( \prod_{l=i}^{\mathbf{n}-1} \sigma^l(\alpha) ) b_i$
 for all $i \in\{ 0, \ldots, \mathbf{n}-1\}$. This implies $a_i=0$ if and only if $b_i=0$, for all $i \in\{ 0, \ldots, \mathbf{n}-1\}$, and similarly, also implies that $a_i$ is invertible if and only if $b_i$ is invertible, for all $i \in\{ 0, \ldots, \mathbf{n}-1\}$.
\end{proof}

\begin{proposition}\label{prop:conditionsonk}
Let $\tau\in \Aut(S)$ commute with $\sigma$, $\alpha \in S^{\times}$,  and set $N_{S/S_0}=N_{\mathbf{m}}^\sigma$.
Let
  $$\tau(a_i) = (\prod_{l=i}^{\mathbf{n}-1}\sigma^l(\alpha) ) b_i$$ for all $i \in\{ 0, \dots,
\mathbf{n}-1\}$.
 Suppose there exists $i_0 \in\{ 0,\dots, \mathbf{n}-1\}$ such that
 $$N_{S/S_0}(\tau(a_{i_0})) = N_{S/S_0}(b_{i_0}) $$
  is invertible, and that there is no nontrivial $(\mathbf{n}-i_0)$th root of unity in $S_0$. Then $N_{S/S_0}(\alpha)  =1$.
\end{proposition}

\begin{proof}
 Since $\tau(a_i) = (\prod_{l=i}^{\mathbf{n}-1}\sigma^l(\alpha) ) b_i$ for all $i \in\{ 0, \ldots,
\mathbf{n}-1\}$ we get
$$N_{S/S_0}(\tau(a_i)) =  \prod_{l=i}^{\mathbf{n}-1}N_{S/S_0}(\sigma^l(\alpha)) N_{S/S_0}(b_i) =  N_{S/S_0}(\alpha)^{\mathbf{n}-i} N_{S/S_0}(b_i)$$
for all $i \in\{ 0, \dots, \mathbf{n}-1\}$ by applying $N_{S/S_0}$ to both sides.
   For all  $i \in\{ 0,\dots, \mathbf{n}-1\}$ such that $N_{S/S_0}(\tau(a_i)) = N_{S/S_0}(b_i) $ is invertible,
this yields $1 =  N_{S/S_0}(\alpha)^{\mathbf{n}-i}$, therefore $N_{S/S_0}(\alpha)\in
S_0^\times$ must be an $(\mathbf{n}-i)$th root of unity. This yields the assertion.
\end{proof}

\begin{corollary}
We assume  the setting of Proposition~\ref{prop:conditionsonk}.
\\
$(i)$ Suppose that $N_{S/S_0}(\tau(a_0)) = N_{S/S_0}(b_0) $ is invertible, and that there is no nontrivial $\mathbf{n}$th root of unity in $S_0$.
 Then $N_{S/S_0}(\alpha)  =1$.
\\ $(ii)$  Suppose that  $N_{S/S_0}(\tau( a_{\mathbf{n}-1})) = N_{S/S_0}(b_{\mathbf{n}-1}) \in S^\times$.
 Then $N_{S/S_0}(\alpha)  =1$.
\\
$(iii)$ If $b_{\mathbf{n}-1}=\tau( a_{\mathbf{n}-1})$ is invertible, then $\alpha=1$.
 \\ $(iv)$ If there is some $i_0 \in\{ 0, \ldots, \mathbf{n}-1\}$ such that $\tau(a_{i_0})\not = b_{i_0}$ then $\alpha\not=1$.
\end{corollary}

\begin{proof} $(i)$ and $(ii)$ are straightforward.
\\
 $(iii)$ If $b_{\mathbf{n}-1}=\tau( a_{\mathbf{n}-1})$ is invertible, then $\tau( a_{\mathbf{n}-1}) = \sigma^{\mathbf{n}-1}(\alpha)   a_{\mathbf{n}-1}$ means $\sigma^{\mathbf{n}-1}(\alpha)=1$, i.e. $\alpha=1$.
\\
$(iv)$ If $\alpha=1$ then $\tau(a_i) = b_i$ for all $i \in\{ 0, \ldots, \mathbf{n}-1\}$.
\end{proof}

As some straightforward applications of our results, we obtain some useful equivalences which have been already observed in various literature, but to our knowledge only when $S$ is a field.

\begin{corollary}  \label{cor:equivskewsigmacyclic}
Let $a,b \in S^\times$.
\\
$(i)$ The classes of skew $(\sigma, a)$-constacyclic codes and  skew $\sigma$-cyclic codes  of length $\mathbf{n}$ over $S$ are equivalent if and only if
$a \in N_{\mathbf{n}}^\sigma (S^\times).$
\\
Equivalence and Chen equivalence coincide.
Isometry and equivalence coincide when $\mathbf{n}$  does not divide $ \mathbf{m}$.
\\ $(ii)$
 The classes of skew $(\sigma, a)$-constacyclic codes and  skew $\sigma$-negacyclic codes  of length $\mathbf{n}$
  are equivalent if and only if
 $ -a \in N_{\mathbf{n}}^\sigma (S^\times).$ \\
 Equivalence and Chen equivalence coincide.
 Isometry and equivalence coincide when $\mathbf{n}$  does not divide $ \mathbf{m}$.
 \\ (iii) The $(\sigma,a)$-constacyclic codes  of length $\mathbf{n}$ are equivalent to the cyclic codes  of length $\mathbf{n}$ if and only if
if there exists $\alpha\in S^\times$ such that $\alpha^\mathbf{n} a=1$.
\end{corollary}

 We allow that $\sigma=id$ in which case the statements hold for constacyclic codes.

\section{Code equivalence and isometry}\label{sec:equiv}

\subsection{Equivalence, and when it coincides with  isometry}
  Let
 $$f(t) = t^{\mathbf{n}}-\sum_{i=0}^{\mathbf{n}-1} a_i t^i ,\quad h(t) = t^{\mathbf{n}}-
\sum_{i=0}^{\mathbf{n}-1} b_i t^i \in S[t;\sigma]$$
be two monic polynomials.

\begin{theorem}  \label{general_isomorphism_theorem1code}
(i) Two classes  ${\bf C}_f$ and  ${\bf C}_h$ of skew polycyclic codes  of length $\mathbf{n}$ over $S$ are equivalent
 if and only if there exists
$\tau\in {\rm Aut}(S)$ that commutes with $\sigma$ and $\alpha \in S^{\times}$  such that
$$
\tau(a_i) = N_{\mathbf{n}-i}^{\sigma}(\sigma^{i} (\alpha))b_i
$$
 for all $i \in\{ 0, \ldots, \mathbf{n}-1\}$ (resp., Chen-equivalent if and only if above equation is true for $\tau=id$).
 \\ (ii) The classes of skew $(\sigma,a)$-constacyclic
and skew $(\sigma,b)$-constacyclic codes  of length $\mathbf{n}$ over $S$ are equivalent if and only if
 there exists some $\tau\in {\rm Aut}(S)$ that commutes with $\sigma$ and some $\alpha\in S^\times$  such that
\begin{equation}
 \tau(a)=N_{\mathbf{n}}^\sigma(\alpha)b,
 \end{equation}
and are  are Chen equivalent if and only if the above holds with  $\tau=id$.
 \end{theorem}

This is a corollary of Theorem \ref{general_isomorphism_theorem1} and generalizes parts of \cite[Theorem 4.20]{OaHA25} to our notion of equivalence and arbitrary rings.

Inspired by  \cite[Theorem 4.20]{OaHA25}, we can  rephrase this and say that ${\bf C}_f$ and ${\bf C}_h$ are equivalent if and only if
 there exists $\alpha\in S^\times$ and $\tau\in {\rm Aut}(S)$ that commutes with $\sigma$ such that
$$ \left( \tau(a_0),\dots, \tau(a_{i}),\dots , \tau( a_{\mathbf{n}-1})\right)$$
$$ =\left(N_n^{\sigma}(\alpha),\dots, N_{\mathbf{n}-i}^{\sigma}(\sigma^{i} (\alpha)),
\dots, N_{\mathbf{n}-1}^{\sigma}(\sigma^{\mathbf{n}-1}(\alpha))\right) \left(b_0,\dots, b_{i},\dots , b_{\mathbf{n}-1}\right) $$
with componentwise multiplication.

 Recall that each monic  right divisor $g(t)=\sum_{i=0}^r g_it^i$ of $f$  generates a principal left ideal in $R/Rf=S[t;\sigma]/S[t;\sigma]f$ which is in one-one correspondence with the generator matrix of a skew polycyclic code in ${\bf C}_f$. For equivalent codes, the image
$$G_{\tau,\alpha}(g(t))=\sum_{i=0}^r \tau(g_i)(\alpha t)^i$$
  $g$ of $f$ under the isometry
$G_{\tau,\alpha}:R/Rff\rightarrow R/Rh$
 yields the generator  matrix of a skew polycyclic code in ${\bf C}_h$ that corresponds to the left ideal in  $R/Rh$ generated by $G_{\tau,\alpha}(g)$ and is straightforward to compute.

\begin{theorem}\label{t:2}
  Let $\mathbf{m} \geq \mathbf{n}-1$ and assume that $f$   does not generate two-sided ideals in $K[t;\sigma]$, where $K$ is a field.  Suppose $\Aut(K)$ is abelian.
  Then the notions of equivalence and isometry (resp., of Chen equivalence and Chen isometry) coincide for all skew $(f,\sigma)$-polycyclic codes over $S$  of length $\mathbf{n}$.
\end{theorem}

This follows immediately from  Theorem  \ref{general_isomorphism_theorem2}.
When we only consider skew constacyclic codes then
we can say more.

\begin{corollary}\label{C:onlyweightone}  \cite[Corollary 4.5]{NevPum2025}
    Suppose $\mathbf{m}\nmid \mathbf{n}$, or that one of $a$ or $b$ is not in $S_0$. Then the notions of equivalence and isometry (resp., of Chen equivalence and Chen isometry) coincide for skew $(\sigma,a)$-constacyclic codes  over $S$  of length $\mathbf{n}$.
    \end{corollary}

  \begin{proposition}\label{p:6}
  Set
  $$f^\tau(t)= t^{\mathbf{n}}-\sum_{i=0}^{\mathbf{n}-1}\tau(a_i) t^i.$$
  Then the following statements are equivalent:
   \\ $(i)$ The classes of skew $(f,\sigma)$-polycyclic codes and skew $(h,\sigma)$-polycyclic codes  of length $\mathbf{n}$ over $S$ are equivalent.
   \\ $(ii)$ The  classes of skew $(f^\tau,\sigma)$-polycyclic and $(h,\sigma)$-polycyclic codes  of length $\mathbf{n}$ over $S$ are Chen-equivalent for all  $\tau\in {\rm Aut}(S)$ that commute with $\sigma$.
\end{proposition}

  This follows from Theorem \ref{t:cycle}.

  \begin{proof}
  There is a Chen isometry $G_{id, \alpha}:R/f\to R/Rh$ if and only if $a_i = N_{\mathbf{n}-i}^{\sigma}(\sigma^{i} (\alpha))b_i$
 for all $i \in\{ 0, \ldots, \mathbf{n}-1\}$. This is the same as $\tau^{-1}(\tau(a_i)) = N_{\mathbf{n}-i}^{\sigma}(\sigma^{i} (\alpha))b_i$
 for all $i \in\{ 0, \ldots, \mathbf{n}-1\}$ and any automorphism $\tau$.
 For any $\tau$ that commutes with $\sigma$, these equalities are equivalent to the existence of an isometry
 $G_{\tau^{-1}, \alpha}:R/f^\tau\to R/Rh$.

 Thus $G_{id, \alpha}:R/f\to R/Rh$ if and only if $G_{\tau^{-1}, \alpha}:R/f^\tau\to R/Rh$ for all $\tau$ that commute with $\sigma$.
\end{proof}

  Moreover,  $(ii)$ implies that $R/Rf$ and $R/Rf^\tau$ and are equivalent ambient rings for all $\tau$ that commute with $\sigma$. We can rephrase this: if ${\bf C}_{f}$ and ${\bf C}_h$ are Chen-isometric then
  ${\bf C}_{f^\tau}$ and ${\bf C}_f$ are equivalent for all  $\tau$ that commute with $\sigma$.

  \begin{proposition}
  Suppose  $\tau\in {\rm Aut}(S)$ commutes with $\sigma$.
  Then the following are equivalent:
   \\ $(i)$ The classes of skew $(f,\sigma)$-polycyclic codes and skew $(h,\sigma)$-polycyclic codes  of length $\mathbf{n}$ over $S$ are Chen-equivalent.
   \\ $(ii)$ The  classes of skew $(f^\tau,\sigma)$-polycyclic and $(h,\sigma)$-polycyclic codes  of length $\mathbf{n}$ over $S$ are equivalent.
  \end{proposition}

  This follows directly from Proposition \ref{p:6}.

   \begin{corollary}
    Suppose that $\tau(a_i)=a_i$ for all $i$. Then the following are equivalent:
     \\ $(i)$ The classes of skew $(f,\sigma)$-polycyclic and skew $(h,\sigma)$-polycyclic codes  of length $\mathbf{n}$ over $S$ are equivalent.
     \\ $(ii)$ The  classes of skew $(f,\sigma)$-polycyclic and skew $(h,\sigma)$-polycyclic codes  of length $\mathbf{n}$ over $S$ are Chen-equivalent.
  \end{corollary}

  We now construct  equivalence classes for a class of skew polycyclic codes.
  For a given $ h(t) = t^{\mathbf{n}}-\sum_{i=0}^{\mathbf{n}-1} b_i t^i \in S[t;\sigma]$ define
  $$ h_{\tau,\alpha}(t) = t^{\mathbf{n}}-\sum_{i=0}^{\mathbf{n}-1} N_{\mathbf{n}-i}^{\sigma}(\sigma^{i} (\tau(\alpha)))\tau(b_i) t^i
  \in S[t;\sigma].$$

\begin{theorem}\label{t:equivclasses}
The equivalence class $[{\bf C}_{h}] $ of classes of skew polycyclic codes  of length $\mathbf{n}$ over $S$ that are equivalent to ${\bf C}_{h}$ (respectively,  Chen equivalent) is given by
$$[{\bf C}_{h}]=\{{\bf C}_{h_{\tau,\alpha}} \,|\, \alpha\in S^\times, \tau\in {\rm Aut}(S) \text{ such that } \tau \text{ commutes with  }\sigma \}$$
and
$$[{\bf C}_{h}]_{Chen}=\{{\bf C}_{h_{id,\alpha}} \,|\, \alpha\in S^\times \}.$$
The equivalences between $R/Rh_{\tau,\alpha}$ and $R/Rh$ are given by $G_{\tau^{-1}, \tau(\alpha)} $, respectively between ${\bf C}_{h_{id,\alpha}}$ and ${\bf C}_h$ by $G_{id, \alpha} $.
\end{theorem}

\begin{proof}
We start with $h(t) = t^{\mathbf{n}}-
\sum_{i=0}^{\mathbf{n}-1} b_i t^i \in S[t;\sigma]$.
Take any $\alpha\in S^\times $. We now define new skew polynomials that by definition will be equivalent to $h$. First define $$h_i:=N_{\mathbf{n}-i}^{\sigma}(\sigma^{i} (\alpha))b_i$$
  for all $i \in\{ 0, \ldots, \mathbf{n}-1\}$ and
$$h_{id,\alpha}(t)=t^{\mathbf{n}}-
\sum_{i=0}^{\mathbf{n}-1} h_i t^i=t^{\mathbf{n}}-
\sum_{i=0}^{\mathbf{n}-1} N_{\mathbf{n}-i}^{\sigma}(\sigma^{i} (\alpha)) b_i t^i.$$
 Then there exists a Chen equivalence  $G_{id, \alpha} $ between  $R/Rh_{id,\alpha}$ and $R/Rh$ by definition of $h_{id,\alpha}$.
 Define
 $$h_i':=\tau(h_i)= N_{\mathbf{n}-i}^{\sigma}(\sigma^{i} (\tau(\alpha)))\tau(b_i) .$$
 Note that we can write
 $$\tau^{-1}(h_i')=\tau^{-1}(\tau(h_i))=N_{\mathbf{n}-i}^{\sigma}(\sigma^{i} (\alpha))b_i$$
  for all $i \in\{ 0, \ldots, \mathbf{n}-1\}$.

 For any $\tau\in {\rm Aut}(S)$ commuting with $\sigma$, put
 $$h_{\tau,\alpha}=t^{\mathbf{n}}-\sum_{i=0}^{\mathbf{n}-1} h_i' t^i= t^{\mathbf{n}}-\sum_{i=0}^{\mathbf{n}-1} N_{\mathbf{n}-i}^{\sigma}(\sigma^{i} (\tau(\alpha)))\tau(b_i) t^i.$$

  Thus there exists an equivalence  $G_{\tau^{-1}, \tau(\alpha)} $ between  $R/Rh_{\tau,\alpha}$ and $R/Rh$ by definition of $h_{id,\alpha}$.
\end{proof}

   Proposition \ref{prop:Chen equivpolycyclic} implies an intricate connection between classes of equivalent  $(\sigma,f)$-polycyclic and equivalent $(\sigma,a)$-constacyclic codes over commutative rings $S$ (we allow that $\sigma=id$):

\begin{proposition}  \label{p:Chen equivpolycycliccode}
  Two classes ${\bf C}_f$ and ${\bf C}_h$ of skew polycyclic codes over $S$ of length  $\mathbf{n}$ are equivalent via $G_{\tau,\alpha}$ if and only if
 for all $i \in\{ 0, \ldots, \mathbf{n}-1\}$ where $a_i\not=0$, the class of $(\sigma,a_i)$-constacyclic codes is $(\mathbf{n}-i)$-equivalent to the class of $(\sigma,b_i)$-constacyclic codes via $G_{\tau,\sigma^i(\alpha)}$.
\end{proposition}

For skew constacyclic codes more is known.

\begin{theorem}  \cite{NevPum2025}
 Suppose $\tau\in \Aut(S)$ commutes with $\sigma$ and $\alpha \in S$ is cancellable.
 Suppose $\mathbf{m}\nmid \mathbf{n}$, or that one of $a$ or $b$ is not in $S_0$.
Then two classes of skew $(\sigma,a)$-constacyclic codes and  skew $(\sigma,b)$-constacyclic codes  of length $\mathbf{n}$ are isometric if and only if they are equivalent.
\end{theorem}

This follows from Theorem \ref{thm:G}.

\begin{theorem}\label{t:Ouazzou}\cite[Theorem 5.9]{NevPum2025}
Suppose that $a,b\in S_0^\times$ and $\mathbf{m}\mid \mathbf{n}$.
 Assume that $\tau$ commutes with $\sigma$.
Then the following statements are equivalent for any integer $1\leq k<\mathbf{n}$:
\begin{enumerate}[$(i)$]
\item $S[t;\sigma]/S[t;\sigma](t^{\mathbf{n}}-a)$ and $S[t;\sigma]/S[t;\sigma](t^{\mathbf{n}}-b)$ are isometric via $G_{\tau,\alpha,k}$;
\item $S[t;\sigma]/S[t;\sigma](t^{\mathbf{n}}-b^k)$ and $S[t;\sigma]/S[t;\sigma](t^{\mathbf{n}}-a)$ are equivalent via $G_{\tau,\alpha}$, where $k$ satisfies $k\equiv 1 \mod \mathbf{m}$ and $\gcd(k,\mathbf{n})=1$.
\item  $\tau(a)=N_{\mathbf{n}}^\sigma(\alpha)  b^k$ where $k\equiv 1 \mod \mathbf{m}$ and  $\gcd(k,\mathbf{n})=1$.
\end{enumerate}
\end{theorem}

This follows from Theorem \ref{t:Ouazzou}.

\begin{corollary}\label{c:varia}
Let $a,b\in S^\times$.
\\
$(i)$ The classes of skew $(\sigma, a)$-constacyclic codes and of skew $\sigma$-cyclic codes  of length $\mathbf{n}$ are  equivalent if and only if
$a \in N_{\mathbf{n}}^\sigma (S^\times).$
 \\ $(ii)$  The classes of skew $(\sigma, a)$-constacyclic codes and  skew $\sigma$-negacyclic codes  of length $\mathbf{n}$ are equivalent if and only if
 $ -a \in N_{\mathbf{n}}^\sigma (S^\times).$
 \\ In both $(i)$ and $(ii)$, isometry and equivalence coincide when $\mathbf{m}$  does not divide $ \mathbf{n}$.
\end{corollary}

\section{Skew polycyclic codes over finite fields}\label{s:finite}

We can now get some straightforward results for skew constacyclic codes over finite fields. More detailed results will be obtained in a future paper.

Let $K=\mathbb{F}_{p^r}$ be a finite field. Let $\xi$ be a primitive element of $K$; it has order $p^r-1$.  
Put $\sigma(x)=x^{p^s}$ and let $K_0=\mathbb{F}_{p^{\gcd(r,s)}}$ denote the fixed field of $\sigma$.  The order of $\sigma$ is $\frac{r}{\gcd(r,s)}$. 
Define
$$
[\mathbf{n}]_{s}=\frac{p^{s \mathbf{n}}-1}{p^{s}-1}=p^{s(\mathbf{n}-1)}+p^{s(\mathbf{n}-2)}+\cdots + p^{s} + 1,
$$
then $N_{\mathbf{n}}^\sigma(x)=x^{[\mathbf{n}]_{s}}$ for all $x\in \mathbb{F}_{p^r}$.
The subgroup $N_{\mathbf{n}}^\sigma(\mathbb{F}_{p^r}^\times)$  of $K^\times$
is  generated by $N_{\mathbf{n}}^\sigma(\xi) = \xi^{[\mathbf{n}]_s}=\xi^w$ with $w=\gcd([\mathbf{n}]_s,p^r-1)$, and thus the order of this group is
    $$
    \frac{p^r-1}{w}.
    $$

The number of different Chen equivalence classes of skew constacyclic codes is
$$w=\gcd([\mathbf{n}]_s,p^r-1)$$
 \cite{Oua2025}. Thus the number of polynomials of the type $t^{\mathbf{n}}-c$ that yield   ambient rings $\mathbb{F}_{p^r}[t,\sigma]/\mathbb{F}_{p^r}[t,\sigma](t^{\mathbf{n}}-c)$ that are Chen equivalent to the  ambient ring $\mathbb{F}_{p^r}[t,\sigma]/\mathbb{F}_{p^r}[t,\sigma](t^{\mathbf{n}}-a)$ of a fixed $a\in \mathbb{F}_{p^r}$ is
$$\frac{p^r-1}{\gcd([\mathbf{n}]_s,p^r-1)}.$$

We will use the following result for skew constacyclic codes obtained in \cite[Theorem 5.6]{NevPum2025}.

 \begin{theorem}\label{c:Ouazzoufinite}
  Suppose $K=\mathbb{F}_{p^r}$ and $\sigma(x)=x^{p^s}$ with $s|r$ so that $n=r/s$ and $K_0=\mathbb{F}_{p^s}$.
    Then the number of distinct Chen isometry classes of families of skew $(\sigma,a)$-constacyclic codes of length $\mathbf{n}$ arising from nonassociative ambient rings $\mathbb{F}_{p^r}[t,\sigma]/\mathbb{F}_{p^r}[t,\sigma](t^{\mathbf{n}}-a)$ is $N$, where
 $$
    N= \begin{cases}
     \gcd([\mathbf{n}]_{s},p^r-1)&\text{if $\mathbf{m} \nmid \mathbf{n}$; and}\\
    \left(1-\frac{1}{[\mathbf{m}]_s}\right)\gcd([\mathbf{n}]_{s},p^r-1) & \text{if  $\mathbf{m}\mid  \mathbf{n}$.}
    \end{cases}
 $$
 There are additionally $\gcd([\mathbf{n}]_{s},p^r-1)/[n]_s$ different Chen equivalence classes (and thus at most this many Chen isometry classes) of families of skew $(\sigma,a)$-constacyclic codes arising from associative ambient rings $\mathbb{F}_{p^r}[t,\sigma]/\mathbb{F}_{p^r}[t,\sigma](t^{\mathbf{n}}-a)$, that is, for which $\mathbf{m}\mid \mathbf{n}$ and $a\in K_0$.
  \end{theorem}

Note that when $\mathbf{m}\mid \mathbf{n}$, we have $N_{\mathbf{n}}^\sigma(\mathbb{F}_{p^r})\subset \mathbb{F}_{p^s}$.

  The next examples show that using our refined notion of equivalence and isometry will give a tighter classification of skew constacyclic codes.

\begin{example}\label{e:important}\cite[Example 5.8]{NevPum2025}
Let
$\gcd([\mathbf{n}]_s,p^r-1)=p^r-1$.  Then $N_{\mathbf{n}}^\sigma(\mathbb{F}_{p^r}^\times)=\{1\}$ and so no two classes of skew $(\sigma,a)$-constacyclic and  skew $(\sigma,b)$-constacyclic codes  of length $\mathbf{n}$ belonging to  distinct $a,b\in \mathbb{F}_{p^r}^\times$ will be  Chen equivalent. Thus, there are exactly $|\mathbb{F}_{p^r}^\times|=p^r-1$ distinct classes of skew constacyclic codes under Chen equivalence. Each class is related to exactly one ambient algebra.

Moreover, if additionally $\mathbf{m} \nmid \mathbf{n}$ then Chen isometry and Chen equivalence coincide. Thus, there are also $p^r-1$ distinct classes of skew constacyclic codes under Chen isometry in that case.

When additionally $\mathbf{m}\mid \mathbf{n}$, then there are  $\left(1-\frac{1}{[n]_s}\right)(p^r-1)$ different Chen isometry classes.

 Now there are $r$ choices for $\tau\in {\rm Gal}(\mathbb{F}_{p^r}/\mathbb{F}_{p})$ ($\mathbf{m}$ of which are corresponding to $\sigma^\ell$ for each $1\leq \ell < n$).
This implies that the set $\{a^{p^{v}}\mid 0\leq v < r\}\subset \mathbb{F}_{p^r}$ is an equivalence class with exactly $r$ elements when $a\not \in \mathbb{F}_p$.  That is, all of the corresponding $(\sigma,a^{p^v})$-constacyclic codes   of length $\mathbf{n}$ are equivalent and so the equivalence class
$ [ {\bf C}_{t^{\mathbf{n}}-a} ] =\{ {\bf C}_{ t^{\mathbf{n}}-a^{p^{v}}  }\,|\, 0\leq v < r  \}$  has $r$ elements when $a\not \in \mathbb{F}_p$, while the equivalence class $[{\bf C}_{t^{\mathbf{n}}-a}]_{Chen}=\{ {\bf C}_{t^{\mathbf{n}}-a}\}$ has only one element. This yields a tighter classification of codes with the same Hamming distance, length and dimension.

Moreover, if $\mathbf{m} \nmid \mathbf{n}$ then isometry and equivalence coincide, so again calculating the  isometry classes will not yield a tighter classification.

\end{example}

We  count the number of equivalence classes in some special cases, to show how the notion of equivalence reduces the number of classes.

\begin{example}(communicated by M. Nevins)\\
 $(i)$ Let $\gcd([\mathbf{n}]_s,p^{l^2}-1)=p^{l^2}-1$,  $l$ be prime, $K=\mathbb{F}_{p^{l^2}}$, $ F=\mathbb{F}_{p^l}$ and $F_0=\mathbb{F}_p$.
Then there are $p^{l^2}-p^l$ different $b\in K\setminus F$, and $l^2$ elements in each equivalence class belonging to one such $b$ (since $|{\rm Gal}(K/F_0)|=l^2$), which yields
$$\frac{p^{l^2}-p^l}{l^2}$$
different equivalence classes of skew constacyclic codes.
Moreover, there are
 $p^{l}-p$ different $b\in F\setminus F_0$, $l$ elements in each equivalence class belonging to one such $b$ (since $|{\rm Gal}(F/F_0)|=l$), which yields
$$\frac{p^{l}-p}{l}$$
different equivalence classes, and $p-1$ different $b\in  F_0^\times$,  yielding an additional $p-1$
different equivalence classes of skew constacyclic codes. We obtain a total of only
 $$\frac{p^{l^2}-p^l}{l^2}+\frac{p^{l}-p}{l}+p-1$$
 equivalence classes, compared with a total of
 $$p^{l^2}-1=(p^{l^2}-p^l)+(p^{l}-p)+(p-1)$$
 equivalence classes of skew constacyclic codes with respect to Chen equivalence.
\\
$(ii)$
 Assume $l=1$ in $(i)$, i.e. $K=\mathbb{F}_{p^{2}}$ and $ F=\mathbb{F}_{p}$. We know $\sigma(\xi)\not=\xi$ and $$\mathbb{F}_{p^2}[t;\sigma]/\mathbb{F}_{p^2}[t;\sigma](t^{\mathbf{n}}-a)\cong \mathbb{F}_{p^2}[t;\sigma]/\mathbb{F}_{p^2}[t;\sigma](t^{\mathbf{n}}-\sigma(a))$$
 via
$G_{id,1}$. However, $$\mathbb{F}_{p^2}[t;\sigma]/\mathbb{F}_{p^2}[t;\sigma](t^{\mathbf{n}}-a)$$
 and
$$\mathbb{F}_{p^2}[t;\sigma]/\mathbb{F}_{p^2}[t;\sigma](t^{\mathbf{n}}-\sigma(a))$$
 are not Chen equivalent, hence the corresponding skew constacyclic codes are also not Chen equivalent.
\end{example}

\begin{example}
Let  $K=\mathbb{F}_{3^2}$ and $F= \mathbb{F}_{3}$, so that $\sigma$ is the Frobenius automorphism and $n=2$. Using the results from the Magma computations performed for \cite[Example 3]{Oua2025} for $n\in \{ 1,\dots, 10^8\}$ and Theorem \ref{c:Ouazzoufinite} we count the different possible equivalence classes of skew constacyclic codes  of length $\mathbf{n}$ over $K=\mathbb{F}_{3^2}$.
There are $$\gcd([\mathbf{n}]_1,8)=\gcd(\frac{3^\mathbf{n}-1}{2},8)$$
 different Chen equivalence classes.
\\ (i) Let $\mathbf{n}$ be odd. Then $\gcd([\mathbf{n}]_1,8)=1$ and there is exactly one Chen equivalence class. Since $2 \nmid \mathbf{n}$ there is also exactly one Chen isometry class.
\\  (ii) Let $\mathbf{n}\equiv 2 \mod  4$
 then $\gcd([\mathbf{n}]_1,8)=4$ and there are 4 Chen equivalence and only  $4-\frac{4}{[2]_1}=3$ Chen isometry classes (since $2\mid  \mathbf{n}$).
\\ (iii) Let $\mathbf{n}\equiv 0 \mod  4$ then $\gcd([\mathbf{n}]_1,8)=8$ and we have 8 Chen equivalence classes and only $8-\frac{8}{[2]_1}=6$  Chen isometry classes (since $2\mid  \mathbf{n}$).
\\ On the other hand, we have $\frac{6}{2}+2=5$ equivalence classes, and potentially even less isometry classes.
\end{example}

\begin{example}
Let  $K=\mathbb{F}_{2^2}$ and $F= \mathbb{F}_{2}$, so that $\sigma$ is the Frobenius automorphism and $n=2$. Let $\mathbf{n}$ be even.
Then $\gcd([\mathbf{n}]_1,3)=\gcd(2^\mathbf{n}-1,3)=3$ and there are three Chen equivalence classes  of skew constacyclic codes  of length $\mathbf{n}$ over $K=\mathbb{F}_{2^2}$, and $3(1-\frac{1}{3})=2$ Chen isometry classes (this fixes \cite[Example 4]{Oua2025} where it was claimed there was only one such class).
On the other hand, there are two equivalence classes and potentially less isometry classes.
\\ When $\mathbf{n}$ is odd, we have $\gcd([\mathbf{n}]_1,3)=\gcd(2^\mathbf{n}-1,3)=1$, so there exists only one Chen equivalence class and we are done.
\end{example}

The exact number of equivalence classes under Chen equivalence for skew polycyclic codes  of length $\mathbf{n}$ over a finite field $\mathbb{F}_{p^r}$ computed in  \cite[Theorem 4.20]{OaHA25} is
$$N=\frac{({p^r}-1)^{m_0}}{{\rm lcm}_{j\in \{i_0,i_1,\dots,i_{m_0}-1\}}\big( \frac{{p^r}-1}{\gcd(|\mathbf{n}-j|_s,{p^r}-1)}\big)}
$$
where the $j$ run over all the $m_0$ different nonzero coefficients of $f$.

\begin{example}\label{e:important2}
Let  $\gcd([\mathbf{n}]_s,p^r-1)=p^r-1$.
We construct the equivalence classes of some given $h(t)$ as in  Theorem \ref{t:equivclasses}. We start with $h(t) = t^{\mathbf{n}}-
\sum_{i=0}^{\mathbf{n}-1} b_i t^i $.
 For all $\alpha\in K^\times$, define
$$h_{id,\alpha}(t)=t^{\mathbf{n}}-
\sum_{i=0}^{\mathbf{n}-1} N_{\mathbf{n}-i}^{\sigma}(\sigma^{i} (\alpha)) b_i t^i.$$
 Then there exists a Chen equivalence  $G_{id, \alpha} $ between the ambient rings  $R/Rh_{id,\alpha}$ and $R/Rh$ and
 $$[{\bf C}_{h}]_{Chen}=\{{\bf C}_{h_{id,\alpha}} \,|\, \alpha\in \mathbb{F}_{p^r}^\times \}.$$
  All the skew polynomials $h_{id,\alpha}(t)$ have constant coefficient $b_0$.

  By Proposition \ref{p:Chen equivpolycycliccode}, employing Example \ref{e:important}, no two classes of skew $(f,\sigma)$-polycyclic and  skew $(h,\sigma)$-polycyclic codes  of length $\mathbf{n}$ belonging to any two $f$ and $h$ with  distinct constant coefficients $a_0,b_0\in K^\times$ will be  Chen equivalent since $N_{\mathbf{n}}^\sigma(K^\times)=\{1\}$.

  Now take $\tau\in {\rm Gal}(\mathbb{F}_{p^r}/\mathbb{F}_{p})$ and set
  $$h_{\tau,\alpha}(t)= t^{\mathbf{n}}-\sum_{i=0}^{\mathbf{n}-1} N_{\mathbf{n}-i}^{\sigma}(\sigma^{i} (\tau(\alpha)))\tau(b_i) t^i,$$
  then
  $$[{\bf C}_{h}]=\{{\bf C}_{h_{\tau,\alpha}} \,|\, \alpha\in \mathbb{F}_{p^r}^\times, \tau\in {\rm Gal}(\mathbb{F}_{p^r}/\mathbb{F}_{p}) \}$$
  and
  $$[{\bf C}_{h}]_{Chen}\subset [{\bf C}_{h}]$$
 Note that now the $h_{\tau,\alpha}(t)$ in the equivalence class of $h$ all have $\tau(b_0)$ as constant coefficients,  for some $\tau\in {\rm Gal}(\mathbb{F}_{p^r}/\mathbb{F}_{p})$.
  There are  $r$ choices for $\tau$.
  So when $b_0\in \mathbb{F}_q \setminus \mathbb{F}_p$ we know that there are $r$ different $\tau(b_0)$ and thus $[{\bf C}_{h}]$ must be larger.
  This means there will be less equivalence classes of skew polycyclic codes.

  Moreover, any two $f$ and $h$ whose non-zero constant coefficients $a_0$ and $b_0$ where the sets $\{ a_0,\tau(a_0),\dots,\tau^r(a_{0})\}$ and $\{b_0,\tau(b_0),\dots,\tau^r(b_{0})\}$ have no elements in common, will not yield equivalent classes of skew polycyclic  codes.
\end{example}

We can apply the connection between skew polycyclic and skew constacyclic codes from Proposition \ref{p:Chen equivpolycycliccode}  to a plethora of different scenarios.

\begin{example}\label{e:important3}
For $i \in\{ 0, \ldots, \mathbf{n}-1\}$ set $\mathbf{n}_i=\mathbf{n}-i$. Suppose there is  $i \in\{ 0, \ldots, \mathbf{n}-1\}$ such that $\gcd([\mathbf{n}_i]_s,p^r-1)=p^r-1$.
By Proposition \ref{p:Chen equivpolycycliccode},  no two classes of skew $(f,\sigma)$-polycyclic and  skew $(h,\sigma)$-polycyclic codes  of length $\mathbf{n}$ where $f$ and $h$ have  distinct $a_{i_0},b_{i_0}\in \mathbb{F}_{p^r}^\times$ will be  Chen-equivalent.
   Thus, there are  $|K^\times|=p^r-1$ distinct classes of skew polycyclic codes under $\mathbf{n}$-Chen equivalence; one for each choice of $a_{i_0}\in \mathbb{F}_{p^r}^\times$. Note that every ${\bf C}_{h_{id,\alpha}}\in [{\bf C}_{h}]_{Chen}$ has $ b_{i_0}$  as $i_0$-th coefficient.
Every ${\bf C}_{h_{\tau,\alpha}}\in [{\bf C}_{h}]$ has $ \tau(b_{i_0})$  as $i_0$th coefficient.
 \end{example}

As a consequence of Theorem \ref{t:2}, the notions of Chen equivalence and Chen isometry coincide for all skew $(f,\sigma)$-polycyclic codes of length $\mathbf{n}$.

  Let $\mathbf{m}\geq \mathbf{n}-1$ and assume that $f$ does not generate a two-sided ideal in $\mathbb{F}_{p^r}[t;\sigma]$.
 Then the notions of Chen equivalence and Chen isometry coincide for all skew $(f,\sigma)$-polycyclic codes over $\mathbb{F}_q$  of length $\mathbf{n}$.
 Therefore   \cite[Theorem 4.20]{OaHA25} yields also the exact number of Chen isometry classes of skew polycyclic codes in certain cases:

\begin{theorem}\label{c:Ouazzoufinite2}
Let $\mathbf{m} \geq \mathbf{n}-1$ and assume that $f$ does not generate a two-sided ideal in $\mathbb{F}_{p^r}[t;\sigma]$.
The exact number of Chen isometry classes of skew polycyclic codes   of length $\mathbf{n}$ over  $\mathbb{F}_{p^r}$   is
$$N=\frac{({p^r}-1)^{m_0}}{{\rm lcm}_{j\in \{i_0,i_1,\dots,i_{m_0}-1\}}\big( \frac{{p^r}-1}{\gcd(|\mathbf{n}-j|_s,{p^r}-1)}\big)}
$$
where the $j$ run over all the $m_0$ different nonzero coefficients of $f$.
\end{theorem}

\section{When two classes cannot be equivalent}\label{s:end}

If we  want to have some criteria for when skew polycyclic codes cannot be equivalent, there are some easy conditions we can check first.
 Let
 $$f(t) = t^{\mathbf{n}}-\sum_{i=0}^{\mathbf{n}-1} a_i t^i ,\quad h(t) = t^{\mathbf{n}}-
\sum_{i=0}^{\mathbf{n}-1} b_i t^i \in S[t;\sigma].$$
We are looking at classes ${\bf C}_f$ and ${\bf C}_h$ of skew polycyclic codes  of length $\mathbf{n}$ over $S$, and look at the coefficients of $f$ and $h$.

\begin{corollary}  \label{cor:Chen equivpolycyclic}
Suppose that
$(a_0,\dots, a_{\mathbf{n}-1})\sim_\mathbf{n} (b_0,\dots,b_{\mathbf{n}-1})$, that is the classes ${\bf C}_f$ and ${\bf C}_h$ of skew polycyclic codes are equivalent. Then the following are true:
\\ (i)  $a_i  = 0$ if and only if $b_i  = 0$ for any $i\in \{0,1,\dots,\mathbf{n}-1 \}$. In particular, $(a_0,\dots, a_{\mathbf{n}-1})$ and $(b_0,\dots,b_{\mathbf{n}-1})$ have the same Hamming weight.
\\ (ii) $a_i  \in S^\times$ if and only if $b_i  \in S^\times$ for any $i\in \{0,1,\dots,\mathbf{n}-1 \}$.
\end{corollary}

This generalizes \cite[Lemma 4.19]{OaHA25} which was stated for finite fields and Chen equivalence (i.e., $\tau=id$).
When $\sigma=id$ we are dealing with polycyclic codes.

\begin{proof}
Suppose that
$(a_0,\dots, a_{\mathbf{n}-1})\sim_\mathbf{n} (b_0,\dots,b_{\mathbf{n}-1})$, which means that
 there exist $\alpha \in S^{\times}$  and some $\tau\in {\rm Aut}(S)$ that commutes with $\sigma$, such that  $G_{\tau,\alpha}:R/Rf\rightarrow R/Rh$ is an equivalence. This in turn is the case if and only if
 there exists $\alpha \in S^{\times}$ and some $\tau\in {\rm Aut}(S)$ that commutes with $\sigma$, such that
$$\tau(a_i) = N_{\mathbf{n}-i}^\sigma ( \sigma^i(\alpha) ) b_i$$
for all $i \in\{ 0, \ldots, \mathbf{n}-1\}$.
This means:
 \\ $(i)$ $a_i = 0$ if and only if $b_i = 0$ for all $i \in\{ 0, \ldots, \mathbf{n}-1\}$.
 \\ $(ii)$ $a_i $ is invertible if and only if $b_i $ is invertible for all $i \in\{ 0, \ldots, \mathbf{n}-1\}$. ($(i)$ and $(ii)$ also follow from Proposition \ref{prop:condition_f_h}.)
\end{proof}

Thus if $(a_0,\dots, a_{\mathbf{n}-1})$ and $ (b_0,\dots,b_{\mathbf{n}-1})$ do not have zero entries in exactly the same slots, then the classes ${\bf C}_f$ and ${\bf C}_h$ of skew polycyclic codes are not equivalent. And even if $(a_0,\dots, a_{\mathbf{n}-1})$ and $ (b_0,\dots,b_{\mathbf{n}-1})$ do  have zero entries in exactly the same slots, hence the same Hamming weights, but do not have the same slots with invertible elements, then the classes ${\bf C}_f$ and ${\bf C}_h$ of skew polycyclic codes are not equivalent.

\begin{corollary}  \label{cor:Chen equivpolycyclic2}
Suppose that
$(a_0,\dots, a_{\mathbf{n}-1})\sim_\mathbf{n} (b_0,\dots,b_{\mathbf{n}-1})$, that is the classes ${\bf C}_f$ and ${\bf C}_h$ of skew polycyclic codes are equivalent. Then the following are true:
\\ (i)  $\tau( a_{\mathbf{n}-1})=\sigma^{\mathbf{n}-1}(\alpha)b_{\mathbf{n}-1}$.
\\ (ii) $N_{S/S_0}(\tau(a_i)b_i^{-1}) \in (N_{S/S_0}(S^\times))^{\mathbf{n}-i}. $
\end{corollary}

\begin{proof}
The relation
$(a_0,\dots, a_{\mathbf{n}-1})\sim_\mathbf{n} (b_0,\dots,b_{\mathbf{n}-1})$ again is equivalent to the existence of some $\alpha \in S^{\times}$ and some $\tau\in {\rm Aut}(S)$ that commutes with $\sigma$, such that
$$\tau(a_i) = N_{\mathbf{n}-i}^\sigma ( \sigma^i(\alpha) ) b_i$$
for all $i \in\{ 0, \ldots, \mathbf{n}-1\}$.
This means:
\\ $(i)$  $\tau( a_{\mathbf{n}-1})=\sigma^{\mathbf{n}-1}(\alpha)b_{\mathbf{n}-1}$.
 \\ $(ii)$ If we apply $N_{S/S_0}$ to both sides of the equation we obtain
 $$N_{S/S_0}(\tau(a_i)) = N_{S/S_0}(\alpha )^{\mathbf{n}-i} N_{S/S_0}(b_i)$$
 which is equivalent to
 $$N_{S/S_0}(\tau(a_i)b_i^{-1}) \in (N_{S/S_0}(S^\times))^{\mathbf{n}-i}\subset (S_0^\times)^{\mathbf{n}-i}. $$
\end{proof}

\section{Skew polycyclic codes and their ambient algebras over commutative chain rings}\label{sec:chain}

\subsection{Finite commutative Chain Rings}

We now briefly turn to skew polycyclic codes over finite (commutative) chain rings. The Chen isometry and Chen equivalence classes of  constacyclic codes over  finite  chain rings
were counted in \cite{CON2025}. Constacyclic codes can indeed be recovered as a special case of our setup, where $R=S[t]$, i.e. $\sigma=id$, $\delta=0$, and $f(t)=t^{\mathbf{n}}-a$. Let $p$ be a prime and $q=p^s$.

A finite unital commutative ring $S\not=\{0\}$ is called a \emph{finite chain ring} if its ideals are linearly
ordered by inclusion.

 A finite chain ring $S$ has a unique maximal ideal $\mathfrak{p}=(\gamma)$ and every ideal of a finite chain ring $S$ is principal. In particular, $S$ is a
local ring with a finite residue  field $K=S/\mathfrak{p}$.  The ideals $(\gamma^i)=\gamma^i S=\mathfrak{p}^i$ of $S$ form the proper chain
$$S=(1)\supseteq(\gamma)\supseteq (\gamma^2) \supseteq \dots \supseteq(\gamma^{s_0})=(0),$$
where the integer ${s_0}$ is the nilpotency index of $S$.  If $K$ has $q$ elements, then $|S|=q^{s_0}$.
 In general, the set $S^\times$ can be difficult to compute explicitly, it is known when $p-1\nmid {s_0}$ and is computed via an algorithm when $p-1\mid {s_0}$ in \cite{HouLM2003}, which also contains an excellent review of the results and literature for $S^\times$ we will use. Results on automorphism groups of $S$ are presented for instance in \cite{A1, A2}.

 Let
$$\pi:S\rightarrow K=S/\mathfrak{p}, \quad x\mapsto \overline{x}=x \,{\rm mod}\, \gamma$$
  be the canonical projection.

For two finite chain rings $S$ and $S_0$, a ring extension $S/S_0$ is called \emph{separable} if $\mathfrak{p}S=\mathfrak{P}$, where $\mathfrak{p}$ is the maximal ideal of $S_0$ and $\mathfrak{P}$ the maximal ideal of $S$, and
\emph{Galois} if $S/S_0$ is separable and
 $S^G=\{s\in S\,|\,\tau(s)=s \text{ for all }\tau\in G\}=S_0$, where
$G={\rm Gal}(S/S_0)=\{ \sigma\in {\rm Aut}(S) \,|\,\sigma|_{S_0}=id\}$ is the Galois group of $S/S_0$.

 We know every Galois extension $S/S_0$ can be written as
$S=S_0[x]/(s(x))$, where $(s(x))$ is the ideal generated by a monic basic irreducible polynomial $s(x)\in S_0[x]$ (i.e., $\bar s$ is irreducible), e.g. see \cite[Theorem XIV.8]{Mc}. 

The Galois group of a Galois extension
$S/S_0$ is cyclic and isomorphic to the Galois group of the extension $\mathbb{F}_{q^n}/\mathbb{F}_{q}$,
where $\mathbb{F}_{q^n}=S/\mathfrak{P}$ is the residue field of $S$ and $\mathbb{F}_{q}=S_0/\mathfrak{p}$ the residue field of $S_0$; its generator is $\sigma(a)=a^q$ for a suitable
primitive element $a\in S$, and $\{a,\sigma(a),\dots,\sigma^{n-1}(a)\}$ is a free $S_0$-basis of $S$.
Since $S$ is also an unramified extension of $S_0$, $\mathfrak{P}=S\mathfrak{p}=Sp$.

\subsection{Petit rings over finite chain rings}

Let $S$ be a finite chain ring with residue class field  $K=S/\mathfrak{p}$, $\mathfrak{p}=(\gamma)$ and $\sigma\in {\rm Aut}(S)$,
$\delta$ a left $\sigma$-derivation.

\begin{remark} A skew polynomial
$f(t)=\sum_{i=0}^{\mathbf{n}}d_it^i\in S[t;\sigma]$ is nilpotent, if and only if $f$ is a zero-divisor, if and only if $d_i\in (\gamma)=\mathfrak{p}$ for all $i$.
If $f$ is not a zero divisor in $S[t;\sigma]$ then $f$ is regular. Moreover,
$f(t)=\sum_{i=0}^{\mathbf{n}}d_it^i\in S[t;\sigma]$ is invertible if and only if $d_i\in (\gamma)$ for all $i$, $1\leq i\leq \mathbf{n}$ and $d_0$ is invertible
\cite{Mc}.
\end{remark}

We observe  that for every $\tau\in {\rm Aut}(S)$ we have $\sigma((\gamma))= (\gamma)$, so that $\tau$ canonically induces an automorphism
$$\overline{\tau}:K\rightarrow K,\,\, \overline{\tau}(\overline{a})=\overline{\tau(a)},$$
that means $\tau=\overline{\tau} \circ \pi$.
 We assume that $\delta((\gamma))\subset(\gamma)$, then $\delta$ canonically induces a left $\overline{\sigma}$-derivation
 $$\overline{\delta}:K\rightarrow K, \,\,  \overline{\delta}(\overline{x})=\overline{\delta(x)}.$$
We hence obtain a canonical surjective ring homomorphism
$$\can:S[t;\sigma,\delta]\rightarrow K[t;\overline{\sigma},\overline{\delta}], \,\,c(t)=\sum_{i=0}^{l}c_it^i \mapsto \overline{c}(t)
=\sum_{i=0}^{l}\overline{c_i}t^i.$$
Note that  $K[t;\overline{\sigma},\overline{\delta}]\cong K[t;\sigma']$ for a suitable $\sigma'$ and that all polynomials in the skew polynomial ring $K[t;\overline{\sigma'}]$ are bounded since $K$ is a finite field. This isomorphism, however, changes the Hamming weight, so we cannot use it.

From now on let $f\in R$ be monic of degree $\mathbf{n}$, \emph{regular} (i.e., $\overline{f}\not=0$) and decomposable. 
Suppose $f=ug$ in $S[t;\sigma,\delta]$, then we know that $\overline{f}=\overline{u}\,\overline{g}\in  K[t;\overline{\sigma},\overline{\delta}]$
is monic of degree $\mathbf{n}$ and decomposable.

Note that if $S$ has cardinality $q^{s_0}$ then
 $$\mathbb{S}_f=S[t;\sigma,\delta]/S[t;\sigma,\delta]f$$
  has $q^{s_0 \mathbf{n}}$ elements and
 $S_{\overline{f}}=K[t;\overline{\sigma},\overline{\delta}]/K[t;\overline{\sigma},\overline{\delta}]\overline{f}$
has $q^{\mathbf{n}}$ elements.

Since $\delta((\gamma)) \subset(\gamma)$,
  $\gamma \mathbb{S}_f$ is a two-sided ideal in $\mathbb{S}_f$
and the canonical surjective ring homomorphism
$\can:S[t;\sigma,\delta]\rightarrow K[t;\overline{\sigma},\overline{\delta}]$
 induces a surjective homomorphism of nonassociative rings
 $$\Psi:S[t;\sigma,\delta]/S[t;\sigma,\delta]f \rightarrow
K[t;\overline{\sigma},\overline{\delta}]/K[t;\overline{\sigma},\overline{\delta}]\overline{f},\quad v(t) \mapsto \overline{v}(t)$$
with  ${\rm ker}(\Psi)=\gamma \mathbb{S}_f$, which then in turn induces a canonical isomorphism of the nonassociative Petit rings
$$\overline{\mathbb{S}_f}:= \mathbb{S}_f/\gamma \mathbb{S}_f \text{ and } K[t;\overline{\sigma},\overline{\delta}]/K[t;\overline{\sigma},\overline{\delta}]\overline{f}= S_{\overline{f}}$$
 given by
\begin{equation}
\mathbb{S}_f/\gamma \mathbb{S}_f\rightarrow
K[t;\overline{\sigma},\overline{\delta}]/K[t;\overline{\sigma},\overline{\delta}]\overline{f},\quad
v(t)+ \gamma \mathbb{S}_f \mapsto \overline{v}(t).
\end{equation}
It is straightforward to see that we have
$$\overline{C(\mathbb{S}_f)}=C(\mathbb{S}_f)/pC(\mathbb{S}_f)\cong C(\overline{\mathbb{S}_f})$$
for the centers.
Moreover, if $ \mathbb{S}_f$ is an associative ring then so is  $\overline{\mathbb{S}_f}$, or in other words, if $f\in R$ is two-sided then so is $\bar f\in K[t;\overline{\sigma},\overline{\delta}]$.

There are interesting phenomena to observe when ``moving'' from Petit rings over $S$ to Petit rings over $K$, which directly influence the isometry and equivalence behaviour of skew polycyclic codes over $S$.

\begin{example} \label{ex:8.2}
Let $S/S_0$ be a Galois extension of Galois rings with Galois group
${\rm Gal}(S/S_0)=\langle\sigma\rangle$  of  order $\mathbf{m}$ and let $K_0$ denote the residue field of $S_0$, $ {\rm char }(K_0)=p$.
Then
$K/K_0$ is a Galois extension with Galois group $\langle\overline{\sigma}\rangle$ of  order $\mathbf{m}$.
\\ $(i)$  Choose $$f(t)=t^{\mathbf{n}}+ph(t)-d\in R=S[t;\sigma]$$
 with $h(t)\in S[t;\sigma]$ of degree $< \mathbf{n}$ and  $\overline{d}\not=0$ in $K$, then
$$\overline{f}(t)=t^{\mathbf{n}}-\overline{d}.$$

Note that $\mathbb{S}_{f}$ is  the ambient ring of $(f,\sigma)$-skew polycyclic codes  of length $\mathbf{n}$ over $S$ and $\mathbb{S}_{t^{\mathbf{n}}-\overline{d}}$ is  the ambient ring of skew $(\bar \sigma,\overline{d})$-constacyclic codes  of length $\mathbf{n}$ over $K$.

 Since ${\rm Fix}(\sigma)={\rm Fix}(\overline{\sigma})$, if $\mathbf{m}=\mathbf{n}$ then $\overline{f}(t)=t^\mathbf{m}-\overline{d}$ is two-sided if and only if $d\in {\rm Fix}(\sigma)$. When $\mathbf{m}=\mathbf{n}$ we hence know that any $f\in R$ such that $\overline{f}(t)=t^\mathbf{m}-\overline{d}$ and $d\in S\setminus S_0$ must be not two-sided, implying that $\mathbb{S}_{f}$ is a proper nonassociative ring.

 If $\mathbf{m}$ is prime or if the elements $1,\overline{d_0},\dots,\overline{d_0}^\mathbf{m}$ are linearly independent over $K_0={\rm Fix}(\overline{\sigma})$, then
$\mathbb{S}_{\overline{f}}$ is a proper nonassociative cyclic division algebra of degree $\mathbf{m}$ and $\bar f$ is irreducible, hence $f$ is irreducible, too. So when working with codes we   assume that $1,\overline{d_0},\dots,\overline{d_0}^\mathbf{m}$
are linearly dependent over $K$, or else $f$  will be irreducible.
\end{example}

\subsection{Isometries and equivalences of skew polycyclic codes over finite chain rings}

We use the notation and conventions of the previous section and assume that $S$ is a finite chain ring and put $R=S[t;\sigma]$ for some $\sigma\in {\rm Aut}(S)$  of  order $\mathbf{m}$. Let $K$ be the residue field of $S$.

Let
 $f(t) = t^{\mathbf{n}}-\sum_{i=0}^{\mathbf{n}-1} a_i t^i $ and $h(t) = t^{\mathbf{n}}-
\sum_{i=0}^{\mathbf{n}-1} b_i t^i \in S[t;\sigma]$ be reducible; to avoid pathological cases, let us assume that $\overline{a_0}\not=0$ and $\overline{b_0}\not=0$  (so in particular, both $f$ and $h$ are regular). We note that  if $g$ is a monic right divisor of $f$, then $\bar g$ is a monic right divisor of $\bar f$, which will be needed when looking at skew polycyclic codes.

\begin{proposition}\label{prop:chain}
 Let $\alpha\in S^\times$. 
 \\ (i) If  $ G_{\tau,\alpha}: R/Rf\to R/Rh $ is an equivalence between nonassociative unital Petit rings, then so is
    $$  G_{\bar \tau,\bar \alpha}: K[t;\bar \sigma]/K[t;\bar \sigma]\bar f\to K[t;\bar \sigma]/ K[t;\bar \sigma]\bar h. $$
(ii)
 If  $ G_{\tau,\alpha,k}: R/Rf\to R/Rh $ is an isometry between nonassociative unital Petit rings, then so is
 $$  G_{\bar \tau,\bar \alpha,k}: K[t;\bar \sigma]/K[t;\bar \sigma]\bar f\to K[t;\bar \sigma]/ K[t;\bar \sigma]\bar h. $$
\end{proposition}

Note that here $G_{\bar \tau,\bar \alpha,k}=\overline{G_{\tau,\alpha}}$ is canonically defined via $G_{\tau,\alpha,k}$.

\begin{proof}
 $(i)$ We use Theorem \ref{general_isomorphism_theorem1} to prove this directly to see the details. We have $  G_{\tau,\alpha}:\mathbb{S}_f\to \mathbb{S}_h$, so we know that $  G_{\tau,\alpha}(t)=\alpha t$
and $  G_{\tau,\alpha}|_{S}=\tau$. Moreover, $\tau(a_i) = N_{\mathbf{n}-i}^{\sigma}(\sigma^{i} (\alpha))b_i$
 for all $i \in\{ 0, \ldots, \mathbf{n}-1\}$. This implies that $\bar\tau(\overline{a_i}) = N_{\mathbf{n}-i}^{\bar \sigma}(\bar \sigma^{i} (\bar \alpha))\overline{b_i}$
 for all $i \in\{ 0, \ldots, \mathbf{n}-1\}$ and hence that there exists an equivalence $G_{\bar \tau,\bar \alpha}: K[t;\bar \sigma]/K[t;\bar \sigma]\bar f\to K[t;\bar \sigma]/ K[t;\bar \sigma]\bar h.$
 This yields the assertion.
 \\ $(ii)$  Let  $ G_{\tau,\alpha,k}: \mathbb{S}_f\to \mathbb{S}_h $ be a ring isomorphism, then it restricts to $\tau$ on $S$, so that
     $ G_{\tau, \alpha,k}(\gamma \mathbb{S}_f ) \subset \gamma \mathbb{S}_h$. Therefore we have a canonically induced isomorphism
     $$  \overline{G_{\tau, \alpha,k}}: \mathbb{S}_f/\gamma \mathbb{S}_f\to \mathbb{S}_h/\gamma \mathbb{S}_h $$
     that restricts to $\bar \tau:S/(\gamma) \to S/(\gamma)$, that is to $\bar \tau\in {\rm Aut}(K)$, and maps $t$ to $\bar \alpha t^k+\gamma \mathbb{S}_h$.

Now $\overline{\mathbb{S}_f}= \mathbb{S}_f/\gamma \mathbb{S}_f $ is canonically isomorphic to $K[t;\overline{\sigma}]/K[t;\overline{\sigma}]\overline{f}= S_{\overline{f}}$, so we also have a canonical ring isomorphism
      $$  \overline{G_{\tau, \alpha,k}}: K[t;\bar \sigma]/K[t;\bar \sigma]\bar f\to K[t;\bar \sigma]/ K[t;\bar \sigma]\bar h. $$
      that restricts to $\bar \tau \in {\rm Aut}(K)$. We know this means that $t$ must be mapped to $\beta t^\ell$ for some nonzero $\beta\in K$ and integer $\ell>0$, so immediately obtain  $ \overline{G_{\tau, \alpha,k}}= G_{\bar \tau,\beta,\ell}$.
      We now show that indeed $  \overline{G_{\tau, \alpha,k}} =G_{\bar \tau,\bar \alpha,k}$. Since $G_{\tau, \alpha,k}(t)=\alpha t^k$, we have the canonical identifications
      $$\overline{G_{\tau, \alpha,k}}(t)=\bar \alpha t^k +\gamma \mathbb{S}_h=\bar \alpha t^k+(S/(\gamma)[t;\sigma])/(S/(\gamma)[t;\sigma])\bar h=\bar \alpha t^k+K[t;\bar \sigma]\bar h,$$
      therefore $\beta=\bar \alpha$ and  $\ell=k$.
 \end{proof}

    \begin{corollary}
    $(i)$ Suppose that
$(a_0,\dots, a_{\mathbf{n}-1})\cong_\mathbf{n}  (b_0,\dots,b_{\mathbf{n}-1})$, that is the classes ${\bf C}_f$ and ${\bf C}_h$ of skew polycyclic codes over $S$ are $\mathbf{n}$-isometric, then
$(\overline{a_0},\dots,\overline{ a_{\mathbf{n}-1}})\cong_\mathbf{n} (\overline{b_0},\dots,\overline{b_{\mathbf{n}-1}})$, that is the classes ${\bf C}_{\bar f}$ and ${\bf C}_{\bar h}$ of skew polycyclic codes over $K$ are $\mathbf{n}$-isometric.
\\ $(ii)$ Suppose that
$(a_0,\dots, a_{\mathbf{n}-1})\sim_\mathbf{n} (b_0,\dots,b_{\mathbf{n}-1})$, that is the classes ${\bf C}_f$ and ${\bf C}_h$ of skew polycyclic codes over $S$ are $\mathbf{n}$-equivalent, then
$(\overline{a_0},\dots,\overline{ a_{\mathbf{n}-1}})\sim_\mathbf{n} (\overline{b_0},\dots,\overline{b_{\mathbf{n}-1}})$, that is the classes ${\bf C}_{\bar f}$ and ${\bf C}_{\bar h}$ of skew polycyclic codes over $K$ are $\mathbf{n}$-equivalent.
    \end{corollary}

For skew constacyclic codes over finite chain rings, $\mathbf{n}$-isometries between their ambient rings where $k>1$ are rare, and can only occur under rigid conditions on $k$ and the ambient algebras.

    \begin{theorem}\label{thm:chainmain}
Let $f,h\in R $ such that $\bar f(t)=t^{\mathbf{n}}-a_0$ and $\bar h(t)=t^{\mathbf{n}}-b_0$.
 Suppose  that $2\leq k < m$ and that there exists an isometry $ G_{\tau,\alpha,k}: R/Rf\to R/Rh $, then
     \\ (i) $k\equiv 1 \mod \mathbf{m}$;
 \\ (ii) $ \mathbf{m}\mid \mathbf{n}$;
 \\ (iii)  $\gcd(k,\mathbf{n})=1$;
 \\ (iv)  $\overline{b_0}, \overline{b_0}\in K_0$;
    \end{theorem}

    \begin{proof}
    The isometry $ G_{\tau,\alpha,k}$ canonically induces the isometry $ G_{\bar \tau,\bar \alpha,k}:K[t;\sigma]/K[t;\sigma]\bar f \to K[t;\sigma]/K[t;\sigma]\bar h $. By Theorem \ref{thm:G} (\cite[Theorem 4.4]{NevPum2025}) the existence of the isometry $ G_{\bar \tau,\bar \alpha,k}:K[t;\sigma]/K[t;\sigma]\bar f \to K[t;\sigma]/K[t;\sigma]\bar h $
      implies that the Petit rings $K[t;\sigma]/K[t;\sigma]\bar f$ and $ K[t;\sigma]/K[t;\sigma]\bar h $ are associative and that
    $k\equiv 1 \mod \mathbf{m}$, we have
     $\mathbf{m}\mid \mathbf{n}$, $\gcd(k,\mathbf{n})=1$ and
     $\overline{b_0}, \overline{b_0}\in K_0$,
      $(N_{K/K_0}(\alpha))^{\mathbf{n}/n}\overline{b_0}^k=\tau(\overline{a_0})$.
 \end{proof}

Under certain additional assumptions, there will exist only equivalences between skew polycyclic codes over finite chain rings $S$.

\begin{theorem}\label{thm:chain}
 Suppose  that $\mathbf{n}\leq \mathbf{m}-1$ and that $\bar f$ and $\bar h$  do not generate two-sided ideals in $K[t;\sigma]$. If  $ G_{\tau,\alpha,k}: R/Rf\to R/Rh $ is an isometry, then $k=1$.
\end{theorem}

\begin{proof}
If $\bar f$ and $\bar h$  do not generate two-sided ideals in $K[t;\sigma]$ and $\mathbf{n}\leq \mathbf{m}-1$ then we know that every
    isometry $ G_{\bar \tau,\bar \alpha,k}:K[t;\sigma]/K[t;\sigma]\bar f \to K[t;\sigma]/K[t;\sigma]\bar h $ must have $k=1$ by Theorem  \ref{general_isomorphism_theorem2}. Therefore we get $ G_{\tau,\alpha,k}= G_{\tau,\alpha,1}$.
\end{proof}

\begin{theorem} \label{thm:chain2}
Suppose that one of the following holds.
\\ (i)  $\mathbf{n} \leq \mathbf{m}-1$
and  $\bar f$  does not generate a two-sided ideal in $K[t;\sigma]$.
 \\ (ii)  $\bar f(t)=t^{\mathbf{n}}-\overline{ a_0}$ and either $\mathbf{m}\nmid \mathbf{n}$, or  $\overline{ a_0}$  is not in $K_0$.
\\ Then the notions of equivalence and isometry (resp., of Chen equivalence and Chen isometry) coincide for all skew $(f,\sigma)$-polycyclic codes over $S$  of length $\mathbf{n}$.
\end{theorem}

\begin{proof}
$(i)$ This follows immediately from  Theorem  \ref{thm:chain}.
\\ $(ii) $ Since $\mathbf{m}\nmid \mathbf{n}$ or  $\bar a\not \in K_0$, the notions of equivalence and isometry (resp., of Chen equivalence and Chen isometry) coincide for skew $(\bar \sigma, \bar a)$-constacyclic codes  over $K$ by Corollary \ref{C:onlyweightone}, implying that every
 isometry $ G_{\bar \tau,\bar \alpha,k}:K[t;\sigma]/K[t;\sigma]\bar f \to K[t;\sigma]/K[t;\sigma]\bar h$ must have $k=1$, hence  $ G_{\tau,\alpha,k}= G_{\tau,\alpha,1}$.
\end{proof}

\subsection{Galois rings}
Given a prime $p$ and positive integers ${m_0}$ and $r$, let $s(x)\in \mathbb{Z}_{p^{m_0}}[x]$ be a monic basic irreducible polynomial of degree $r$. Then $G(p^{m_0},r)=\mathbb{Z}_{p^{m_0}}[x]/(s(x))$
 is called the \emph{Galois ring}  of characteristic
$p^{m_0}$ and rank $r$. Different choices of  monic basic irreducible polynomials $s(x)\in  \mathbb{Z}_{p^{m_0}}[x]$ of degree $r$ yield isomorphic Galois rings. Galois rings are examples of finite chain rings. The Galois ring $G(p^{m_0},r)$ is a local ring with maximal ideal generated by $p$ and residue field $K=G(p^{m_0},n)/pG(p^{m_0},r)=\mathbb{F}_{p^r}$.
The cardinality of $G(p^{m_0},r)$ is $p^{{m_0}r}$.

We sometimes use the notation $G(p^{m_0},r)=\mathbb{Z}_{p^{m_0}}[\omega]$ with $\omega=x+(s(x))$. Then $s(x)$ is the unique polynomial of degree $\leq r$ such that $s(\omega)=0$, and $\mathbb{F}_{p^r}=\mathbb{F}_{p}[\bar \omega]$. This also shows that there exists an element in $G(p^{m_0},r)$ of order $p^r-1$ which is a root of a basic primitive polynomial of degree $r$ over $\mathbb{Z}_{p^{m_0}}$, and dividing $x^{p^r-1}-1$ in $\mathbb{Z}_{p^{m_0}}[x].$

Every $a\in G(p^{m_0},r)$ has a unique representation as a $p$-adic expansion
$$a=\xi_0+p\xi_1+p^2\xi_2+\cdots+ p^{m_0-1}\xi_{{m_0}-1}$$
for suitable  $\xi_i\in T\subset S$, where $T=\{0,\xi_1,\dots,\xi_{{m_0}-1} \}$ is the Teichmueller set of $G(p^{m_0},r)$.

Moreover, when $r=\mathbf{m}s$  then $G(p^{m_0},r)$  is a Galois extension of the Galois ring $G(p^{m_0},s)$ with cyclic Galois group of  order $\mathbf{m}$; in particular  $G(p^{m_0},r)$ is a Galois extension of the Galois ring $G(p^{m_0},1)=\mathbb{Z}/(p^{m_0})$ with cyclic Galois group of order $r$.
The  automorphism generating the Galois group of $G(p^{m_0},r)/G(p^{m_0},s)$ is
$$\sigma(a)=\xi_0^q+p\xi_1^q+p^2\xi_2^q+\cdots+ p^{{m_0}-1}\xi_{{m_0}-1}^q;$$
$\sigma$ has  order $\mathbf{m}$ and $\overline{\sigma}:K\to K$, $\overline{\sigma}(y)=y^q$, $K_0=\mathbb{F}_{p^s}$.

For every Galois extension $G(p^{m_0},r)/G(p^{m_0},s)$, $r=\mathbf{m}s$, there is a basic primitive polynomial $p(x)\in G(p^{m_0},s)[x]=\mathbb{Z}_{p^{m_0}}[\omega][x]$ (i.e. $\overline{p(x)}\in \mathbb{F}_{p^s}[x]=\mathbb{F}_{p}(\bar \omega)[x]$ is primitive) of degree $\mathbf{m}$ such that we can write $G(p^{m_0},r)\cong G(p^{m_0},s)/(p(x))$.

Put $S=G(p^{m_0},r)$ and $T^\times=T\setminus \{0\}$ then it follows that
$$S^\times =  T^\times (1+pS) \cong T^\times \times (1+pS).$$
 Since $T^\times$ is a cyclic group of order $p^r-1$ and $(1+pS)$ is a $p$-group of order $p^{r(m_0-1)}$, it follows that
$$|S^\times|=|S^\times| |(1+pS)|= (p^r-1) p^{r({m_0}-1)}$$
\cite{Wood2008, BiFl2002}.

\subsection{Finite chain rings defined via Galois rings}
A general finite chain ring $S$ with maximal ideal $\mathfrak{p}$, residue field $S/\mathfrak{p}\cong\mathbb{F}_{p^r}$ and characteristic $p^{m_0}$
has nilpotency index $s_0$ and there is an integer $e_0$, $1\leq e_0\leq s_0$ such that $pS=\mathfrak{p}^{e_0}$. Let $t_0$ be an integer, $1\leq t_0\leq e_0$ such that $s_0=(m_0-1)e_0+t_0$ and such that $t_0=e_0$ for $m_0=1$. The integers $(p,{m_0},r,e_0,t_0)$ are considered to be the invariants of $S$.
\\\\
Every finite chain ring $S$ with invariants $(p,{m_0},r,e_0,t_0)$ is of the form
$$S\cong G(p^{m_0},r)[x]/(e(x),p^{m_0}x^{t_0}),$$
 where $e(x)=x^{e_0}-p(u_{e_0-1}x^{e_0-1}+\cdots +u_0)$, $u_0\in S^\times$, is an Eisenstein polynomial of degree $e_0$ and $u_0\in G(p^{m_0},r)$ is invertible. Here, the maximal ideal of $S$ is given by $xS$.

 Every $a\in S$ has again a unique representation as
 $$a=\xi_0+x\xi_1+x^2\xi_2+\cdots+ x^{s_0-1}\xi_{{s_0}-1}$$
  for suitable $\xi_i\in T\subset S$, where $T=\{0,\xi_1,\dots,xi_{m_0-1} \}$  is the Teichmueller set of $S$.

Put $T^\times=T\setminus \{0\}$ then  $T^\times$ is a cyclic group of order $p^r-1$ and
$$S^\times =  T^\times (1+xS)\cong T^\times \times (1+xS),\quad |S^\times|=(p^r-1)  p^{r(s_0-1)}.$$
Let us assume  that in our setup $s$ divides $r$, and that $r=\mathbf{m}s$.

\subsection{Skew constacyclic codes over chain rings}
 
 For general finite chain rings $S$, the automorphisms are not always straightforward to write down.
We will assume  that $S_0=G(p^{m_0},s)$  and that $S/S_0$ is a Galois extension of finite chain rings with  cyclic Galois group of  order $\mathbf{m}$ generated by $\sigma$ \cite{A1, A2}. Let $K=S/\mathfrak{P}=\mathbb{F}_{p^r}$ with  $r=\mathbf{m}s$ be the residue field of $S$, and $K_0=S_0/\mathfrak{p}=\mathbb{F}_{p^s}$ the residue field of $S_0$.

 We can count the number of Chen isometry classes of skew constacyclic codes over $S$ for classes of codes  employing Theorem \ref{c:OuazzoufiniteS}. In order to do so, we have to compute  $|S^\times|/|N_{\mathbf{n}}^\sigma(S^\times)|$ and
$|S_0^\times|/|N_{\mathbf{n}}^\sigma(S^\times)|$.

 
 Let $q=p^s$, then the  automorphism $\sigma$ generating the Galois group of $S/S_0$ is defined via
$$\sigma(a)=\xi_0^q+x\xi_1^q+x^2\xi_2^q+\cdots+ x^{s_0-1}\xi_{{s_0}-1}^q;$$
$\sigma$ has  order $\mathbf{m}$ and  $\bar \sigma\in {\rm Gal}(K/K_0)$, $\bar\sigma(y)=y^{q}$.

In order to  compute $N_{i}^\sigma(S^\times)$ we use that every $\alpha\in S^\times$ can be written as $\alpha=zu$ with $z\in T^\times$ and $u\in (1+xS)$. This implies that
$$N_i^{\sigma}(\alpha) = \prod_{j=0}^{i-1}\sigma^j(z)\prod_{j=0}^{i-1}\sigma^j(u)=N_i^{\sigma}(z)N_i^{\sigma}(u),$$
therefore
$$N_{i}^\sigma(S^\times)=N_{i}^\sigma(T^\times)N_{i}^\sigma(1+xS).$$
Since $T^\times\cong (\mathbb{F}_{p^r})^\times$
we can conclude that
 $N_{\mathbf{n}}^\sigma(z)=z^{[\mathbf{n}]_{s}}$ for all $z\in T^\times$.
Since $T^\times$ is cyclic, it is generated by some $\theta$ and the subgroup $N_{\mathbf{n}}^\sigma(T^\times)$  of $N_{\mathbf{n}}^\sigma(S^\times)$
is thus generated by $N_{\mathbf{n}}^\sigma(\theta) = \theta^{[\mathbf{n}]_s}=\theta^w$ with $w=\gcd([\mathbf{n}]_s,p^r-1)$; thus
 $$|N_{\mathbf{n}}^\sigma(T^\times)|= \frac{p^r-1}{w}. $$
 The order of $N_{i}^\sigma(1+xS)$ is harder to establish, but can be computed by ``brute force'' for small rings.

\begin{example} Let $S=GR(4,2)=\mathbb{Z}_4[\omega]/(\omega^2+\omega+1)$, so that $K=\mathbb{F}_4$, and define $\sigma(\omega)=\omega^2$, then $\sigma$ has order two, $S_0=\mathbb{Z}_4,$ $|S_0^{\times}|=2$, $K_0=\mathbb{F}_2$, $|S|=16$, $|S^\times|=12$, $|1+2S|=4$, $T^\times=\{ 1,\omega,\omega^2\}$ and $1+2S=\{1,3,1+2\omega, 3+2\omega \}$. Then
$$
|N_{\mathbf{n}}^\sigma (T^\times)|=
\begin{cases}
     3& \text{if } \mathbf{n} \text{ odd},\\
    1 & \text{if }  \mathbf{n} \text{ even.}
    \end{cases}
 $$
 and a tedious but straightforward calculation shows that
 \[|N_{\mathbf{n}}^\sigma (1+2S)|=
\begin{cases}
4 & \text{if } \mathbf{n}\text{  odd},\\
2 & \text{if } \mathbf{n} \equiv 2 \pmod{4},\\
1 & \text{if } \mathbf{n} \equiv 0 \pmod{4}.
\end{cases}
\]

We obtain
$$
|N_{\mathbf{n}}^\sigma (S^\times)|=
\begin{cases}
12 & \text{if } \mathbf{n} \text{ is odd},\\
2 & \text{if } \mathbf{n} \equiv 2 \pmod{4},\\
1 & \text{if } \mathbf{n} \equiv 0 \pmod{4}
\end{cases}
 $$
 and
 $$
|S^\times/N_{\mathbf{n}}^\sigma (S^\times)|=
\begin{cases}
1 & \text{if } \mathbf{n} \text{ is odd},\\
6 & \text{if } \mathbf{n} \equiv 2 \pmod{4},\\
12 & \text{if } \mathbf{n} \equiv 0 \pmod{4}.
\end{cases}
 $$
Hence the number of distinct Chen isometry classes of families of skew $(\sigma,a)$-constacyclic codes  of length $\mathbf{n}$ arising from nonassociative ambient rings $S[t,\sigma]/S[t,\sigma](t^{\mathbf{n}}-a)$ is
    $$
     N=
     \begin{cases}
1 &\text{if $\mathbf{n}$ is odd,}\\
5 & \text{if } \mathbf{n} \equiv 2 \pmod{4},\\
10 & \text{if } \mathbf{n} \equiv 0 \pmod{4}.
\end{cases}
    $$
    There is one additional Chen equivalence class (and thus at most one additional Chen isometry class) of a family of skew $(\sigma,a)$-constacyclic codes  of length $\mathbf{n}$ arising from associative ambient rings $S[t,\sigma]/S[t,\sigma](t^{\mathbf{n}}-a)$ (i.e., $\mathbf{n}$ is odd and $a\in S_0^\times$) (Theorem \ref{c:OuazzoufiniteS}).
\end{example}

\begin{example}
Let $S=G(p^{m_0},r)$, $S_0=G(p^{m_0},s)$, $r=\mathbf{m}s$, and suppose that $\gcd([\mathbf{n}]_s,p^r-1)=1$.  Then $|N_{\mathbf{n}}^\sigma(T^\times)|= p^r-1, $ that means $N_{\mathbf{n}}^\sigma(T^\times)=T^\times$ is largest possible. Now all depends on the order of the $p$-group $N_{\mathbf{n}}^\sigma(1+pS).$

Since $K=\mathbb{F}_{p^r}$, $\bar \sigma(x)=x^{p^s}$ with $s|r$  and $K_0=\mathbb{F}_{p^s}$, we know that
 $N_{\mathbf{n}}^{\bar \sigma}(K^\times)=K^\times$ so  all Petit ambient rings  $K[t;\bar \sigma]/K[t;\bar \sigma](t^{\mathbf{n}}-\bar a)$ are  $\mathbf{n}$-Chen equivalent, as was already observed in \cite[Corollary 2]{Oua2025} and \cite[Proposition 1]{BoulanouarBatoulBoucher2021}. Thus all skew constacyclic codes  of length $\mathbf{n}$ over $K$ are $\mathbf{n}$-Chen equivalent, but not all over $S$ need to be.

 If  $N_{\mathbf{n}}^\sigma(1+pS)=1+pS$ 
 then  all $S[t;\sigma]/S[t;\sigma](t^{\mathbf{n}}-a)$ are $\mathbf{n}$-Chen equivalent as well, and  in particular all  are $\mathbf{n}$-Chen equivalent to the ambient ring $S[t;\sigma]/S[t;\sigma](t^{\mathbf{n}}-1)$, which corresponds to the class of  skew cyclic codes.  Moreover, the notions of $\mathbf{n}$-Chen equivalence and $\mathbf{n}$-equivalence coincide in this case just like for the skew constacyclic codes  of length $\mathbf{n}$ over $K$. When  $\mathbf{m}\nmid \mathbf{n}$, then the notions of $\mathbf{n}$-equivalence,  $\mathbf{n}$-isometry and  $\mathbf{n}$-isometry coincide as well, again like for the skew constacyclic codes of length $\mathbf{n}$ over $K$ (Theorem \ref{thm:chainmain}).
Hence in this special case, for all $a\in S^\times$, the skew $(\sigma,a)$-constacyclic codes  of length $\mathbf{n}$ over $S$ are equivalent to the skew cyclic codes  of length $\mathbf{n}$ \cite[Example 5.7]{NevPum2025}.

However, if the order of $N_{\mathbf{n}}^\sigma(1+pS)$ is strictly smaller than the order $p^{r(m_0-1)}$ of $1+pS$, we will find skew $(\sigma,a)$-constacyclic codes  of length $\mathbf{n}$ over $S$ which are not equivalent. This demonstrates one of the advantages of choosing chain rings over finite fields.
\end{example}

\begin{example}\label{e:important}
At the other extreme, when
$\gcd([\mathbf{n}]_s,p^r-1)=p^r-1$, then both $N_{\mathbf{n}}^\sigma(T^\times)=\{1\}$ and $N_{\mathbf{n}}^\sigma(K^\times)=\{1\}$. The latter implies that no two classes of skew $(\bar \sigma,\bar a)$-constacyclic codes and skew $(\bar\sigma,\bar b)$-constacyclic codes  with distinct $\bar a,\bar b\in K^\times$ will be  $\mathbf{n}$-Chen equivalent. Therefore for all $a,b\in S^\times$ such that $\bar a\not=\bar b$ in $K$, the ambient rings $S[t;\sigma]/S[t;\sigma](t^{\mathbf{n}}-a)$ and $S[t;\sigma]/S[t;\sigma](t^{\mathbf{n}}-a)$ are not $\mathbf{n}$-Chen equivalent, either. We conclude that  no two classes of skew $(\sigma, a)$-constacyclic codes and skew $(\sigma, b)$-constacyclic codes  with distinct $\bar a,\bar b\in K^\times$ will be  $\mathbf{n}$-Chen equivalent.

When we work with Galois extensions of Galois rings where $S=G(p^{m_0},sn)$  and $S_0=G(p^{m_0},s)$, $r=s \mathbf{m}$, we  have
$r$ choices for $\tau\in {\rm Gal}(G(p^{m_0},sn)/G(p^{m_0},1))$, hence the set $\{\bar a^{p^{v}}\mid 0\leq v < \mathbf{m}\}\subset {\rm Gal}(G(p^{m_0},s \mathbf{m})$ is an equivalence class with $r$ elements, so all of the corresponding $( \sigma, a^{p^v})$-constacyclic codes are $\mathbf{n}$-equivalent. Thus there are fewer $\mathbf{n}$-equivalence classes of skew constacyclic codes over $S$ than  $\mathbf{n}$-Chen equivalence classes; the notion of equivalence again gives a tighter classification.

The order of the subgroup $|N_{\mathbf{n}}^\sigma(1+xS)|$ determines the exact number of (Chen) equivalence classes of skew constacyclic codes over $S$ of length $\mathbf{n}$. If $N_{\mathbf{n}}^\sigma(1+pS)\not=\{1\}$ then the above do not represent all possible skew $(\sigma,a)-$constacyclic codes, as then there exist $a\in S^\times \setminus T^\times$.
\end{example}

 The abelian group $N_{\mathbf{n}}^\sigma(1+xS)$ seems difficult to compute in general, but we can get some immediate estimates on its size.

\begin{enumerate}
\item \label{1} Since $N_{\mathbf{n}}^\sigma(1+xS)$ is a subgroup of $1+xS$, we know that $N_{\mathbf{n}}^\sigma(1+xS)$ is a $p$-group and
  $$|N_{\mathbf{n}}^\sigma(1+xS)|\text{ divides } p^{r(s_0-1)}.$$

  \item \label{2} When $\mathbf{m}\mid \mathbf{n}$ we know that $N_{\mathbf{n}}^\sigma(1+xS)$ is a subgroup of $S_0^\times$, so
  $$|N_{\mathbf{n}}^\sigma(1+xS)|\text{ divides } p^{s(m_0-1)}.$$

   \item \label{3} We assume that $S/S_0$ is Galois, hence separable, so that $xS=xS_0$ and therefore $1+xS_0$ is well-defined and $N_{\mathbf{n}}^\sigma(1+xS_0)$ is a subgroup of $1+xS_0$, thus
  $$|N_{\mathbf{n}}^\sigma(1+xS_0)|\text{ divides } p^{s(m_0-1)}.$$
       \end{enumerate}

    \subsection{The group structure of $N_{\mathbf{n}}^\sigma(1+xS)$}
    We can refine our observations (\ref{1}), (\ref{2}), (\ref{3})  employing the well-known Fundamental Theorem for finite abelian groups, which shows how the group structure of $S^\times$ influences the group structure of $N_i^\sigma(1+xS)$ and thus of $S^\times/N_i^\sigma(S^\times)$ for all positive integers $i$.

   \begin{enumerate}\setcounter{enumi}{3}
  \item By the Fundamental Theorem for finite abelian groups, we know that there exist a unique partition $(d_1,\dots,d_u)$ of $p^{r(s_0-1)}$
   with positive integers $d_j$,  $d_1\leq d_2\leq\dots \leq d^u$, where
    $$\sum_{j=0}^u d_j={r(s_0-1)},$$
     such that there is a group isomorphism
      $$\Psi: 1+xS \to \bigoplus_{j=0}^{u}\mathbb{Z}_{p^{d_j}}$$
      which canonically extends to a group isomorphism
      $$\Psi:S^\times \to \bigoplus_{j=0}^{u}\mathbb{Z}_{p^{d_j}}\oplus \mathbb{Z}_{p^r-1}.$$
       Now $\Psi$ canonically induces a group isomorphism
      $\Psi:N_i^\sigma(1+xS)\to \Psi(N_i^\sigma(1+xS))$, and so
  $$\Psi(N_i^\sigma(1+xS)) \text{ is a subgroup of } \bigoplus_{j=0}^{u}\mathbb{Z}_{p^{d_j}}.$$

  \item  Since $S_0=G(p^{m_0},s)$ there also exists a unique partition $(e_1,\dots,e_v)$ of $p^{r(m_0-1)}$ with positive integers $e_i$, $e_1\leq e_2\leq\dots \leq e^v$,
      $$\sum_{i=0}^v e_i={r(m_0-1)},$$ 
       such that  there is a group isomorphism
      $$\Psi: 1+xS_0 \to \bigoplus_{j=0}^{v}\mathbb{Z}_{p^{e_j}}$$
      which canonically extends to a group isomorphism
      $$\Psi:S_0^\times \to \bigoplus_{j=0}^{v}\mathbb{Z}_{p^{e_j}}\oplus \mathbb{Z}_{p^s-1}.$$
When $ \mathbf{m}\mid i$ we know that $N_{i}^\sigma(1+xS)$ is a subgroup of $S_0^\times$, so
  $$N_{i}^\sigma(1+xS)\text{ is isomorphic to a subgroup of } \bigoplus_{j=0}^{v}\mathbb{Z}_{p^{e_j}}\oplus \mathbb{Z}_{p^s-1}.$$
\end{enumerate}

Indeed we  know more:

\begin{lemma}
(i) If $p$ is odd or if $p=2$ and $m_0\leq 2$ then
$$GR(p^{m_0}, r)^\times \cong \mathbb{Z}_{p^{m_0-1}}^r\oplus \mathbb{Z}_{p^r-1}$$
and
$$N_{i}^\sigma(1+p \,GR(p^{m_0}, r))\text{ is isomorphic to a subgroup of } \mathbb{Z}_{p^{m_0-1}}^r.$$
(ii) If  $m_0\geq 3$ then
$$GR(2^{m_0}, r)^\times \cong  \mathbb{Z}_2 \oplus  \mathbb{Z}_{2^{m_0-2}} \oplus  \mathbb{Z}_{2^{m_0-1}}^{r-1}\oplus \mathbb{Z}_{p^r-1}$$
and
$$N_{i}^\sigma(1+2\, GR(2^{m_0}, r))\text{ is isomorphic to a subgroup of } \mathbb{Z}_2 \oplus  \mathbb{Z}_{2^{m_0-2}} \oplus  \mathbb{Z}_{2^{m_0-1}}^{r-1}.$$
\end{lemma}

\begin{proof}
By \cite[Theorem XVI.9 (b)]{Mc}, we know that
$$1+p \,GR(p^{m_0}, r)\cong \mathbb{Z}_{p^{m_0-1}}^r$$
 if $p$ is odd or if $p=2$ and $m_0\leq 2$
 and that
$$1+2\, GR(2^{m_0}, r)\cong \mathbb{Z}_2 \oplus   \mathbb{Z}_{2^{m_0-2}} \oplus   \mathbb{Z}_{2^{m_0-1}}^{r-1},$$
if $p=2$ and $m\geq 3$.
In particular, this means that
$S^\times \cong \mathbb{Z}_{p^{m_0-1}}^r\oplus \mathbb{Z}_{p^r-1}$
if $p$ is odd or if $p=2$ and $m_0\leq 2$,
and that
$S^\times \cong\mathbb{Z}_2 \oplus   \mathbb{Z}_{2^{m_0-2}} \oplus    \mathbb{Z}_{2^{m_0-1}}^{r-1}\oplus \mathbb{Z}_{p^r-1}$
if $p=2$ and $m\geq 3$.
\end{proof}

Similar but much more complex results about the structure of $1+xS$ are known also for general commutative finite chain rings $S$ \cite[Proposition 2.7, Corollary 3.2]{HouLM2003} and  \cite[Equation (29)]{CON2025}. These were used in \cite{CON2025} to compute the Chen isometry and Chen equivalence classes of constacyclic codes over $S$.

When we have $ \mathbf{m}\mid i$ for the  order $\mathbf{m}$ of $\sigma$, we can say more.

\begin{corollary}
Let $S=G(p^{m_0},s \mathbf{m})$  and $S_0=G(p^{m_0},s)$, $r=s \mathbf{m}$.
Suppose that  $ \mathbf{m}\mid i$.
  \\ (i) If $p$ is odd or if $p=2$ and $m_0\leq 2$ then
$$N_{i}^\sigma(1+p \,GR(p^{m_0}, r))\text{ is isomorphic to a subgroup of } \mathbb{Z}_{p^{m_0-1}}^s.$$
In particular, if $s=1$ (i.e., $S_0=\mathbb{Z}_{p^{m_0}}$) then $N_{i}^\sigma(1+p \,GR(p^{m_0}, r))$ is a cyclic group.
\\ (ii) If  $m_0\geq 3$ then
$$N_{i}^\sigma(1+2\, GR(2^{m_0}, r))\text{ is isomorphic to a subgroup of } \mathbb{Z}_2 \oplus  \mathbb{Z}_{2^{m_0-2}} \oplus  \mathbb{Z}_{2^{m_0-1}}^{s-1}.$$
\end{corollary}

\subsection{Conclusions and further work}

We initiated the study of skew polycyclic codes up to equivalence, identified skew polycyclic codes which are the same in terms of protecting the underlying information, and obtained some catalogues of skew constacyclic codes up to isometries and equivalences with tools for de-duplicating them, which should help to identify gaps on which further research may be focused. Our equivalence relations are based on the algebra structure of the mambient algebras, and the most general and canonical ones possible from this point of view. They will yield tighter classifications of isometric and equivalent codes that previously used ones.
We also now have an easy way to see how the generator skew polynomials of two isometric codes are related, and how to compute the dimensions of codes that are isometric under the isometry $G_{\tau,\alpha,k}$ of their ambient rings.

Our next goals are to understand  isometries better, which we have started to do in a forthcoming joint paper with Monica Nevins,  and to then look at special rings $S$, like finite chain rings or Galois rings in much more detail.
The isometry classes for skew polycyclic codes seem to be hard to find, however, the special case of skew constacyclic codes seems tractable.
Indeed we conjecture that the only isometries between nonassociative ambient rings are the equivalences $G_{\tau,\alpha}$.

For finite fields, it is well known that $\F_q[t;\sigma,\delta_\beta]\cong \F_q[t;\sigma]$. However, this isomorphism of noncommutative rings changes the Hamming weight of corresponding skew $(\sigma,\delta)$-polycyclic codes, and hence their performance.
It was pointed out already in \cite{BouU2014} that even when $S$ is a finite field, codes constructed with some $f\in  S[t;\sigma,\delta_\beta]$ where $\delta_\beta(a)=\beta(\sigma(a)-a)$ is nonzero, can have a better minimum Hamming distance than those constructed when $\delta=0$.
 Most of the current literature, however, only considers the case that $S$ is a finite field or Galois ring, and $\delta=0$, e.g. \cite{B}.

We plan to investigate the Hamming weight preserving isomorphisms, in particular $G_{\tau,\alpha}$, also when $\delta\not=0$.
We also plan to look at the case where $S$ is not commutative.

\subsection*{Acknowledgments} This paper was written during the second author's stay as a Simons Professor in Residence at the University of Ottawa. She gratefully acknowledges the support of CRM and the Simons Foundation. She thanks the Department of Mathematics and Statistics for its hospitality and its congenial and inspiring atmosphere, and especially Monica Nevins for lots of inspiring conversations.


\providecommand{\bysame}{\leavevmode\hbox to3em{\hrulefill}\thinspace}
\providecommand{\MR}{\relax\ifhmode\unskip\space\fi MR }
\providecommand{\MRhref}[2]{%
  \href{http://www.ams.org/mathscinet-getitem?mr=#1}{#2}
}
\providecommand{\href}[2]{#2}

\end{document}